\documentclass[11pt, letterpaper, onecolumn, accepted=2026-01-05]{quantumarticle} 
\pdfoutput=1

\usepackage[utf8]{inputenc}
\usepackage{enumerate}

\usepackage{amsthm}
\usepackage{amsmath}
\usepackage{amsfonts}
\usepackage[mathscr]{eucal}
\usepackage{amssymb}
\usepackage{bbm} 
\usepackage{xcolor}
\usepackage{tikz-cd}
\usepackage{quiver}
\usepackage[T1]{fontenc}
\usepackage{xurl}
\usepackage[style=alphabetic, maxbibnames=99, url=false, eprint=false, isbn=false, backend=bibtex]{biblatex}
\definecolor{myBlue}{RGB}{25,132,232}
\definecolor{myGreen}{RGB}{164,180,43}
\definecolor{myPink}{RGB}{208,12,133}
\definecolor{myYellow}{RGB}{255,160,0}
\definecolor{myBrown}{RGB}{173,100,82}
\definecolor{myOrange}{RGB}{245,124,0}
\definecolor{myPurple}{RGB}{178,121,219}
\definecolor{myPlum}{RGB}{123,31,162}
\definecolor{myTeal}{RGB}{0,151,167}

\usepackage[breaklinks]{hyperref}
\usepackage{graphicx}
\hypersetup{
    pdftitle={Quantum games and synchronicity},
    pdfauthor={Adina Goldberg}, 
    pdfsubject={Quantum article},
    pdfkeywords={nonlocal game} {quantum game} {quantum information} {graphical calculus} {synchronicity} {correlation} {strategy} {quantum graph} {quantum set} {quantum function} {synchronous game} {categorical quantum mechanics} {quantum isomorphism} {quantum homomorphism} {isomorphism game} {homomorphism game},
    colorlinks=true,
    linkcolor=myPink!80!myPlum,       
    citecolor=myTeal,        
    filecolor=myBlue,      
    urlcolor=myPlum           
}
\allowdisplaybreaks

\newtheorem{thm}{Theorem}[section] 

\newtheorem{cor}[thm]{Corollary}
\newtheorem{lemma}[thm]{Lemma}
\newtheorem{prop}[thm]{Proposition}

\theoremstyle{definition}
\newtheorem{defn}[thm]{Definition}
\newtheorem{remark}[thm]{Remark}
\newtheorem{exam}[thm]{Example}

\newcommand{\bb}[1]{\mathbb{#1}}
\newcommand{\cl}[1]{\mathcal{#1}}
\newcommand{\ket}[1]{|#1\rangle}
\newcommand{\bra}[1]{\langle#1|}

\newcommand{\op}[1]{#1^\mathrm{op}}
\newcommand{\braket}[2]{\left\langle #1 |#2 \right\rangle}

\counterwithout{figure}{section}
\numberwithin{table}{subsection}

\input{./style/my_tikzit.sty}
\input{./style/nlgames.sty}
\input{./style/tikz_adjustments.sty}

\title[Quantum Games and Synchronicity]{Quantum games and synchronicity}
\author{Adina Goldberg}
\affiliation{University of Waterloo}
\affiliation{Institute for Quantum Computing}
\homepage{https://sites.google.com/view/mathwithadina}
\thanks{was supported by the Natural Sciences and Engineering Research
Council of Canada (NSERC) through the Postgraduate Scholarship - Doctoral (PGS D) program.}
\orcid{0000-0002-0909-8791}
\date{January 8, 2026}

\begin{document}

\begin{abstract}
\noindent In the flavour of categorical quantum mechanics, we extend nonlocal games to allow quantum questions and answers, using quantum sets (special symmetric dagger Frobenius algebras) and the quantum functions of \cite{MRV}. Equations are presented using a diagrammatic calculus for tensor categories. To this quantum question and answer setting, we extend the standard definitions, including strategies, correlations, and synchronicity, and we use these definitions to extend results about synchronicity. We extend the graph homomorphism (isomorphism) game to quantum graphs, and show it is synchronous (bisynchronous) and connect its perfect (bi)strategies to quantum graph homomorphisms (isomorphisms). Our extended definitions agree with the existing quantum games literature, except in the case of synchronicity.
\end{abstract}

\maketitle

\tableofcontents
\addtocontents{toc}{\setcounter{tocdepth}{1}}

\newpage
\section*{Introduction}


Synchronicity for nonlocal games, introduced in \cite{PSSTW16}, is a well studied property in the typical setting, where questions and answers are classical. In that setting, synchronicity forces the players to have matching strategies. In particular, strategies corresponding to synchronous correlations are forced to be shared functions or operators (depending on the strategy class) that both players apply to the received question to produce an answer. Furthermore, if a game is synchronous, the associated perfect correlations are forced to be synchronous themselves.

We extend the above ideas to \emph{quantum games} -- nonlocal games with quantum questions and answers -- using string diagrams interpreted in the category of finite dimensional Hilbert spaces. This fits into the program of categorical quantum mechanics (CQM), initiated in \cite{AC04}, which approaches quantum theory as a process theory, modelling physical systems and transformations as the objects and morphisms in a category equipped with some additional structure. Our quantum games can be thought of as bipartite instances of the quantum relations of \cite{Wea10, Kor20}, but it is thinking of them as games that invites extensions of the classical literature.

We upgrade many notions from nonlocal games to the quantum game setting. Most of these notions have been extended to the quantum setting in the literature, but often in less generality, and never using a categorical or graphical approach. In particular, our extensions of synchronicity and the graph isomorphism game appear to be new.

For a detailed but elementary introduction to nonlocal games, strategies, synchronicity, and some examples, see the author's PhD thesis \cite[Chapter 2]{myThesis}. For a one page operational introduction, \cite[Section 1.1]{myThesis} will suffice. For a more in depth survey of nonlocal games as they relate to operator theory, see \cite{PV16}.

\subsection*{Other approaches to quantum games}

 The author is aware of three other approaches to general purpose quantum-to-quantum games that appear mutually distinct. Making precise connections between these approaches would be a significant undertaking and is outside the scope of this paper.

The quantum games defined here, viewed from an operational perspective, are similar to the quantum games of \cite{LTW08}, who give an operational definition. There are further definitions along the same lines, referred to as \emph{rank-one} \cite{CJPP11} and \emph{finite-rank} \cite{CLTT23} projection games.

A more recent approach to defining quantum games  uses \emph{hypergraphs}. See \cite{HT22, CLTT23, HT25}. The author only became aware of these definitions in the late stages of this work. There appears to be a strong connection between the categorical approach here and the hypergraph approach.

Another recent approach is the \emph{projection lattice} quantum games defined in \cite{TT24}, whose strategies were studied earlier in \cite{BHTT21, BHTT23}, along with an extension of typical synchronicity to the quantum setting called \emph{concurrency}. Concurrency is discussed towards the end of this paper, and shown to be distinct from but related to synchronicity.

Other authors have defined partial extensions of nonlocal games to the quantum setting \cite{Bus11, fritz12, TFKW13, JMRW16}.

Finally, there have been quantisations of specific classes of nonlocal games, including quantum XOR games \cite{RV15} and a quantum-to-classical \cite{BGH22} and quantum-to-quantum graph homomorphism game \cite{TT24}.

For a more detailed discussion of quantum games in the literature, and how they may relate to the quantum games presented here, please refer to the author's PhD thesis \cite[Section 4.7]{myThesis}. It would be useful to analyse the connections between these recent, varied approaches.

\subsection*{Overview}

To start Section \ref{sec:quantum_games}, we use string diagrams to talk about quantum sets and quantum functions, building on the work of \cite{MRV}. We define three types of linear maps between quantum sets that allow us to generalize the same notions from nonlocal games: \emph{Games} (self-conjugate maps satisfying a diagrammatic idempotency condition), \emph{strategies} (quantum functions), and \emph{correlations} (completely positive counital maps). These correlations are the quantum correlations of \cite{BKS21}. We go on to define \emph{nonsignalling, quantum commuting, quantum tensor}, and \emph{deterministic} correlations, generalizing these notions as defined for typical nonlocal games. We briefly mention the tensor product and the \emph{commuting} product of games.

We are then equipped to define \emph{synchronicity} in Section \ref{sec:sync}. This definition is distinct from that of \cite{BKS21}, and builds off something that we call \emph{sharing}. We check that our notion of synchronicity extends typical synchronicity. We also use sharing to graphically represent the \emph{classical dimension} of a quantum set.

In Section \ref{sec:sync_cons}, we extend properties of typical synchronicity. One of our main results is that perfect correlations/strategies for synchronous games are forced to be synchronous (Theorem \ref{thm:perf_sync}). We examine the structure of synchronous strategies and correlations. We show that Alice's operators determine Bob's operators when implementing a synchronous quantum commuting correlation (Theorem \ref{thm:sync_qc_corr}) or strategy (Theorem \ref{thm:sync_qc_strat}). As a key tool, we present a diagrammatic version of the Cauchy-Schwarz inequality for quantum functions (Theorem \ref{thm:cs_ineq_qfunc}).

In Section \ref{sec:bisync}, we discuss \emph{bistrategies}, \emph{bicorrelations}, and \emph{bisynchronicity}. We show (Theorem \ref{thm:bisync_dims_equal}) that bisynchronicity of increasingly structured maps preserves a growing list of quantum set properties: dimension, classical dimension, quantum isomorphism class, isomorphism class. We also describe the structure of bistrategies in different strategy classes. 

We then introduce examples, in Section \ref{sec:graph_games}, of games 
played on quantum graphs. First we study a truly quantum, synchronous game (the quantum graph homomorphism game), generalizing the well-known graph homomorphism game of \cite{MR16}.  We connect its perfect correlations to quantum graph homomorphisms (Theorems \ref{thm:hom_to_perf_qt_corr}, \ref{thm:perf_det_to_hom}). Second, we discuss the bisynchronous quantum graph isomorphism game, and connect its perfect bicorrelations to quantum graph isomorphisms (Theorems \ref{thm:iso_to_perf_qt_bicor}, \ref{thm:perf_det_to_iso}).

Finally, in Section \ref{sec:other_approaches}, we compare our definitions to two other notions of synchronicity in the literature that also aim to generalize typical synchronicity. One of these generalizations is the \emph{concurrency} property of \cite{BHTT21, BHTT23}, and the other is the synchronicity property of \cite{BKS21}, defined without the context of games, in the study of quantum correlations between quantum spaces.

\subsection*{Notes to the reader}
Sometimes equals signs are underset with numbers, followed by a numbered list of justifications. Other times, equals signs are underset with a hyperlinked number in parentheses, referring to a numbered equation earlier in the paper.

The reader should also know that an earlier version of this paper, but with additional material, constitutes Chapters 4-6 of the author's PhD thesis \cite{myThesis}. The material omitted here introduces values of quantum games, comments on other approaches to quantum games,
speculates about generalizing strategies to channels, and takes steps towards operational interpretations of this quantum game formalism.

\section{Quantum games via diagrams}
\label{sec:quantum_games}

To define games, strategies and correlations in the quantum question and answer setting, we follow the program of categorical quantum mechanics, using the language of monoidal dagger categories, and the accompanying graphical calculus. We apply this approach to the category of finite dimensional complex Hilbert spaces, where the dagger is the Hilbert space adjoint. There is an excellent exposition of this topic in \cite{catsQT}. We follow the conventions introduced there.

\subsection*{Quantum sets and functions}
In this section, we generalize the setting of typical nonlocal games to quantum question and answer games. Instead of question and answer spaces being finite sets, they will now be some kind of ``quantum sets". Quantum sets are introduced in \cite{MRV} (and are similar in spirit to the quantum sets of \cite{Kor20}). We include their definition here for readability.

In the following, let $\bf{FHilb}$ be the category of finite dimensional complex Hilbert spaces. In this category, the dagger ($\dagger$) is the Hilbert space adjoint.

\begin{defn}[Dagger Frobenius algebra]
\label{defn:frob}
A \emph{dagger Frobenius algebra} in $\bf{FHilb}$ is an object $X$ of $\bf{FHilb}$ that has an associative unital multiplication $m:X\otimes X \to X$ with unit $u:\bb{C}\to X$, and a coassociative counital comultiplication $\Delta:X \to X\otimes X$ with counit $\varepsilon:X\to\bb{C}$. These maps must satisfy $\Delta = m^\dagger, \varepsilon = u^\dagger$, and
\begin{equation}
\label{eq:frob}(1 \otimes m)(\Delta \otimes 1) = \Delta m = (m\otimes 1)(1\otimes \Delta)
\end{equation}
\end{defn}
Equation \eqref{eq:frob} is called the Frobenius equation. These dagger Frobenius algebras are an alternate presentation of finite-dimensional unital C$^*$-algebras \cite{Vic10}. Their underlying algebras are direct sums of finite-dimensional complex matrix algebras, with additional structure of a choice of inner product. The inner product determines $\varepsilon$ and $\Delta$.

\begin{defn}[Quantum set, \cite{MRV}]
\label{defn:qset}
A dagger Frobenius algebra $X$ in {\bf FHilb} is called a \emph{quantum set} if it is \emph{special}, that is
\begin{equation}
\label{eq:special}
m\Delta = \textrm{Id}_X,
\end{equation}
and \emph{symmetric}, meaning
\begin{equation}
\label{eq:symmetric}
\varepsilon m \sigma = \varepsilon m,
\end{equation}
where $\sigma$ is the canonical isomorphism that swaps tensor factors.  
\end{defn}

From an operator algebraic perspective, we can think of \eqref{eq:special} as requiring that $\Delta$ is an isometry, and of \eqref{eq:symmetric} as requiring that $\varepsilon$ is tracial.

Quantum sets are referred to most everywhere else as $\dagger$-SSFAs (for \textbf special \textbf symmetric dagger \textbf Frobenius \textbf algebras) in {\bf FHilb}, and when they are classical, as $\dagger$-SCFAs (where the \textbf c stands for \textbf commutative). We use the terminology quantum set here because it is more evocative, given that we are generalizing nonlocal games with classical question and answer sets.

If a dagger Frobenius algebra is only symmetric, but not special, we will refer to it as a $\dagger$-SFA (and if commutative, as a $\dagger$-CFA).
 
\begin{remark}[Structure of quantum sets]\label{rem:qset_structure}
As laid out in \cite[Corollary 5.34]{catsQT}, all quantum sets are specially weighted direct sums of complex matrix algebras. More precisely, we mean that the underlying set of a quantum set $A$ is a finite direct sum of full matrix algebras $\bigoplus_k M_{n_k}(\mathbb C)$. The multiplication $m$ is given by the usual matrix multiplication, but divided by $\sqrt{n_k}$ on each block, so that given $a,b \in A$ with $a = \bigoplus_k a_k, b = \bigoplus_k b_k$, $$m(a\otimes b) = \bigoplus_k \frac{1}{\sqrt{n_k}}a_kb_k.$$
 The unit $u$ is given by the usual identity matrix, but scaled in each block by $\sqrt{n_k}$, so that
 $$u(1) = \bigoplus_k \sqrt{n_k} I_{n_k}.$$ Taking the dagger to be the Hermitian adjoint (conjugate-transpose) operator, the compatible Hilbert space inner product that arises is $\langle a, b\rangle = \mathrm{Tr}(a b^\dagger)$. Under this inner product, the standard matrix-unit basis $\{E^k_{ij}: 1\leq i,j \leq n_k\}$ is orthonormal, where $E^k_{ij}\in A$ has a $1$ in the $(i,j)$ entry of block $k$ and $0$s elsewhere. The comultiplication is then given on this basis by 
 $$\Delta(E^k_{ij}) = \frac{1}{\sqrt{n_k}}\sum_t E^k_{it} \otimes E^k_{tj}$$
 and the counit is given by
 $$\varepsilon(a) = \sum_k\sqrt{n_k}\mathrm{Tr}(a_k).$$
 \end{remark}\

The structure result above illustrates a sliding-scale of ``quantumness''. When all matrix algebras are size $n_k=1$, a quantum set is a finite direct sum of copies of $\mathbb{C}$, and the above multiplication becomes commutative. We call the resulting quantum set \emph{classical}. On the other end of the spectrum, when there is only one direct summand, so the quantum set is a full matrix algebra, we call it \emph{maximally quantum}.

\begin{exam}[Classical set]\label{exam:classical}
Because we will refer to classical sets often, let's specify  that we will often implicitly identify the set element $i\in [n]$ with the standard basis vector $\ket i \in \mathbb{C}^n$. The operations in Remark \ref{rem:qset_structure} become
\begin{align*}
   m(\ket i\otimes \ket j) &= \delta_{ij} \ket i, \\
   u(1) &= \sum_i \ket i,\\
   \Delta(\ket i) &= \ket i\otimes \ket i, \\
   \varepsilon(\ket i) &= 1. 
\end{align*}

Here multiplication can be thought of as \emph{merging} of classical elements, the unit is (an unnormalized copy of) the maximally mixed state, comultiplication can be thought of as \emph{copying} classical elements, and the counit has a \emph{deletion} effect on classical elements.

Classical sets are also commonly identified with diagonal matrices under $\ket i \mapsto E_{ii}$. In that case the above operations give us back matrix multiplication (or the Schur product, as they coincide here). 
\end{exam}

\begin{exam}[Maximally quantum set]\label{exam:max_quant}
In the case that the quantum set is a full matrix algebra $M_n$, the operations of Remark \ref{rem:qset_structure} become
\begin{align*}
   m(a \otimes b) &= \frac{1}{\sqrt n} ab, \\
   u(1) &= \sqrt{n} I_n,\\
   \Delta(E_{ij}) &= \frac{1}{\sqrt n} \sum_t E_{it} \otimes E_{tj}, \\
   \varepsilon(a) &= \sqrt n \mathrm{Tr}(a). 
\end{align*}
    
\end{exam}

It is convenient and elegant to use a graphical calculus to work with Frobenius algebras. Symbolic equations only possess one dimension of latitude in which to perform two types of composition: tensor products, which in physics represent composing across space, and composition of maps, physically representing composing in time. This results in the eyes jumping back and forth matching $n^\textrm{th}$ tensor factors when performing compositions.
Using a two dimensional graphical calculus allows us to depict tensor products horizontally and composition of maps vertically, freeing the brain from this extra matching step.

From this point onwards, equations will be presented diagrammatically. Using the graphical calculus given in \cite{catsQT}, we encode Frobenius algebras as follows. The diagrams are read from bottom to top, as processes evolving in time, and left to right for tensor products. Each string or wire represents the Frobenius structure $X$, and the structural maps are represented by nodes with the appropriate number of $X$ inputs and outputs. Empty space represents $\bb{C}$ (the tensor unit of $\bf{FHilb}$).

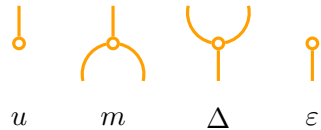
\ctikzfig{frob\_ops}

We will often omit the small circles for $m$ and $\Delta$, as the multiplication or comultiplication can be inferred just from the joining or splitting of strings. For convenience, we will also draw caps and cups to represent the common compositions $\varepsilon m$ and $\Delta u$.

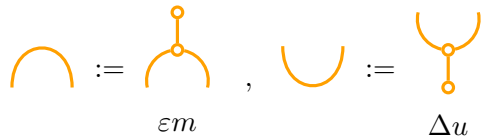
\ctikzfig{cups\_caps}

These cups and caps define a duality of $X$ to itself, resulting in the ``snake equations''. This is the chosen duality for quantum sets in our work.

\ctikzfig{snake\_eqns}

Given quantum sets $X,A$ if $f:X\to A$ is a linear map, then we represent $f$ as a node transforming an $X$-wire to an $A$-wire.

Taking the dagger of a diagram is done by flipping the whole diagram upside down, and applying daggers to any maps, which corresponds here to the Hilbert space adjoint operation.

The dagger of the map $f$ is denoted $f^\dagger$. The categorical dual (transpose) of the map $f$ is denoted $f^*$. The categorical conjugate of $f$, denoted $f_*$, is defined as the composition of dagger and transpose (which can be shown to commute).

When working with an entire diagram and taking conjugates and transposes, a horizontal flip of each element of a diagram corresponds to its categorical conjugate, and a half-turn (rotation by $\pi$) of the diagram while keeping inputs and outputs fixed corresponds to its transpose.

In accordance with the above graphical rule, it is sometimes convenient, when working with daggers, transposes, and conjugates, to use the following trapezium notation for maps.
\ctikzfig{trap_notation}

We have been working in the (self-dual) quantum set setting. In general, if there is no quantum set structure present, i.e., if $f:\cl H\to \cl K$, then letting $\cl H^*$ be the Hilbert space dual of $\cl H$,
\begin{align*}
    f^\dagger&: \cl K \to \cl H,\\
    f_*&: \cl H^* \to \cl K^*,\\
    f^*&: \cl K^* \to \cl H^*
\end{align*}

We can now graphically represent the equations defining quantum sets. 

\begin{center}
\begin{tabular}{cc}
\multicolumn{2}{c}{\begin{tikzpicture}[scale=0.75]
	\begin{pgfonlayer}{nodelayer}
		\node [style=none] (6) at (0.5, 1) {};
		\node [style=none] (11) at (1.75, 1) {};
		\node [style=none] (12) at (3.25, 1) {};
		\node [style=none] (15) at (2.25, -1) {};
		\node [style=none] (16) at (2.25, -0.5) {};
		\node [style=none] (17) at (3.75, -1) {};
		\node [style=none] (18) at (3.25, 0.5) {};
		\node [style=none] (19) at (-1.75, 1) {};
		\node [style=none] (20) at (-3.25, 1) {};
		\node [style=none] (21) at (-2.25, -1) {};
		\node [style=none] (22) at (-2.25, -0.5) {};
		\node [style=none] (23) at (-3.75, -1) {};
		\node [style=none] (24) at (-3.25, 0.5) {};
		\node [style=none] (25) at (-1.125, 0) {$=$};
		\node [style=none] (26) at (1.125, 0) {$=$};
		\node [style=none] (27) at (-0.5, 1) {};
		\node [style=none] (28) at (0, 0.25) {};
		\node [style=none] (29) at (-0.5, -1) {};
		\node [style=none] (31) at (0.5, -1) {};
		\node [style=none] (32) at (0, -0.25) {};
	\end{pgfonlayer}
	\begin{pgfonlayer}{edgelayer}
		\draw [style=x-wire, in=-90, out=-180] (16.center) to (11.center);
		\draw [style=x-wire] (16.center) to (15.center);
		\draw [style=x-wire, in=90, out=0] (18.center) to (17.center);
		\draw [style=x-wire, in=0, out=180, looseness=1.50] (18.center) to (16.center);
		\draw [style=x-wire] (18.center) to (12.center);
		\draw [style=x-wire, in=-90, out=0] (22.center) to (19.center);
		\draw [style=x-wire] (22.center) to (21.center);
		\draw [style=x-wire, in=90, out=180] (24.center) to (23.center);
		\draw [style=x-wire, in=180, out=0, looseness=1.50] (24.center) to (22.center);
		\draw [style=x-wire] (24.center) to (20.center);
		\draw [style=x-wire, in=180, out=-90, looseness=1.25] (27.center) to (28.center);
		\draw [style=x-wire, in=0, out=90, looseness=1.25] (31.center) to (32.center);
		\draw [style=x-wire, in=0, out=-90, looseness=1.25] (6.center) to (28.center);
		\draw [style=x-wire, in=-180, out=90, looseness=1.25] (29.center) to (32.center);
		\draw [style=x-wire] (32.center) to (28.center);
	\end{pgfonlayer}
\end{tikzpicture}} \\
\multicolumn{2}{c}{Equation \eqref{eq:frob}: Frobenius}  \\
\begin{tikzpicture}
	\begin{pgfonlayer}{nodelayer}
		\node [style=none] (28) at (-1, 1) {};
		\node [style=none] (30) at (-1, 0.5) {};
		\node [style=none] (31) at (-0.6, 0) {};
		\node [style=none] (32) at (-1, -1) {};
		\node [style=none] (33) at (-1.4, 0) {};
		\node [style=none] (34) at (-1, -0.5) {};
		\node [style=none] (35) at (1, 1) {};
		\node [style=none] (38) at (1, -1) {};
		\node [style=none] (39) at (0, 0) {$=$};
	\end{pgfonlayer}
	\begin{pgfonlayer}{edgelayer}
		\draw [style=x-wire] (30.center) to (28.center);
		\draw [style=x-wire, bend right=45] (33.center) to (34.center);
		\draw [style=x-wire, bend left=45] (31.center) to (34.center);
		\draw [style=x-wire] (34.center) to (32.center);
		\draw [style=x-wire, bend right=45] (31.center) to (30.center);
		\draw [style=x-wire, bend left=45] (33.center) to (30.center);
		\draw [style=x-wire] (35.center) to (38.center);
	\end{pgfonlayer}
\end{tikzpicture}        & \begin{tikzpicture}
	\begin{pgfonlayer}{nodelayer}
		\node [style=none] (8) at (-0.6, 0) {};
		\node [style=none] (9) at (-1.4, 0) {};
		\node [style=none] (10) at (-0.6, -0.775) {};
		\node [style=none] (11) at (-1.4, -0.775) {};
		\node [style=none] (12) at (1.45, 0) {};
		\node [style=none] (13) at (0.65, 0) {};
		\node [style=none] (14) at (1.45, -0.775) {};
		\node [style=none] (15) at (0.65, -0.775) {};
		\node [style=none] (16) at (0, -0.25) {$=$};
	\end{pgfonlayer}
	\begin{pgfonlayer}{edgelayer}
		\draw [style=x-wire, bend right=90, looseness=2.25] (8.center) to (9.center);
		\draw [style=x-wire, in=90, out=-90, looseness=1.25] (8.center) to (11.center);
		\draw [style=x-wire, bend left=90, looseness=2.25] (9.center) to (8.center);
		\draw [style=x-wire, in=90, out=-90, looseness=1.25] (9.center) to (10.center);
		\draw [style=x-wire, bend right=90, looseness=2.25] (12.center) to (13.center);
		\draw [style=x-wire, bend left=90, looseness=2.25] (13.center) to (12.center);
		\draw [style=x-wire] (13.center) to (15.center);
		\draw [style=x-wire] (12.center) to (14.center);
	\end{pgfonlayer}
\end{tikzpicture}      \\
Equation \eqref{eq:special}: special,         & Equation \eqref{eq:symmetric}: symmetric      
\end{tabular}
\end{center}

\begin{exam}[Example \ref{exam:classical} continued]
    When $X$ is classical, the Frobenius equation shows us three equivalent ways of projecting onto pairs of equal elements. Given a pair of elements $(x,y)$, either copy $y$ and then check one output against $x$, or check $x$ against $y$ and then copy the output, or copy $x$ and check one output against $y$. The specialness equation says that copying an element $x$ and then checking it against itself does nothing. The symmetry equation says that checking $(x,y)$ against each other and deleting the result is the same as if the input was $(y,x)$. In fact, in the classical case, we have the stronger property -- commutativity. Checking $(x,y)$ against each other is the same as checking $(y,x)$ against each other.
\end{exam}

\begin{exam}[Example \ref{exam:max_quant} continued]
When $X$ is maximally quantum, $X$ can be thought of as $\mathbb{C}^n \otimes (\mathbb{C}^n)^*$, and we can prove the above properties by splitting the $X$ wire into two oriented Hilbert space wires (where the upward arrow depicts a Hilbert space and the downward arrow its canonical dual space), and depicting multiplication and unit as follows:
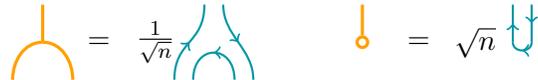
\ctikzfig{mult\_unit\_pop}
This is a weighted instance of the infamous graphical \emph{pair of pants algebra}, named for the graphical similarity of the multiplication to the article of clothing. See, for instance, \cite{catsQT} for details.
\end{exam}

In the general quantum set case, $X$ can be graphically broken down as a direct sum of normalized pair of pants algebras. We will use this in several proofs.

Note that the properties of being an algebra can also be represented graphically (associativity, compatibility of unit with multiplication).

\begin{remark}[Graphical proofs]
    The result of all of these diagrammatic definitions is that the algebraic rules for quantum sets allow for somewhat freely manipulating their diagrams in the plane. Wires can be bent and passed over nodes. Nodes can slide along wires. The quantum set operations can be rotated inside the plane, keeping their input and output wires in the same relative positions. These are all correct manipulations up to \emph{planar oriented isotopy}, as stated to be valid for $\dagger$-SSFAs in a pivotal category (which $\textbf{FHilb}$ is) in \cite[Theorem 3.28]{catsQT}. Referencing each algebraic property would make some of our proofs quite arduous, so we will proceed graphically at times.
\end{remark}

Now that we have introduced our quantum objects, let us speak of the morphisms between them.

\begin{defn}[Quantum function, \cite{MRV}]
\label{defn:qfunc}
Given quantum sets $X, A$ and a Hilbert space $\cl H$, a quantum function

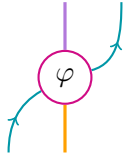
\ctikzfig{q_func}
from $X$ to $A$ over $\cl H$, denoted $\varphi: X \to_\cl H A$, is a linear map $\cl H \otimes X \to A \otimes \cl H$ satisfying

\begin{equation}\label{eq:q_comult}
\begin{tikzpicture}
	\begin{pgfonlayer}{nodelayer}
		\node [style=none] (0) at (-0.75, -1) {};
		\node [style=none] (4) at (0, -1) {};
		\node [style=none] (6) at (1, 1.75) {};
		\node [style=q-func] (7) at (0, 0) {$\varphi$};
		\node [style=none] (8) at (0.4, 1.75) {};
		\node [style=none] (10) at (-0.4, 1.75) {};
		\node [style=none] (11) at (0, 1) {};
		\node [style=none] (12) at (1.5, -1) {};
		\node [style=none] (14) at (4, 1.75) {};
		\node [style=none] (15) at (2.75, -1) {};
		\node [style=q-func] (16) at (3.4, 1) {$\varphi$};
		\node [style=q-func] (17) at (2.1, 0.25) {$\varphi$};
		\node [style=none] (18) at (2.75, -0.5) {};
		\node [style=none] (19) at (3.4, 1.75) {};
		\node [style=none] (20) at (2.1, 1.75) {};
		\node [style=none] (21) at (1.25, 0.5) {$=$};
		\node [style=none] (22) at (3.4, 0.25) {};
		\node [style=none] (23) at (2.75, -0.5) {};
	\end{pgfonlayer}
	\begin{pgfonlayer}{edgelayer}
		\draw [style=hilbert, in=-90, out=15] (7) to (6.center);
		\draw [style=x-wire] (7) to (4.center);
		\draw [style=a-wire, in=180, out=-90] (10.center) to (11.center);
		\draw [style=a-wire, in=0, out=-90] (8.center) to (11.center);
		\draw [style=a-wire] (7) to (11.center);
		\draw [style=x-wire, in=165, out=-90, looseness=1.25] (17) to (18.center);
		\draw [style=a-wire] (20.center) to (17);
		\draw [style=a-wire] (19.center) to (16);
		\draw [style=hilbert, bend left=15] (12.center) to (17);
		\draw [style=hilbert, in=-135, out=45] (17) to (16);
		\draw [style=hilbert, bend right, looseness=0.75] (16) to (14.center);
		\draw [style=x-wire] (18.center) to (15.center);
		\draw [style=x-wire, in=15, out=-90, looseness=1.25] (22.center) to (23.center);
		\draw [style=x-wire] (22.center) to (16);
		\draw [style=hilbert, in=-165, out=90] (0.center) to (7);
	\end{pgfonlayer}
\end{tikzpicture}
\end{equation}
\begin{equation}
\begin{tikzpicture}
	\begin{pgfonlayer}{nodelayer}
		\node [style=none] (19) at (1.75, -1) {};
		\node [style=x-frob] (14) at (2.5, -0.5) {};
		\node [style=none] (4) at (0, -1) {};
		\node [style=a-frob] (5) at (0, 0.75) {};
		\node [style=none] (6) at (0.75, 1) {};
		\node [style=none] (8) at (1.75, -1) {};
		\node [style=none] (9) at (2.5, -1) {};
		\node [style=none] (11) at (3.25, 1) {};
		\node [style=none] (12) at (2.5, 0) {};
		\node [style=none] (16) at (1.25, 0) {$=$};
		\node [style=none] (17) at (-0.75, -1) {};
		\node [style=q-func] (18) at (0, 0) {$\varphi$};
	\end{pgfonlayer}
	\begin{pgfonlayer}{edgelayer}
		\draw [style=hilbert] (19.center)
			 to [in=-165, out=90] (12.center)
			 to [in=-90, out=15] (11.center);
		\draw [style=x-wire] (14) to (9.center);
		\draw [style=hilbert, in=-165, out=90] (17.center) to (18);
		\draw [style=hilbert, in=-90, out=15] (18) to (6.center);
		\draw [style=x-wire] (18) to (4.center);
		\draw [style=a-wire] (18) to (5);
	\end{pgfonlayer}
\end{tikzpicture}\label{eq:q_counital}
\end{equation}
\begin{equation}\label{eq:self_conj}
\begin{tikzpicture}
	\begin{pgfonlayer}{nodelayer}
		\node [style=none] (0) at (-0.75, 1) {};
		\node [style=none] (4) at (0, 1) {};
		\node [style=none] (5) at (0, -1) {};
		\node [style=none] (6) at (0.75, -1) {};
		\node [style=q-func] (7) at (0, 0) {$\varphi^\dagger$};
		\node [style=none] (11) at (4.5, 1) {};
		\node [style=none] (17) at (1.5, 0) {$=$};
		\node [style=none] (23) at (2.25, -1) {};
		\node [style=none] (26) at (2.25, 0) {};
		\node [style=none] (27) at (3, -1) {};
		\node [style=none] (29) at (5.25, 1) {};
		\node [style=none] (30) at (5.25, 0) {};
		\node [style=q-func] (31) at (3.75, 0) {$\varphi$};
	\end{pgfonlayer}
	\begin{pgfonlayer}{edgelayer}
		\draw [style=hilbert, in=-90, out=165] (7) to (0.center);
		\draw [style=x-wire] (7) to (4.center);
		\draw [style=a-wire] (7) to (5.center);
		\draw [style=hilbert, in=-15, out=90] (6.center) to (7);
		\draw [style=a-wire] (23.center)
			 to (26.center)
			 to [bend left=90, looseness=1.50] (31.center);
		\draw [style=x-wire] (29.center)
			 to (30.center)
			 to [bend left=90, looseness=1.50] (31.center);
		\draw [style=hilbert, in=-165, out=90] (27.center) to (31);
		\draw [style=hilbert, in=-90, out=15] (31) to (11.center);
	\end{pgfonlayer}
\end{tikzpicture}
\end{equation}

\end{defn}

Letting $\cl{H} = \bb{C}$ in the above definition recovers the definition of a $*$-cohomomorphism from $X$ to $A$, thought of as tracial C$^*$-algebras. In that case,  \eqref{eq:q_comult} becomes comultiplicativity, \eqref{eq:q_counital} becomes counitality or trace preservation, and \eqref{eq:self_conj} becomes $*$-preservation. 

Quantum sets as objects with quantum functions as morphisms form a category, as shown in \cite{MRV}. The Hilbert space wire in the definition of a quantum function can be thought of as a resource space. When $X,A$ are classical, a quantum function is a family of PVMs indexed by $X$ with outcomes in $A$.

The wire $\cl H$ could be infinite-dimensional in general, but because several results in this work rely on traces, we will stick to finite-dimensional spaces here. We expect our results generalize to $\cl H$ infinite dimensional at least whenever the dual space $\cl H^*$ is not involved in the proof. In fact, infinite dimensional Hilbert spaces do not have duals in the dagger categorical sense \cite{catsQT}. The Hilbert space dual is only a dagger categorical dual for finite dimensional Hilbert spaces.

\subsection*{Games, strategies, and correlations}
In the following, let $X,Y,A,B$ be quantum sets. We will sometimes use a single wire instead of two to graphically represent the quantum set $X\otimes Y$ and similarly $A\otimes B$. In order to make sense of this, note that $X\otimes Y$ is itself a quantum set when equipped with operations
\begin{align}
\label{eq:tensor_mult}
m &:= (m_X \otimes m_Y)(\mathrm{Id}_X \otimes \sigma \otimes \mathrm{Id}_Y),\\
\label{eq:tensor_unit}
u &:= u_X \otimes u_Y
\end{align}
and $\Delta := m^\dagger, \varepsilon := u^\dagger$.

Graphically, this looks like using a thick wire to replace the pair of wires $X \otimes Y$
\ctikzfig{product_wire}
and then defining a new multiplication and unit
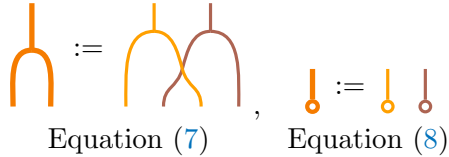
\begin{center}
\begin{tabular}{c c}
\begin{tikzpicture}[scale=0.5]
	\begin{pgfonlayer}{nodelayer}
		\node [style=none] (1) at (-1, 0) {};
		\node [style=none] (2) at (1, 0) {};
		\node [style=none] (4) at (3, 0) {};
		\node [style=none] (7) at (1.75, 2) {};
		\node [style=none] (9) at (1.75, 2.75) {};
		\node [style=none] (11) at (-1.5, 2.75) {};
		\node [style=none] (13) at (-2, 0) {};
		\node [style=none] (14) at (-1.5, 1.5) {};
		\node [style=none] (15) at (4, 0) {};
		\node [style=none] (16) at (2, 0) {};
		\node [style=none] (17) at (2.5, 1) {};
		\node [style=none] (18) at (3.25, 2) {};
		\node [style=none] (19) at (3.25, 2.75) {};
		\node [style=none] (20) at (0, 1.5) {$:=$};
	\end{pgfonlayer}
	\begin{pgfonlayer}{edgelayer}
		\draw [style=x-wire] (9.center) to (7.center);
		\draw [style=x-wire, in=90, out=180] (7.center) to (2.center);
		\draw [style=xy-wire, in=90, out=180] (14.center) to (13.center);
		\draw [style=xy-wire] (11.center) to (14.center);
		\draw [style=xy-wire, in=90, out=0] (14.center) to (1.center);
		\draw [style=y-wire] (19.center) to (18.center);
		\draw [style=y-wire, in=90, out=0] (18.center) to (15.center);
		\draw [style=y-wire, in=75, out=180] (18.center) to (17.center);
		\draw [style=y-wire, in=90, out=-105, looseness=1.25] (17.center) to (16.center);
		\draw [style=x-wire, in=105, out=0] (7.center) to (17.center);
		\draw [style=x-wire, in=90, out=-75, looseness=1.25] (17.center) to (4.center);
	\end{pgfonlayer}
\end{tikzpicture}, & \begin{tikzpicture}[scale=0.5]
	\begin{pgfonlayer}{nodelayer}
		\node [style=xy-frob] (0) at (-1, 0) {};
		\node [style=none] (1) at (-1, 1) {};
		\node [style=x-frob] (2) at (1, 0) {};
		\node [style=none] (3) at (1, 1) {};
		\node [style=y-frob] (4) at (2, 0) {};
		\node [style=none] (5) at (2, 1) {};
		\node [style=none] (6) at (0, 0.5) {$:=$};
	\end{pgfonlayer}
	\begin{pgfonlayer}{edgelayer}
		\draw [style=xy-wire] (0) to (1.center);
		\draw [style=y-wire] (4) to (5.center);
		\draw [style=x-wire] (2) to (3.center);
	\end{pgfonlayer}
\end{tikzpicture} \\
Equation \eqref{eq:tensor_mult} &

Equation \eqref{eq:tensor_unit}  
\end{tabular}
\end{center}

We now generalize nonlocal games to allow quantum questions and answers.

\begin{defn}[Game] \label{defn:game} A \emph{(quantum) game} is a linear map $\lambda:X\otimes Y \to A \otimes B$ satisfying
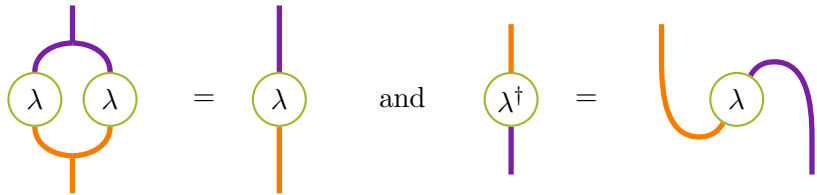
\begin{align*}
    &\begin{tikzpicture}
	\begin{pgfonlayer}{nodelayer}
		\node [style=game] (2) at (1, 0) {$\lambda$};
		\node [style=none] (5) at (-1.75, 1.25) {};
		\node [style=none] (6) at (-1.75, -1.25) {};
		\node [style=none] (7) at (1, -1.25) {};
		\node [style=none] (8) at (1, 1.25) {};
		\node [style=none] (9) at (0, 0) {$=$};
		\node [style=none] (18) at (-1.75, 0.75) {};
		\node [style=game] (21) at (-1.25, 0) {$\lambda$};
		\node [style=none] (22) at (-1.75, -0.75) {};
		\node [style=game] (23) at (-2.25, 0) {$\lambda$};
	\end{pgfonlayer}
	\begin{pgfonlayer}{edgelayer}
		\draw [style=ab-wire] (8.center) to (2);
		\draw [style=xy-wire] (2) to (7.center);
		\draw [style=xy-wire, in=180, out=-90] (23) to (22.center);
		\draw [style=xy-wire] (6.center) to (22.center);
		\draw [style=ab-wire, in=180, out=90] (23) to (18.center);
		\draw [style=ab-wire] (18.center) to (5.center);
		\draw [style=ab-wire, in=0, out=90] (21) to (18.center);
		\draw [style=xy-wire, in=0, out=-90] (21) to (22.center);
	\end{pgfonlayer}
\end{tikzpicture}& \textrm{and}\qquad & \begin{tikzpicture}
	\begin{pgfonlayer}{nodelayer}
		\node [style=none] (31) at (-1, 0.5) {};
		\node [style=none] (33) at (-0.5, -0.5) {};
		\node [style=game] (36) at (0, 0) {$\lambda$};
		\node [style=none] (37) at (1, -0.5) {};
		\node [style=none] (38) at (0.5, 0.5) {};
		\node [style=none] (39) at (1, -1) {};
		\node [style=none] (41) at (-1, 1) {};
		\node [style=none] (43) at (-3, 1) {};
		\node [style=game] (44) at (-3, 0) {$\lambda^\dagger$};
		\node [style=none] (46) at (-3, -1) {};
		\node [style=none] (47) at (-2, 0) {$=$};
	\end{pgfonlayer}
	\begin{pgfonlayer}{edgelayer}
		\draw [style=ab-wire] (39.center) to (37.center);
		\draw [style=ab-wire] (37.center)
			 to [in=0, out=90] (38.center)
			 to [in=60, out=-180] (36);
		\draw [style=xy-wire] (36)
			 to [in=0, out=-120] (33.center)
			 to [in=-90, out=-180] (31.center)
			 to (41.center);
		\draw [style=ab-wire] (46.center) to (44);
		\draw [style=xy-wire] (44) to (43.center);
	\end{pgfonlayer}
\end{tikzpicture}
\end{align*}
\end{defn}

In case $X,Y,A,B$ are commutative, they can be thought of as classical sets\footnote{This is via Gelfand duality, and is discussed in greater detail in \cite{MRV}.}, and then $\lambda$ is a linear map $C(X\times Y) \to C(A \times B)$ with idempotent coefficients. This can be thought of as the usual rule function $X\times Y \times A \times B \to \{0,1\}$, for nonlocal games with classical questions and answers. In the classical case, the second diagram (self-conjugacy) is redundant given the first.

\begin{remark}[Games are bipartite quantum relations]
    Consider linear maps $f:X\otimes Y \to A\otimes B$. Given a game $\lambda$, the space of all linear maps $f$ such that $f \star \lambda = f$ is a quantum relation as defined in \cite{Kor20}, which are known to be equivalent to the quantum relations of \cite{Wea10}. In this sense, a quantum game is a quantum relation from the question space to the answer space. There are connections to be explored in future work between the following presentation of quantum games, and the point of view of quantum relations.
\end{remark}

In the remainder of this section, when we refer to the game $\lambda$, we mean $\lambda:X\otimes Y \to A\otimes B$. We now introduce a ``container'' object that can hold any type of strategy.

\begin{defn}[Combined strategy]\label{defn:combined_strat} A \emph{combined strategy} for the game $\lambda$ over a Hilbert space $\cl{H}$ is a quantum function $\varphi: X\otimes Y \to_\cl{H} A\otimes B$, depicted by
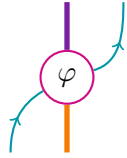
\ctikzfig{combined\_strat}
\end{defn}
The word \emph{combined} is used because it rolls both players' strategies together into a single map. There is no guarantee a given combined strategy can be implemented under non-communication requirements on the players. That is because there is no separation of Alice and Bob baked into the definition. Later we will define various types of \emph{valid strategies}, meaning strategies that can be understood to physically model spacetime-separated players with a non-communication requirement. The definition of combined strategy is useful merely because it allows us to talk about all types of valid strategies at once. 

To be concise, we define strategies with reference to a game $\lambda$ instead of listing the relevant question and answer sets, even though only the question and answer sets of $\lambda$ are part of the definition, and not the rule function. Essentially, if $\lambda_1, \lambda_2: X\otimes Y \to A \otimes B$, then a strategy for $\lambda_1$ will also a strategy for $\lambda_2$ and vice-versa.

When we aren't interested in varying the state of the resource system $\cl H$ used by the players, we use a correlation instead of a strategy. To define correlations, we first diagrammatically define channels between quantum sets, and to do this we make use of the diagrammatic definition of complete positivity.

\begin{defn}[Completely positive]\label{defn:CP}
A linear map $f:X\to A$ for quantum sets $X,A$ is \emph{completely positive} (CP) if there exists a linear map $g:X\otimes A \to Q$ for some other quantum set $Q$ satisfying
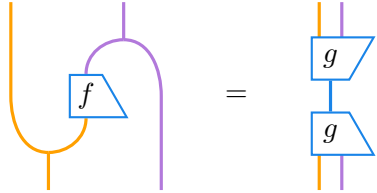
\ctikzfig{CP_defn}
\end{defn}

In \textbf{FHilb}, this diagrammatic notion agrees with the usual definition of complete positivity. See \cite[Theorem 7.18]{catsQT} for details.

\begin{prop}\label{prop:CP_self_conj}
CP maps between quantum sets are self-conjugate.
\end{prop}

The proof is in \ref{pf:CP_self_conj}.

\begin{defn}[Counital]\label{defn:counital}
A linear map $f:X\to A$ for quantum sets $X,A$ is \emph{counital} if
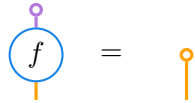
\ctikzfig{counital}
\end{defn}

Counitality implements the physical property of causality when the counit is thought of as a discarding map. See \cite{PQP} for details. This condition is called trace-preserving when thinking of quantum sets as $C^*$-algebras or matrix algebras, although note that the counit will not generally be given by the canonical trace, but rather by a reweighted version. See Remark \ref{rem:qset_structure}.

\begin{defn}[Channel] Given quantum sets $X,A$, a \emph{channel from $X$ to $A$} is a completely positive, counital linear map $f:X\to A$.
\end{defn}

Channels here are just the channels in \cite{CHK14}, given as normalized morphisms in\\ $\mathbf{CP^* [FHilb]}$.

\begin{exam}[Maximally quantum and classical channels]
In the maximally quantum case, this notion of channel recovers the usual definition of a quantum channel as a CPTP map between matrix algebras, but with scaled traces.

In the classical case, this definition picks out exactly classical channels, which are probability distributions $p(a|x)$ on the set $A$ conditioned on the set $X$. In that setting, complete positivity is equivalent to non-negativity of $p$ and counitality is equivalent to $\sum_a p(a|x) = 1$.
\end{exam}

\begin{defn}[Correlation] A \emph{quantum set correlation} is a channel $P$ from $X\otimes Y$ to $A \otimes B$.
\end{defn}

Note, that this should not be confused with the typical notion of correlations induced by quantum strategies, but rather is a generalization to when the question and answer sets are quantum. With that said, in the following we will drop the words \emph{quantum set} and just write \emph{correlation}.  Our definition of correlation is the dagger dual of \cite[Definition 5.3]{BKS21}.

To relate strategies and correlations, we need a notion of (pure) states. We use \cite[Definition 1.10]{catsQT}, as follows.

\begin{defn}[(Normalized) state]
     Let $\cl H$ be a finite-dimensional Hilbert space. A \emph{state} of $\cl H$ is a linear map $\ket\psi: \mathbb{C}\to \cl{H}$, depicted
    \ctikzfig{state}
    We denote the corresponding \emph{effect} by $\bra \psi := \ket \psi ^\dagger$, and depict it
    \ctikzfig{effect}
    We say $\ket \psi$ is \emph{normalized} if
    \ctikzfig{normalized}
\end{defn}

Recall that the empty diagram above represents $\mathrm{Id}_\mathbb{C}$, or scalar multiplication by $1$.

\begin{defn}[Realized by]\label{defn:real}
    A correlation $P$ is said to be \emph{realized by} a combined strategy $\varphi$ on $\cl H$ and a normalized state $\ket\psi \in \cl H$ when
    \ctikzfig{corr\_from\_strat}
\end{defn}

The physical interpretation of a correlation as a channel from question pairs to answer pairs emphasizes the observable ``signature" of a strategy, from the referee's perspective, not having access to the players' resource state. 

\begin{remark}\label{rem:corr_from_strat}
    A combined strategy $\varphi$ over a finite-dimensional $\cl{H}$ and a normalized state $\ket{\psi} \in \cl{H}$ always give rise to some correlation $P$, as in Definition \ref{defn:real}. See \ref{pf:corr_from_strat} for the proof.
\end{remark}

In the study of nonlocal games, we are especially interested in correlations that are guaranteed to win a game, also known as perfect correlations. Let us define a perfect correlation in this general setting.

\begin{defn}[Perfect correlation]\label{defn:perf_corr}
A correlation $P$ for the game $\lambda$ is called \emph{perfect} if
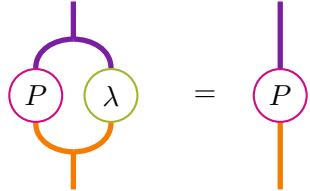
\ctikzfig{perf\_corr}
\end{defn}

Note that swapping $P$ and $\lambda$ in the definition is equivalent due to both maps being self-conjugate.

\begin{remark}[Typical perfect correlations]
    In the case that questions and answers are classical, this says that the question-answer pairs that occur with nonzero probability according to $P$ are winning question-answer pairs according to $\lambda$, which is exactly the typical definition of perfect. Indeed, the above equation becomes $$p(a,b|x,y)\lambda(a,b|x,y) = p(a,b|x,y)$$ for all $a \in A, b\in B, x\in X, y\in Y$. This is equivalent to saying $\lambda(a,b|x,y) =0$ implies $p(a,b|x,y)=0$.
\end{remark}

By adding a resource wire to the definition of a perfect correlation, we can define perfect strategies.

\begin{defn}[Perfect strategy]\label{defn:perf_strat}
A (combined) strategy $\varphi$ for the game $\lambda$ is called \emph{perfect} if
\ctikzfig{perf\_strat}
\end{defn}

If a strategy over resource space $\cl H$ is perfect for $\lambda$ and realizes a correlation $P$ together with a normalized state in $\cl H$, then $P$ is perfect for $\lambda$.

\begin{remark}[Perfect strategies versus perfect correlations]\label{rmk:perf_strat}
    If a strategy over resource space $\cl H$ is perfect for $\lambda$ and realizes a correlation $P$ together with a normalized state in $\cl H$, then $P$ is perfect for $\lambda$. However, the converse implication does not hold. Perfect correlations need not come from perfect strategies. Indeed, being a perfect strategy has much stronger consequences than realizing a perfect correlation. In light of this, it may be more appropriate to think of perfect strategies as \emph{absolutely perfect}. For typical nonlocal games, the words \emph{perfect strategy} are used to refer to the correlation being perfect.
\end{remark}

\subsection*{Strategy classes}

We have the equipment to describe some of the typical strategy classes for nonlocal games (specifically -- deterministic, quantum tensor, and quantum commuting strategies), and easily extend them to quantum games. All of these classes fall under the umbrella of non-signalling strategies (realizing non-signalling correlations). We will begin with the diagrammatic definition of non-signalling, and work towards the stricter strategy classes.

\begin{defn}[Non-signalling, marginals] A quantum-set correlation $P: X\otimes Y \to A\otimes B$ is called \emph{non-signalling} if there are linear maps $P_A:X\to A$ and $P_B:Y\to B$ such that
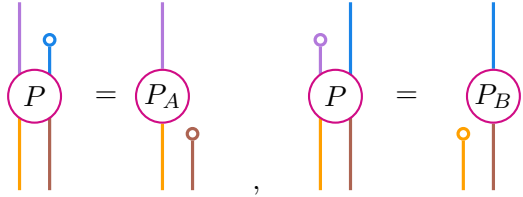
\ctikzfig{nonsig}
The maps $P_A, P_B$ are called the \emph{marginals} of $P$.
\end{defn}

This reduces to the typical definition of non-signalling for classical question and answer sets. Non-signalling quantum-set correlations are also referred to as QNS (quantum non-signalling) correlations, and were introduced in \cite{DW16}.

\begin{cor}\label{cor:margs_channels}
    The marginals of a non-signalling correlation $P$ are channels.
\end{cor}

See \ref{pf:margs_channels} for the proof. If we are interested not just in marginals, but in Alice and Bob's implementation of their strategies, we need to look at valid strategy classes. The most general valid strategy class we will deal with here is quantum commuting.

\begin{defn}[Quantum commuting strategy] A \emph{quantum commuting strategy} for the game $\lambda$ is a pair of quantum functions $(E,F)$ on a Hilbert space $\cl H$, given by $E: X \to_\cl H A, F: Y \to_\cl H B$, that commute on $\cl H$. In diagrams, \ctikzfig{qc_strat}
\end{defn}

A quantum commuting strategy gives rise to a combined strategy on $\cl H$ as follows: \ctikzfig{combined_qc}

If Alice and Bob each have their own resource Hilbert space instead of sharing a Hilbert space, we get quantum tensor strategies.

\begin{defn}[Quantum tensor strategy] A \emph{quantum tensor strategy} for the game $\lambda$ is a pair of quantum functions $(E,F)$, given by $E: X \to_{\mathcal H_A} A, F: Y \to_{\mathcal H_B} B$. In diagrams, \ctikzfig{qt_strat}.
\end{defn}

A quantum tensor strategy gives rise to a combined strategy on $\cl H := \cl H_A \otimes \cl H_B$ as follows:
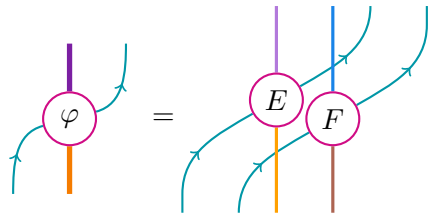
\ctikzfig{combined_qt}

Note that a quantum tensor strategy is the composition defined in \cite{MRV} of the quantum functions $(E\otimes I_B)\circ (I_A \otimes F)$.

\begin{remark}
    As $\cl H$ is finite-dimensional here, all quantum commuting strategies should decompose as quantum tensor strategies. The reasons for introducing both classes here are: to pave the way to consider an infinite-dimensional $\cl H$; and to mirror the language already used in nonlocal games.
\end{remark}

\begin{defn}[Deterministic strategy]  A \emph{deterministic strategy} for the game $\lambda$ is a quantum tensor strategy with classical resource spaces, i.e. $\cl H_A = \cl H_B = \mathbb{C}$. 
\end{defn}

\begin{prop}
    Deterministic strategies are quantum tensor. Quantum tensor strategies are quantum commuting. Quantum commuting strategies realize non-signalling correlations.
\end{prop}

\begin{proof}
    Deterministic strategies are quantum tensor by definition.
    
    A quantum tensor strategy can be thought of as quantum commuting by setting $\cl H := \cl H_A \otimes \cl H_B$ and extending the resource space of the given $E$, $F$ to act trivially on the newly accessible half of $\cl H$.
    
    A quantum commuting strategy realizes a nonsignalling correlation due to counitality of $E, F$.    
\end{proof}

\begin{remark}[Beyond quantum functions] Quantum functions quantize the classical question and answer game strategies given by families of PVMs. Seemingly more generally, Alice and Bob may use families of POVMs, but this is made unnecessary, via Naimark's dilation theorem (a consequence of Stinespring's theorem), by working in a larger resource space. For reference, see \cite[Proposition 9.6 and Theorem 9.8]{paulsen16}.

In the quantum question-and-answer setting, there is a similar result to Naimark's dilation theorem presented in \cite{BKS21}.
\end{remark}

\subsection*{Products of games}
\label{sec:products}

We mention two products of games -- their tensor product, and their commuting product, and we relate perfect strategies for products of games to perfect strategies for the individual games.

\begin{defn}[Tensor product of games]
    The \emph{tensor product} of games $\lambda_1, \lambda_2$ is denoted $\lambda_1 \otimes \lambda_2$ and given by
    \ctikzfig{tensor\_prod\_games}
\end{defn}

Some wires are crossed in the above definition so that the new game keeps Alice's questions and answers separate from Bob's. Note that the tensor product of games is a game. The tensor product (with the same wires crossed) of perfect strategies or correlations for each game is perfect for the product.

\begin{defn}[Commuting product of games]
    Given games $\lambda_1, \lambda_2$ with the same question and answer sets satisfying
    \ctikzfig{games\_commute}
    we denote the above product $\lambda_1 \star \lambda_2$ and we call it the \emph{commuting product} of games.
\end{defn}

Note that the commuting product of games is a game. Note also that in the classical question and answer case, this is just the Schur product of the rule matrices. 

If the commuting product of games is well-defined, we say we have \emph{commuting games}.

\begin{prop} \label{prop:perf_product}
    Let $\lambda_1, \lambda_2$ be commuting games $X\otimes Y \to A\otimes B$. Let $\varphi$ be a perfect strategy for both $\lambda_1, \lambda_2$. Then $\varphi$ be a perfect strategy for both $\lambda_1, \lambda_2$ if and only if $\varphi$
    is a perfect strategy for $\lambda_1 \star \lambda_2$.   Analogously, $P$ is a perfect correlation for both $\lambda_1, \lambda_2$ if and only if $P$ is a perfect correlation for $\lambda_1 \star \lambda_2$.
\end{prop}

The proof follows from the definitions of perfect strategy and perfect correlation, and the fact that $\lambda_1, \lambda_2$ commute. See \cite[Proposition 4.43]{myThesis} for details.

The above proposition supports viewing the commuting product of games as a logical AND operation.

\section{Synchronicity}
\label{sec:sync}

We start by defining the opposite of a quantum set, and then studying the Frobenius map, which we will also refer to as \emph{sharing}. We then use sharing to define synchronicity, the central property of this paper. We ensure that our terminology is apt by considering the classical question and answer case.

\subsection*{Opposite of a quantum set}
\begin{defn}[Opposite of a quantum set]
    Given a quantum set $X$, its \emph{opposite}, denoted $\op X$, is another quantum set with the same underlying set $X$, the same unit as $X$,
    \ctikzfig{opposite\_unit}
    and multiplication given by $m\sigma$,
    \ctikzfig{opposite\_mult}
\end{defn}

In the following, we always take $Y = \op X$, and $B=\op A$, with  filled nodes corresponding to the opposite multiplication. The choice to use the opposite quantum set for Bob initially resulted from some technical issues in the proofs. There is a more operational interpretation: Using the opposite quantum set for Bob results in a half-twisted multiplication on the composite system belonging to Alice and Bob. Untwisting this multiplication leaves Alice's input wires adjacent, and Bob's input wires adjacent, representing the players' physical separation. This is partially demonstrated below, in the proof of property (d) of sharing.

There is no need to distinguish the cups and caps that come from the original quantum set or its opposite, as \eqref{eq:symmetric} shows us that they are equal. Under the duality coming from the tensor product quantum set structure of $X\otimes \op X$, we can check that the transpose of the filled multiplication, viewed as a linear map $X\otimes \op X \to X$, is the empty comultiplication and vice versa. Equivalently, the filled nodes become the empty nodes under conjugation and vice versa.

\subsection*{Sharing}

We now consider a \emph{sharing} map, defined as follows. This map features in the definition of synchronicity, and will also allow us to introduce the notion of classical dimension.

\begin{defn}[Sharing]
    The \emph{sharing map of a quantum set} $X$ is the Frobenius map of \eqref{eq:frob} given by
    \ctikzfig{frob}
\end{defn}

It may seem unnecessary to give a new name to the Frobenius map, but it makes some of our statements more operationally accessible. It also leaves room for different sharing maps to be used in the definition of synchronicity -- for instance, we could introduce a twist into the sharing map, and define synchronicity that way, with many of the same properties.

In the classical case, the sharing map is the projection onto the diagonal (or ``shared'') questions $\ket x\otimes \ket x$, which is the reason for calling it sharing.

We will often work with the thick wire $X\otimes \op X$ (and $A\otimes \op A$). It will be convenient to have a shorthand notation for sharing, using thick wires.

\ctikzfig{sharing\_thick\_wire}

We take note of some straightforward but useful properties of the sharing map.

\begin{lemma}[Properties of sharing]
\label{lem:sharing}
    Consider the sharing map of $X$. 
    \begin{enumerate}[(a)]
        \item \label{item:sharing_sa} Sharing is self-adjoint.
        \item \label{item:sharing_proj} Sharing is a projection.
        \item \label{item:sharing_unit} Sharing the unit produces the cup state.
        \ctikzfig{sharing\_unit}
        \item \label{item:sharing_mult} Sharing commutes with the right leg of multiplication.
        \ctikzfig{sharing\_right\_mult}
    \end{enumerate}
\end{lemma}

\begin{proof}
\begin{enumerate}[(a)]
    \item  Self-adjointness follows from applying a dagger to the Frobenius equations \eqref{eq:frob}, or graphically by noticing their vertical symmetry.
    \item Composing the sharing map with itself and then applying specialness of the quantum set $X$ returns the sharing map.
    \item Let's see that sharing the unit produces the cup state.
    \ctikzfig{pf\_sharing\_unit}
    \item Using thin wires, and repeated applications of the Frobenius equation and associativity, we have
    \ctikzfig{pf\_sharing\_mult}
\end{enumerate}
\end{proof}

Property (\ref*{item:sharing_unit}) of sharing motivates one more piece of shorthand notation, allowing us to represent the cups and caps of $X$ using only thick wires. We visually retract the unit in property (\ref*{item:sharing_unit}) to define
\ctikzfig{sharing\_unit\_shorthand}

We denote the conjugate of sharing, with respect to the duality given by the thick-wire cups and caps of $X\otimes \op X$, by the filled square,
\ctikzfig{sharing\_conjugate}

This is also the transpose of sharing, due to self-adjointness, and is also sharing in the opposite quantum set $\op X$. 

Property (\ref*{item:sharing_mult}) of sharing is equivalent to two other equations relating sharing, the conjugate of sharing, and multiplication. The equivalence comes from bending wires.

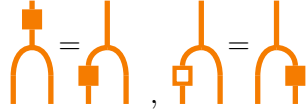
\ctikzfig{sharing\_mult\_bent}

Another way to approach the relationship between sharing and multiplication, is that if we think of $X\otimes \op X$ as a left-module over itself, then sharing is a left-module endomorphism. Similarly, opposite sharing is a right-module endomorphism.

\subsection*{Synchronicity}

We now define synchronicity in a way that turns out to be distinct from both \cite{BKS21} and \cite{BHTT21}. See Section \ref{sec:other_approaches} for comparisons.

\begin{defn}[Synchronous map]
    \label{defn:sync_map}
    We say a linear map $f:X\otimes \op X\to A\otimes \op A$ is \emph{synchronous} if
    \ctikzfig{sync\_map}
\end{defn}

Note that the sharing map of $X$, although defined on $X\otimes X$, also defines a linear map on $X\otimes \op X$.

The above definition, in casual terms, says a map is synchronous if ``shared inputs lead to shared outputs''. If the linear map is a game or correlation, we get the notion \emph{synchronous game}, or \emph{synchronous correlation}.

It will sometimes be useful to use the equivalent thin-wire version of the above definition, given by
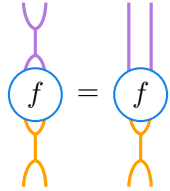
\ctikzfig{sync\_map\_thin\_wire}

Synchronous games and correlations with classical questions and answers are already an established notion, defined in \cite{PSSTW16}. Let's briefly call that property \emph{typically synchronous}. A correlation or game $f(a,b|x,y)$ is typically synchronous if for all $a\neq b$, $f(a,b|x,x) = 0$ for all $x$. It may be momentarily confusing to not give our more general property a unique name, but Proposition \ref*{prop:sync_is_sync} resolves this confusion.

\begin{prop}\label{prop:sync_is_sync}
    A classical correlation (game) is synchronous if and only if it is typically synchronous.
\end{prop}

\begin{proof}
    Let $X,A$ be classical. Let $P$ be a correlation $X\otimes X \to A \otimes A$. Denote the coefficients $p(a,b|x,y):= \bra a \bra b P \ket x \ket y.$ Being synchronous says that $P\ket x \ket x$ is in the diagonal subspace of $A\otimes A$. This means that the coefficients, $p(a,b|x,x)$ are zero unless $a=b$. The same proof applies when replacing the correlation $P$ by a game $\lambda$.
\end{proof}

\begin{remark}
A quantum processes interpretation of synchronicity for maps on classical sets: In the case that $X,A$ are classical, this equation can be understood as saying that after one instance of decoherence (the map $\Delta m$), the map $f$ commutes with any further decoherence. See \cite{PQP} for a diagrammatic definition of decoherence.
\end{remark}

\begin{defn}[Synchronous quantum function]
    \label{defn:sync_qfunc}
    We say a quantum function $\varphi: X\otimes \op X \to_{\cl H} A\otimes \op A$ is \emph{synchronous} if
    \ctikzfig{sync\_strat}
\end{defn}

The above definition gives us the notion of a \emph{synchronous strategy}. Note that if a strategy on $\cl H$ is  synchronous, then together with any normalized state $\ket \psi \in \cl H$, it realizes a synchronous correlation.

\section{Consequences of synchronicity}\label{sec:sync_cons}

\subsection*{Impact on perfect strategies}
One of the most useful properties of (typical) synchronous games is that they force associated perfect strategies to be synchronous. We generalize this result to quantum question and answer games.

\begin{thm}\label{thm:perf_sync}
    Perfect correlations for synchronous games are synchronous.
\end{thm}

\begin{proof}
    Let $P$ be a perfect correlation for a synchronous game $\lambda$. We have that
    \ctikzfig{pf\_perf\_sync}
    where
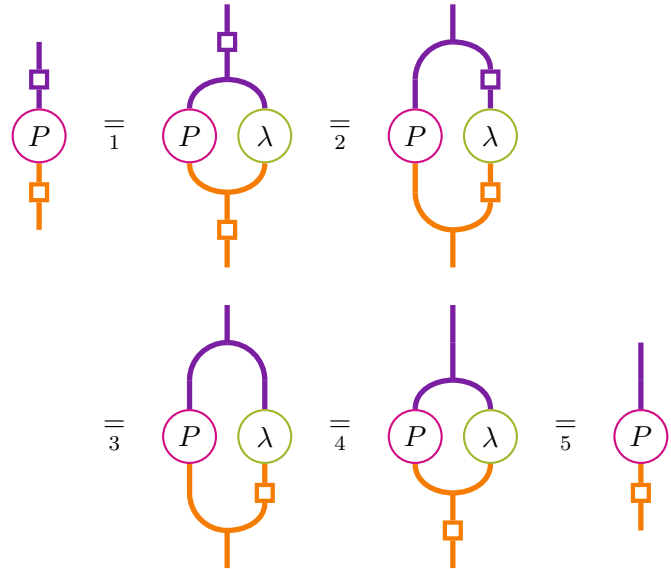
    \begin{enumerate}
        \item uses the definition of a perfect correlation,
        \item uses property (\ref*{item:sharing_mult}) of sharing,
        \item uses synchronicity of $\lambda$,
        \item again uses property (\ref*{item:sharing_mult}) of sharing, and
        \item again uses the definition of a perfect correlation.
    \end{enumerate}
\end{proof}

\begin{cor} \label{cor:perf_sync}
    Perfect strategies for synchronous games are synchronous.
\end{cor}

\begin{proof}
    While not truly a corollary, the steps of the proof are identical to the above, with $P$ replaced by a strategy $\varphi$, and Hilbert space wires tacked on to $\varphi$ in each diagram.
\end{proof}

We have presented the results in this order as the proof is cleaner and easier to read without the additional Hilbert space wire.

\subsection*{First examples of synchronous games}
\label{sec:example_games}

Let us now make some guesses as to what may constitute a synchronous quantum game. Our candidates are the identity map and the sharing map. In the classical case, the identity map can be thought of as the game \emph{Simon Says}, where Alice and Bob must each repeat their question as their answer, and the sharing map can be thought of as an unfair \emph{Simon Says}, where they can only win if their questions also agree with each other.

These are both synchronous maps, due to the fact that sharing is a projection. The identity is always a game, by specialness of a quantum set. However, this is not true of sharing.

\begin{prop}\label{prop:classical_sharing_game}
    $X$ is classical if and only if sharing is a game.
\end{prop}

\begin{proof}
    Let $X$ be classical. Then sharing is self-conjugate because $X = X^\mathrm{op}$.

    For classical sets, we have access to the spider theorem for $\dagger$-SCFAs, that asserts that all connected diagrams from $n$ input wires to $k$ output wires using only the quantum set operations are equal. For a full treatment of the spider theorem, see \cite[Section 5.2]{catsQT}. We can use the spider theorem to show sharing is a game. The left-hand-side of the idempotency equation for games is a connected 2-input, 2-output classical spider. By the spider theorem, this is equal to any other connected 2-input, 2-output spider, and so the idempotency equation defining a game holds.

    Conversely, assume sharing is a game. Recall that $X = \bigoplus_{k=1}^{K} M_{n_k}$, as in Remark \ref{rem:qset_structure}. We first show graphically that
    $$
        \begin{tikzpicture}[scale=0.6]
	\begin{pgfonlayer}{nodelayer}
		\node [style=half-xy-proj] (21) at (-0.025, 0) {};
		\node [style=half-xy-proj] (23) at (-1.025, 0) {};
	\end{pgfonlayer}
	\begin{pgfonlayer}{edgelayer}
		\draw [style=xy-wire, bend left=270, looseness=1.75] (21) to (23);
		\draw [style=xy-wire, bend left=90, looseness=1.75] (21) to (23);
	\end{pgfonlayer}
\end{tikzpicture}
 = K,
    $$
    as follows.
        \ctikzfig{pf\_classical\_dim\_sharing}
    with the coefficients introduced as given in Remark \ref{rem:qset_structure}, due to four instances of multiplication or comultiplication. But the complicated looking knot is just two closed loops in $\mathbb{C}^{n_k}$, so is equal to $\dim^2 \mathbb{C}^{n_k} = n_k^2$, so we get 
    $$
         = \sum_{k=1}^K 1 = K.
    $$
    Then applying units and counits to one of the defining equations for sharing being a game, we get
    \ctikzfig{pf\_classical\_sharing\_game}
    where the final loop comes from the fact that $X$ is special. The loop evaluates to $\dim X$. We conclude that $K = \dim X$, so we must have $n_k=1$ for all $k$. Therefore, $X$ is classical.
\end{proof}

So in general, sharing is synchronous but it is not a game outside of the classical setting. The identity map is our only example so far of a synchronous game with non-classical questions and answers.

We can also look at the \emph{function} game $\lambda_\to$, given by
    \ctikzfig{function\_game}
This is a game, and is synchronous. In fact, it is the ``easiest'' synchronous game, in the sense that perfect strategies for other synchronous games can also win the function game.

\begin{prop}[Function game is the easiest synchronous game]
    Given any synchronous game, $\lambda:X\otimes \op X \to A \otimes \op A$, the function game on the same sets commutes with $\lambda$, and satisfies $$\lambda \star \lambda_\to = \lambda.$$
\end{prop}

The proof follows directly from computing both products, and using synchronicity of $\lambda$ (with respect to both sharing and conjugate sharing) to cancel two of the three terms.

We will construct examples, in Section \ref{sec:graph_games}, of more interesting synchronous games with truly quantum questions and answers. 

We can also always construct more synchronous games by taking the commuting product of any game that commutes with a synchronous game, because of the following fact.

\begin{prop}\label{prop:prod_sync_sync}
    If $\lambda$ is synchronous, then $\lambda' \star \lambda$ is synchronous.
\end{prop}

\begin{proof}
    The proof is identical to steps 2-4 of the proof of Theorem \ref{thm:perf_sync}, replacing $P$ with $\lambda'$.
\end{proof}

Some of the games in this section motivate one more definition, that applies to a general game $\lambda:X\otimes Y\to A\otimes B$.

\begin{defn}[Fair game]
We say a game is \emph{fair} if there is no $\ket x \in X, \ket y \in Y$ such that $\lambda(\ket x \otimes \ket y) = 0$.    
\end{defn}

The identity map and the function game are fair games. The graph games in Section \ref{sec:graph_games} are also fair.

\subsection*{Classical dimension}  

The proof of Proposition \ref{prop:classical_sharing_game} motivates us to define the classical dimension of $X$ as the number $K$ of direct summands when $X$ is decomposed as a direct sum of matrix algebras. In diagrams, we can use the sharing map to make our definition without appealing to a structure theorem.

\begin{defn}[Classical dimension]
\label{def:classical_dim}
The \emph{classical dimension} of a quantum set $X$ (seen from the perspective of $X\otimes \op X$) is
\ctikzfig{classical\_dim\_sharing}
\end{defn}

We call this notion ``classical dimension'' because it counts classically distinct subspaces. At the level of Abelian categories, this notion is called the \emph{length} of $X$ as an object of $\textbf{FHilb}$. It should not be confused with other notions of classical dimension that may be in the literature.

We have defined classical dimension of $X$ using the wires of $X\otimes \op X$. Using the notation for the opposite Frobenius algebra, we can give a simpler, equivalent definition.

\begin{prop}
The classical dimension of a quantum set $X$ is also given by
\begin{tikzpicture}[scale=0.6]
	\begin{pgfonlayer}{nodelayer}
		\node [style=none] (37) at (0, 0.5) {};
		\node [style=x-frob] (38) at (-0.45, 0) {};
		\node [style=none] (39) at (0, -0.5) {};
		\node [style=dark-x-frob] (40) at (0.45, 0) {};
	\end{pgfonlayer}
	\begin{pgfonlayer}{edgelayer}
		\draw [style=x-wire, bend right=45] (39.center) to (40);
		\draw [style=x-wire, bend right=45] (40) to (37.center);
		\draw [style=x-wire, bend right=45] (37.center) to (38);
		\draw [style=x-wire, bend right=45] (38) to (39.center);
		\draw [style=x-wire] (38) to (40);
	\end{pgfonlayer}
\end{tikzpicture}.
\end{prop} 

Here, the horizontal wire is a common notational shortcut coming from a direct consequence of the Frobenius equation. Recall that multiplication and comultiplication nodes can be freely rotated in the plane, keeping the inputs/outputs fixed.

\begin{proof}
    This diagram also counts the number of matrix blocks (denoted $K$) in the quantum set $X$, as follows (see Remark \ref{rem:qset_structure} for block decomposition of $X$):
        $$
        \begin{tikzpicture}[scale=0.6]
	\begin{pgfonlayer}{nodelayer}
		\node [style=none] (37) at (0, 0.5) {};
		\node [style=x-frob] (38) at (-0.45, 0) {};
		\node [style=none] (39) at (0, -0.5) {};
		\node [style=dark-x-frob] (40) at (0.45, 0) {};
	\end{pgfonlayer}
	\begin{pgfonlayer}{edgelayer}
		\draw [style=x-wire, bend right=45] (39.center) to (40);
		\draw [style=x-wire, bend right=45] (40) to (37.center);
		\draw [style=x-wire, bend right=45] (37.center) to (38);
		\draw [style=x-wire, bend right=45] (38) to (39.center);
		\draw [style=x-wire] (38) to (40);
	\end{pgfonlayer}
\end{tikzpicture} = \sum_{k=1}^K \left(\frac{1}{\sqrt{n_k}}\right)^2 \begin{tikzpicture}[scale=0.7]
	\begin{pgfonlayer}{nodelayer}
		\node [style=none] (0) at (-1, 0.75) {};
		\node [style=none] (1) at (-0.625, 0.75) {};
		\node [style=none] (6) at (-1, -0.25) {};
		\node [style=none] (8) at (-0.75, -0.25) {};
		\node [style=none] (9) at (-0.25, -0.25) {};
		\node [style=none] (11) at (-0.5, 0.25) {};
		\node [style=none] (12) at (0.875, 0.5) {};
		\node [style=none] (16) at (0.75, -0.5) {};
		\node [style=none] (17) at (1, -0.5) {};
		\node [style=none] (21) at (0, 1) {};
		\node [style=none] (22) at (0, 1.25) {};
		\node [style=none] (23) at (1.125, 0.5) {};
		\node [style=none] (24) at (0, 0) {};
		\node [style=none] (25) at (0, -1) {};
		\node [style=none] (26) at (0, -0.75) {};
	\end{pgfonlayer}
	\begin{pgfonlayer}{edgelayer}
		\draw [style=hilbert] (25.center)
			 to [in=-60, out=180] (6.center)
			 to [bend left, looseness=0.75] (0.center)
			 to [in=180, out=60] (22.center);
		\draw [style=hilbert] (22.center)
			 to [in=90, out=0] (23.center)
			 to [in=-45, out=-90, looseness=0.75] (9.center);
		\draw [style=hilbert, in=-45, out=135, looseness=1.25] (9.center) to (11.center);
		\draw [style=hilbert] (11.center)
			 to [in=120, out=150, looseness=2.25] (8.center)
			 to [in=-180, out=-60] (26.center);
		\draw [style=hilbert] (26.center)
			 to [in=-105, out=0, looseness=0.75] (16.center)
			 to [in=-90, out=75, looseness=0.75] (12.center)
			 to [in=0, out=90, looseness=0.75] (21.center);
		\draw [style=hilbert] (21.center)
			 to [in=90, out=180, looseness=0.75] (1.center)
			 to [in=90, out=-90, looseness=0.75] (24.center)
			 to [in=90, out=-90, looseness=0.75] (17.center)
			 to [in=0, out=-90, looseness=0.75] (25.center);
	\end{pgfonlayer}
\end{tikzpicture} = \sum_{k=1}^K \frac{1}{n_k}\ \begin{tikzpicture}[scale=0.6]
	\begin{pgfonlayer}{nodelayer}
		\node [style=none] (4) at (0, -0.5) {};
		\node [style=none] (5) at (0, 0.5) {};
		\node [style=none] (6) at (0.5, 0) {};
		\node [style=none] (7) at (-0.5, 0) {};
	\end{pgfonlayer}
	\begin{pgfonlayer}{edgelayer}
		\draw [style=hilbert] (4.center)
			 to [bend left=45] (7.center)
			 to [bend left=45] (5.center);
		\draw [style=hilbert] (5.center)
			 to [bend left=45] (6.center)
			 to [bend left=45] (4.center);
	\end{pgfonlayer}
\end{tikzpicture}
        = \sum_{k=1}^K \frac{\dim \mathbb C^{n_k}}{n_k} 
        =\sum_{k=1}^K 1
        = K
        $$
\end{proof}

We can now make a key observation that follows directly from applying the definition of classical dimension to the statement of Proposition \ref{prop:classical_sharing_game}.

\begin{cor}[of Proposition \ref{prop:classical_sharing_game}]
    A quantum set $X$ is classical if and only if its (quantum) dimension is equal to its classical dimension.
\end{cor}

\subsection*{Structure of synchronous correlations}

\begin{lemma} \label{lem:loop}
    A synchronous correlation $P:X\otimes \op X \to A \otimes \op A$ can ``slide upwards'' off a loop:
    \ctikzfig{loop\_slide\_sync}
    
    A synchronous (combined) strategy $\varphi: X\otimes \op X\to_{\cl H} A \otimes \op A$ can ``slide upwards'' off a loop, leaving behind its resource wire:
    \ctikzfig{loop\_slide\_sync\_strat}
\end{lemma}

\begin{proof}
    Beginning with the definition of a synchronous correlation, we make a series of graphical derivations.
    \ctikzfig{pf\_loop\_slide}
    The steps are justified as follows.
    \begin{enumerate}
    \item Apply units and counits.
    \item Contract the units and counits using Lemma \ref{lem:sharing}, and use the fact that $P$ is counital.
    \item Contract the remaining counit.
    \item Rewrite using single wire notation, and use specialness of $X$.
    \end{enumerate}
    All of the above steps can be applied to a synchronous strategy $\varphi$ instead of a correlation $P$, by adding a resource Hilbert space wire through the map. Then, in step 2, on the right hand side, we use the counitality equation \eqref{eq:q_counital} for quantum functions, and the resource wire remains in the diagram.
\end{proof}

\begin{remark}
    Notice that by defining $\Hat{P}$ to be an operator on $X\otimes A$, given by the partial transpose of $P$
    \ctikzfig{P\_hat}
    Lemma \ref{lem:loop} says that for $P$ synchronous, $\mathrm{tr}\Hat P = \dim X$ where the trace of $\Hat P$ is given by
    \ctikzfig{P\_hat\_trace}
\end{remark}

\begin{prop} \label{prop:marg}
    A non-signalling correlation $P:X\otimes \op X \to A\otimes \op A$ has marginals given by
    \ctikzfig{marg}
\end{prop}

\begin{proof}
    We give a diagrammatic proof for $P_A$
    \ctikzfig{pf\_marg}
    with the steps justified as follows:
    \begin{enumerate}
        \item $P$ is nonsignalling, and
        \item comultiplication is counital.
    \end{enumerate}
    Flip the proof horizontally for $P_B$.
\end{proof}

 Note that conjugating the above statement allows us to use the opposite quantum set instead, as $P_A$ and $P_B$ are channels and so are self-conjugate. When Alice and Bob have the same sets, we define the swapped correlation $P^\sigma := \sigma P \sigma$, corresponding to Alice and Bob switching places. When $P$ is non-signalling, $P^\sigma$ is also non-signalling, and $P_B = P^\sigma_A$. If additionally, $P$ is synchronous, we get $P_A = P_B$, as follows.

\begin{prop} \label{prop:marg_sync}
    A non-signalling synchronous correlation $P$ has equal marginals given by
    \ctikzfig{marg\_sync}
\end{prop}

\begin{proof}
    We give a diagrammatic proof
        \ctikzfig{pf\_marg\_sync\_empty}
    with the steps justified as follows:
    \begin{enumerate}
        \item Apply Proposition \ref{prop:marg}.
        \item $P$ is synchronous.
        \item Contract loose counit.  
    \end{enumerate}
\end{proof}

Note that conjugating the statement of the above proposition replaces all multiplication and comultiplication with that of the opposite algebra. So we also have
    \ctikzfig{marg\_sync\_filled}

In the classical question and answer case, consider the quantum processes interpretation (see \cite{PQP}) of the above proposition: Each of Alice's and Bob's marginals looks like encoding their own question as a quantum state, applying $P$, and then measuring to get back classical data.

If the non-signalling correlation $P$ is realized by some quantum commuting strategy and a shared state, we can say even more about the strategy. This is a diagrammatic alternative to the proof of \cite[Theorem 6.6, (1)]{BKS21}.

\begin{thm}\label{thm:sync_qc_corr}
If a synchronous correlation is realized by a quantum commuting strategy $(E,F)$ on a finite-dimensional resource space $\cl H$, and a normalized state $\ket\psi \in \cl H$, then
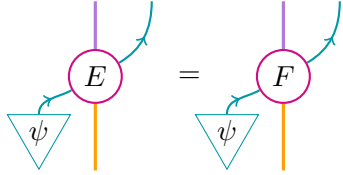
\ctikzfig{equal\_qc} 
\end{thm}

Before proving this result, we will present a diagrammatic version of the Cauchy-Schwarz inequality for $\textbf{FHilb}$. (Note that the diagrams being compared in the following equation are both non-negative scalars.)

\begin{lemma}[Cauchy-Schwarz]\label{lem:cs_ineq}
    Let $f, g: \cl K \to \cl H$ in $\textbf{FHilb}$. Then
    
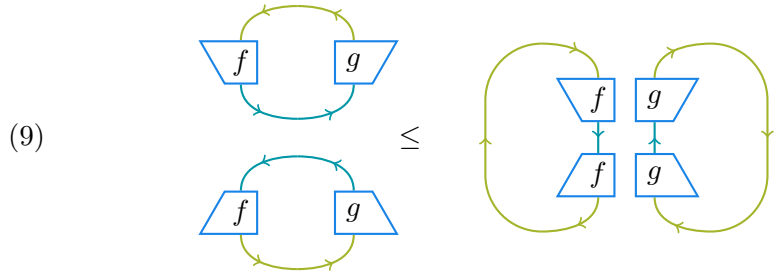
\begin{equation}\label{eq:cs_ineq}
    \begin{tikzpicture}
	\begin{pgfonlayer}{nodelayer}
		\node [style=none] (38) at (-8.5, 0.25) {};
		\node [style=none] (34) at (-8.5, -1.75) {};
		\node [style=transpose] (0) at (-4.5, 0.5) {$f$};
		\node [style=conj] (9) at (-4.5, -0.5) {$f$};
		\node [style=dagger] (14) at (-3.75, 0.5) {$g$};
		\node [style=none] (15) at (-2.25, 0.5) {};
		\node [style=none] (16) at (-3, 1.25) {};
		\node [style=trap] (17) at (-3.75, -0.5) {$g$};
		\node [style=none] (18) at (-2.25, -0.5) {};
		\node [style=none] (19) at (-3, -1.25) {};
		\node [style=none] (22) at (-8.5, 1.75) {};
		\node [style=dagger] (27) at (-7.75, 1) {$g$};
		\node [style=transpose] (31) at (-9.25, 1) {$f$};
		\node [style=none] (35) at (-8.5, -0.25) {};
		\node [style=conj] (36) at (-9.25, -1) {$f$};
		\node [style=trap] (37) at (-7.75, -1) {$g$};
		\node [style=none] (39) at (-7, 0) {$\leq$};
		\node [style=none] (43) at (-6, -0.5) {};
		\node [style=none] (44) at (-5.25, -1.25) {};
		\node [style=none] (45) at (-6, 0.5) {};
		\node [style=none] (46) at (-5.25, 1.25) {};
	\end{pgfonlayer}
	\begin{pgfonlayer}{edgelayer}
		\draw [style=hilbert2, in=180, out=90] (14) to (16.center);
		\draw [style=hilbert2] (16.center)
			 to [in=90, out=0] (15.center)
			 to (18.center)
			 to [in=0, out=-90] (19.center);
		\draw [style=hilbert] (17) to (14);
		\draw [style=hilbert2, in=0, out=90] (27) to (22.center);
		\draw [style=hilbert2, in=90, out=180] (22.center) to (31);
		\draw [style=hilbert2, in=-180, out=-90] (36) to (34.center);
		\draw [style=hilbert2, in=-90, out=0] (34.center) to (37);
		\draw [style=hilbert, in=0, out=90] (37) to (35.center);
		\draw [style=hilbert, in=90, out=180] (35.center) to (36);
		\draw [style=hilbert] (0) to (9);
		\draw [style=hilbert2, in=-90, out=-180] (19.center) to (17);
		\draw [style=hilbert2] (44.center)
			 to [in=-90, out=180] (43.center)
			 to (45.center)
			 to [in=-180, out=90] (46.center);
		\draw [style=hilbert2, in=0, out=-90] (9) to (44.center);
		\draw [style=hilbert2, in=90, out=0] (46.center) to (0);
		\draw [style=hilbert, in=180, out=-90] (31) to (38.center);
		\draw [style=hilbert, in=-90, out=0] (38.center) to (27);
	\end{pgfonlayer}
\end{tikzpicture}
    \end{equation}
    with equality if and only if $f=\alpha g$ for some $\alpha \in \bb C$.
\end{lemma}

\begin{proof}
    By sliding the left-side $f$s along the $\cl H$ wire to the other side of their loops, this just says $\langle f, g \rangle \langle g, f \rangle \leq \langle f, f \rangle \langle g, g \rangle $, which is the Cauchy-Schwarz inequality for the Hilbert-Schmidt inner product, and it holds with equality exactly when $f$ is a scalar multiple of $g$.
\end{proof}

The Cauchy-Schwarz inequality in this format leads to two inequalities for pairs of quantum functions.

\begin{thm}\label{thm:cs_ineq_qfunc}
    Let $E:X\to _\cl H A, F:\op X\to _\cl H \op A$ be quantum functions, with $\cl H$ finite-dimensional.
    \begin{enumerate}
    \item  Let $\ket \psi \in \cl H$ be a normalized state.
        \ctikzfig{cs\_ineq\_qfunc\_state}
        with equality if and only if
        \ctikzfig{equal\_qc}
     \item We have
        \ctikzfig{cs\_ineq\_qfunc}
        with equality if and only if
        \ctikzfig{equal\_qc\_strat}
    \end{enumerate}
\end{thm}

See \ref{pf:cs_ineq_qfunc} for proof. With these tools in hand, we are equipped to prove the previous theorem, \ref{thm:sync_qc_corr}.

\begin{proof}[Proof of Theorem \ref{thm:sync_qc_corr}]
Because $(E,F)$ is quantum commuting, we have that
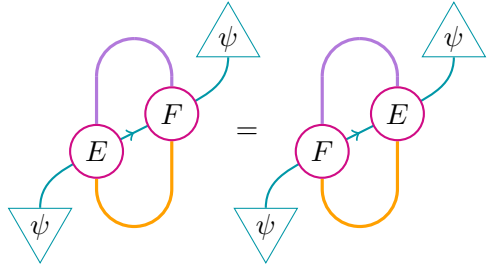
\ctikzfig{loop\_EF\_equals\_FE}
By the loop-sliding lemma \ref{lem:loop}, given that the realized correlation is synchronous, these are also equal to
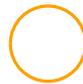
\ctikzfig{loop}
Then, we have equality in Theorem \ref{thm:cs_ineq_qfunc} (1), implying the desired equality.
\end{proof}

\begin{thm}\label{thm:sync_qc_strat}
    If a quantum commuting strategy $(E,F)$ on a finite-dimensional resource space $\cl H$ is synchronous, then 
    \ctikzfig{equal\_qc\_strat}
\end{thm}

\begin{proof}
Because $(E,F)$ is quantum commuting, we have that
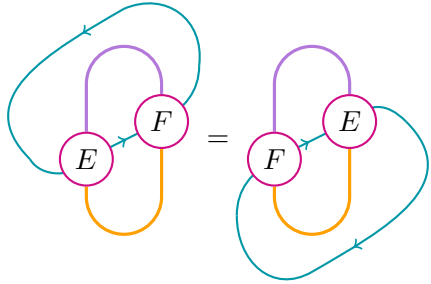
\ctikzfig{double\_loop\_EF\_equals\_FE}

By the loop-sliding lemma \ref{lem:loop}, given that the strategy is synchronous, these are also equal to
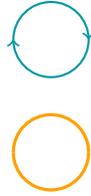
\ctikzfig{double\_loop}

Then, we have equality in Theorem \ref{thm:cs_ineq_qfunc}, (2), implying the desired equality.
    
\end{proof}

The consequence $E=F$ of Theorem \ref{thm:sync_qc_strat} may be surprising. It shows us that synchronicity at the level of strategies is an extremely strong property -- much stronger than synchronicity at the level of correlations. In a sense, it results in an $n$ dimensional restriction on the players' behaviour instead of a $1$ dimensional restriction.

As a corollary, we will show that synchronous quantum commuting strategies are essentially families of \textit{local} strategies. Over an $n$-dimensional Hilbert space, these strategies are merely recipes for convex interpolation between $n$ deterministic strategies. The specific local strategy implemented varies depending on the selected resource state. A more precise statement follows.

\begin{cor}[Corollary of Theorem \ref{thm:sync_qc_strat}]\label{cor:sync_qc}
Let $(E,F)$ be a synchronous quantum commuting strategy on an $n$-dimensional resource space $\cl H$. Then $(E,F)$ together with any state $\psi$ of $\cl H$ can be realized by a convex combination of $n$ synchronous deterministic strategies.
\end{cor}

\begin{proof}
    By Theorem \ref{thm:sync_qc_strat}, we have $E=F$. Fix a self-conjugate basis of each of $X$ and $A$ (this translates to being self-adjoint when thought of as operators). For each choice of basis elements $x,y$ of $X$ and $a,b$ of $A$, composing these with the quantum commuting condition on $(E,E)$ gives a condition $$E_{y,b} E_{x,a} = E_{x,a} E_{y,b},$$ where $E_{x,a}:\cl H \to \cl H$ comes from $E$ by plugging $x$ into the $X$ wire and $a^\dagger$ into the $A$ wire. Note also that each $E_{x,a}$ satisfies
    $$E_{x,a}^\dagger = E_{x,a}$$
    due to the realness property of the quantum function $E$ and due to self-conjugacy of the basis elements $x,a$.
    
    So we have a family of pairwise commuting self-adjoint operators, $\{E_{x,a}\}$. Fix an orthonormal basis $\{\ket i: i \in [n]\}$ of $\cl H$ in which this family of operators is diagonal. Define $f_i:= \bra i E \ket i$, with a slight abuse of notation. It follows that each $f_i$ is a (classical) quantum function both $X\to_{\bb C} A$ and $\op X\to_{\bb C} \op A$ so $(f_i, f_i)$ is a deterministic strategy for the same game.
    
    Now let $Q$ be the classical quantum set structure on $\cl H$ coming from the specified orthonormal basis. Let $\ket \alpha \in \cl H$ be given by $$\ket \alpha := \sum_{i=1}^n \left|\braket{\psi}{i}\right|^2 \ket i.$$ Then the local strategy given by $\sum_{i=1}^n \braket{\alpha}{i} f_i \otimes f_i$ realizes the same correlation as $(E,E)$ with $\ket \psi$.
\end{proof}

\begin{remark}[Synchronous correlations versus synchronous strategies]
    This result does not conflict with the existence of synchronous correlations (between classical sets) that are quantum commuting but not local. The key difference here is that being a synchronous \emph{strategy} is much stronger (and therefore more restrictive) than being a synchronous \emph{correlation}. To be a synchronous \emph{strategy} means here to be synchronous in a state-agnostic way. To make the proper analogies to the classical question-and-answer literature, one must talk about synchronous correlations implemented by quantum commuting strategies, as opposed to synchronous quantum commuting strategies. Remark \ref{rmk:perf_strat} makes a related observation.
\end{remark}

\section{Bisynchronicity}
\label{sec:bisync}
Bisynchronous games and correlations are introduced in \cite{PR21}. We extend these definitions to the quantum question and answer setting.

In the classical question and answer setting, the graph isomorphism game is an example of a bisynchronous game. In Section \ref{sec:graph_games}, we discuss a quantum generalization of the graph isomorphism game, and use the fact that it is bisynchronous to relate its perfect strategies and correlations to quantum graph isomorphisms.

\subsection*{Bicorrelations and bistrategies}
\begin{defn}[Bicorrelation]
We say a correlation $P$ is a \emph{bicorrelation} if $P^\dagger$ is also a correlation.
\end{defn}

A bicorrelation is simply a unital correlation.

\begin{defn}[Bistrategy]
We say a strategy $\varphi:X\otimes Y \to_{\cl H} A\otimes B$ is a bistrategy if 
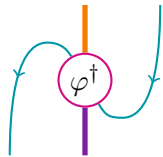
\ctikzfig{costrategy} 
is also a strategy $\overline\varphi: A\otimes B \to_{\cl H^*} X\otimes Y$.
\end{defn}

This is equivalent to defining a bistrategy $\varphi$ to be a quantum bijection as given in \cite[Definition 4.3]{MRV}. From the perspective of C$^*$-algebras, this can be thought of as an entanglement assisted $*$-isomorphism.

Note that correlations realized by bistrategies are automatically bicorrelations.

It turns out bistrategies have the desirable property of $\overline\varphi$ inheriting the strategy class of $\varphi$, as follows.

\begin{prop} \label{prop:bistrat_structure}
    Let $\varphi:X\otimes Y \to A \otimes B$ be a bistrategy. Let $\dim X = \dim A$.
    \begin{itemize}
        \item If $\varphi$ is quantum commuting implemented by $(E, F)$, then $\overline \varphi$ is quantum commuting implemented by $(\overline E, \overline F)$.
        \item If $\varphi$ is quantum tensor implemented by $(E, F)$, then $\overline \varphi$ is quantum tensor implemented by $(\overline E, \overline F)$.
        \item If $\varphi$ is deterministic implemented by $(E,F)$, then $\overline \varphi$ is deterministic implemented by $(E^\dagger, F^\dagger)$.
    \end{itemize}
\end{prop}

The proof is given in \ref{pf:bistrat_structure}. In particular, the proof begins with the fact that because $\varphi$ is a bistrategy, it automatically holds that $\dim X \dim Y = \dim A \dim B$.  However, there is no guarantee that $\dim X = \dim A$, as Alice's operator itself is not a priori a quantum bijection. The assumption $\dim X = \dim A$ ensures Alice and Bob's operators are quantum functions exactly (instead of the equations holding up to a constant factor). An alternate assumption could be made instead, namely that the players' individual operators are quantum bijections.

Note also that the apparent difference of implementation of $\overline \varphi$ in the deterministic case, with daggers instead of bars, is only due to the fact that when a quantum function is classical, the dagger and the bar coincide.

\subsection*{Bisynchronous maps}

\begin{defn}[Bisynchronous map]
We say a linear map $f:X\otimes \op X \to A\otimes \op A$ is \emph{cosynchronous} if $f^\dagger$ is synchronous. We say $f$ is \emph{bisynchronous} if $f$ is both synchronous and cosynchronous.
\end{defn}

\begin{defn}[Bisynchronous quantum function]
     We say a quantum function $\varphi:X\otimes \op X \to_{\cl H} A \otimes \op A$ is \emph{bisynchronous} if both $\varphi$ and $\overline\varphi$ are synchronous.
 \end{defn}

\begin{prop}\label{prop:bisync_sharing}
    A map or quantum function is bisynchronous if and only if it preserves sharing.
\end{prop}

\begin{proof}
    Let $f:X\otimes \op X\to A\otimes \op A$ be a bisynchronous linear map. If $f$ is bisynchronous,
    \ctikzfig{pf\_bisync\_sharing}
    where the first equality is synchronicity of $f$, and the second is synchronicity of $f^\dagger$.
    Let $\varphi:X\otimes \op X \to_{\cl H} A \otimes \op A$ be a bisynchronous quantum function. Then the same proof holds with an $\cl H$ wire tacked on.

    Conversely, if $f$ preserves sharing, then $f$ and $f^\dagger$ are synchronous because sharing is a projection. The same goes for $\varphi$ and $\overline\varphi$.
\end{proof}

In the quantum commuting case, there is something redundant about the definition of a bisynchronous bistrategy, as follows.

\begin{prop}\label{prop:bisync_qc}
    Let $\varphi$ be a synchronous quantum commuting strategy.
    \begin{enumerate}
        \item If $\varphi$ is a bistrategy, then $\varphi$ is bisynchronous.
        \item If $\varphi$ is bisynchronous and unital -- in the sense that $\overline \varphi$ satisfies \eqref{eq:q_counital} -- then $\varphi$ is a bistrategy.
    \end{enumerate}
\end{prop}

\begin{proof}
    By Theorem \ref{thm:sync_qc_strat}, $\varphi:X \otimes \op X \to_{\cl H} A \otimes \op A$ is given by some quantum function $E:X\to_{\cl H} A$ applied twice.
    \begin{enumerate}
        \item If $\varphi$ is a bistrategy then applying \eqref{eq:q_comult} to $E$ and $\overline E$, we see that $\varphi$ commutes with sharing and is therefore bisynchronous.
        \item By assumption, we have that $\overline \varphi$ satisfies \eqref{eq:q_counital}. We know that $\overline \varphi$ satisfies \eqref{eq:self_conj} because $\varphi$ does and the condition is the same. To see that $\overline\varphi$ satisfies \eqref{eq:q_comult}, given that $\varphi$ is bisynchronous, start with the result of Proposition \ref{prop:bisync_sharing}, and apply a counit to the left $A$ wire of the equation. By counitality of $E$, we get that \eqref{eq:q_comult} holds for $E$, and then it follows for $\varphi$ as well. Therefore, $\varphi$ is a bistrategy.
    \end{enumerate}
\end{proof}

\begin{thm}
\label{thm:bisync_dims_equal}
Let $X, A$ be quantum sets.
\begin{enumerate}
    \item If $P:X\otimes \op X \to A \otimes \op A$ is a bisynchronous bicorrelation, then $\dim X = \dim A$.
    \item If additionally, $\varphi: X\otimes \op X \to_{\cl H} A \otimes \op A$ is a bisynchronous bistrategy, then $X$ and $A$ have the same classical dimension. \item If furthermore, $\varphi$ is quantum commuting, then $X,A$ are quantum isomorphic as quantum sets (in the sense that there is a quantum bijection between them).
    \item If furthermore, $\varphi$ is deterministic, then $X,A$ are isomorphic as quantum sets.
\end{enumerate}
\end{thm}

\begin{proof}
Because $P$ is a synchronous correlation, the loop-sliding lemma \ref{lem:loop} gives us that
\ctikzfig{loop\_slide\_sync}
Because $P^\dagger$ is a synchronous correlation, by applying the loop-sliding lemma to $P^\dagger$ and then taking daggers (turning diagrams upside down), we have that
\ctikzfig{loop\_slide\_cosync}
Therefore, the two circles are equal. In a quantum set, the circle (given by $\varepsilon m \Delta u$) computes the dimension of the underlying Hilbert space.

In the case that $\varphi$ is a bisynchronous bistrategy, then using the fact that $\varphi$ and $\overline{\varphi}$ are both quantum functions, and fixing any normalized state of $\cl H$,
\ctikzfig{bisync\_preserves\_classical\_dim}

In the above, the steps are justified as follows:
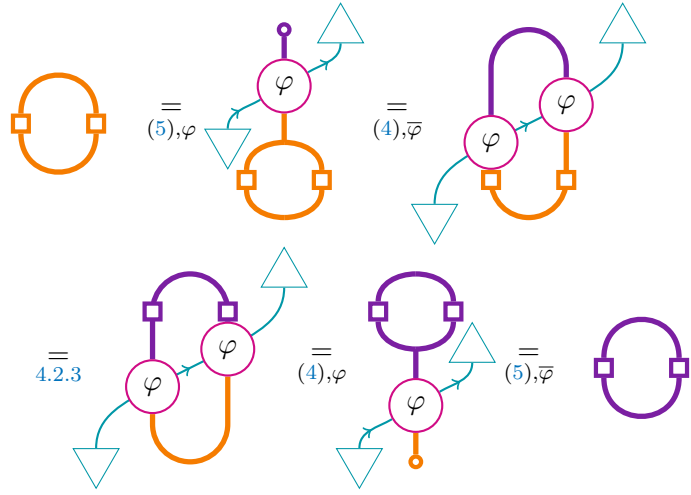
\begin{enumerate}
    \item Apply the counitality equation \eqref{eq:q_counital} for $\varphi$.
    \item Apply the comultiplicativity equation \eqref{eq:q_comult} for $\overline{\varphi}$.
    \item Apply Proposition \ref{prop:bisync_sharing} to interchange $\varphi$ with the sharing map on both sides.
    \item Apply the comultiplicativity equation \eqref{eq:q_comult} for $\varphi$.
    \item Apply the counitality equation \eqref{eq:q_counital} for $\overline{\varphi}$.
\end{enumerate}

If, furthermore, $\varphi$ is quantum commuting, then because it is synchronous, it is realized by two copies of some quantum function $E$, as shown in Theorem \ref{thm:sync_qc_strat}. We have that $\varphi$ is a bistrategy and $\dim X = \dim A$, so by Proposition \ref{prop:bistrat_structure}, $\overline{E}$ is also a quantum function. Therefore, $E$ is a quantum bijection $X\to_{\cl H} A$.

In the case that $\varphi$ is deterministic, $\cl H = \mathbb{C}$ and the quantum bijection $E$ just becomes a quantum-set isomorphism.

\end{proof}

\begin{cor}[Corollary of Theorem \ref{thm:bisync_dims_equal}]
Consider a bisynchronous game $\lambda$.
\begin{enumerate} 
    \item If $\lambda$ has a perfect bicorrelation, then the question and answer sets have the same dimension.
    \item If $\lambda$ has a perfect bistrategy, then the question and answer sets have the same classical dimension.
    \item If $\lambda$ has a perfect quantum commuting bistrategy, then the question and answer sets are quantum isomorphic.
    \item If $\lambda$ has a perfect deterministic bistrategy, then the question and answer sets are isomorphic as quantum sets.
\end{enumerate}
\end{cor}

\begin{proof}
    A perfect correlation $P$ for a synchronous game $\lambda$ is synchronous by Theorem \ref{thm:perf_sync}. We also know that $P^\dagger$ is perfect for the game $\lambda^\dagger$ by taking daggers of Definition \ref{defn:perf_corr}. Because $\lambda^\dagger$ is also synchronous, we can apply Theorem \ref{thm:perf_sync} again to see that $P^\dagger$ is synchronous. Therefore, $P$ is bisynchronous, and so we can conclude that $\dim X = \dim A$. Making the same argument to show that a perfect strategy is also  forced to be bisynchronous, we can apply the rest of Theorem \ref{thm:bisync_dims_equal} to show the remaining items.
\end{proof}

\section{Games on quantum graphs} \label{sec:graph_games}
As an example of how to apply our graphical theory of games, we construct the quantum graph homomorphism game, generalizing the (classical) graph homomorphism game of \cite{MR16} to allow the question and answer spaces to be quantum graphs. We prove that it is a synchronous game, and we show that quantum graph homomorphisms realize perfect correlations for this game.

Note that other quantisations of the graph homomorphism game exist. There is a more restrictive, quantum-to-classical graph homomorphism game given in \cite{BGH22}. There is also a fully quantum graph homomorphism game given in \cite{TT24}. It remains to be seen if that game is equivalent to the one given here.

We also define the very similar quantum graph isomorphism game, generalizing the (classical) graph isomorphism game of \cite{AMRSSV}, later studied by 
\cite{BCEHPSW}. This game appears here for the first time. We prove that it is a bisynchronous game and that quantum graph isomorphisms realize perfect bicorrelations for this game.

We follow the quantum homomorphism/isomorphism game results with partial converses for the deterministic setting, in that perfect deterministic correlations \emph{only} come from homomorphisms/isomorphisms.

For background on quantum graphs that informs our approach, see \cite{MRV}. We include the definition here.

\begin{defn}[Quantum graph]
    Given a quantum set $X$, a \emph{quantum graph} is a self-adjoint linear map $G:X\to X$ satisfying 
    \begin{align*}
        &\begin{tikzpicture}
	\begin{pgfonlayer}{nodelayer}
		\node [style=game] (2) at (1, 0) {$G$};
		\node [style=none] (5) at (-1.75, 1) {};
		\node [style=none] (6) at (-1.75, -1) {};
		\node [style=none] (7) at (1, -1) {};
		\node [style=none] (8) at (1, 1) {};
		\node [style=none] (9) at (0, 0) {$=$};
		\node [style=none] (18) at (-1.75, 0.75) {};
		\node [style=game] (21) at (-1.25, 0) {$G$};
		\node [style=none] (22) at (-1.75, -0.75) {};
		\node [style=game] (23) at (-2.25, 0) {$G$};
	\end{pgfonlayer}
	\begin{pgfonlayer}{edgelayer}
		\draw [style=x-wire] (8.center) to (2);
		\draw [style=x-wire] (2) to (7.center);
		\draw [style=x-wire, in=180, out=-90] (23) to (22.center);
		\draw [style=x-wire] (6.center) to (22.center);
		\draw [style=x-wire, in=180, out=90] (23) to (18.center);
		\draw [style=x-wire] (18.center) to (5.center);
		\draw [style=x-wire, in=0, out=90] (21) to (18.center);
		\draw [style=x-wire, in=0, out=-90] (21) to (22.center);
	\end{pgfonlayer}
\end{tikzpicture}& \textrm{and}\qquad & \begin{tikzpicture}
	\begin{pgfonlayer}{nodelayer}
		\node [style=none] (31) at (-1, 0.5) {};
		\node [style=none] (33) at (-0.5, -0.5) {};
		\node [style=game] (36) at (0, 0) {$G$};
		\node [style=none] (37) at (1, -0.5) {};
		\node [style=none] (38) at (0.5, 0.5) {};
		\node [style=none] (39) at (1, -1) {};
		\node [style=none] (41) at (-1, 1) {};
		\node [style=none] (43) at (-3, 1) {};
		\node [style=game] (44) at (-3, 0) {$G$};
		\node [style=none] (46) at (-3, -1) {};
		\node [style=none] (47) at (-2, 0) {$=$};
	\end{pgfonlayer}
	\begin{pgfonlayer}{edgelayer}
		\draw [style=x-wire] (39.center) to (37.center);
		\draw [style=x-wire] (37.center)
			 to [in=0, out=90] (38.center)
			 to [in=60, out=-180] (36);
		\draw [style=x-wire] (36)
			 to [in=0, out=-120] (33.center)
			 to [in=-90, out=-180] (31.center)
			 to (41.center);
		\draw [style=x-wire] (46.center) to (44);
		\draw [style=x-wire] (44) to (43.center);
	\end{pgfonlayer}
\end{tikzpicture}
    \end{align*}
    We say a quantum graph $X$ is \emph{irreflexive} if
    \ctikzfig{irreflexive}
\end{defn}

The quantum set $X$ can be thought of as a quantum vertex set. Notice that the definition of a quantum graph looks similar to the definition of a game, plus self-adjointness. In fact, we can equivalently define a quantum graph as a self-adjoint one-player game on $V$, by treating Bob as a trivial player having questions and answers given by the classical one-element set, and using the fact that for a quantum set $X$, $X \cong X \otimes \bb C$.

In order to work more easily with quantum graphs, we introduce thick wire notation for a quantum graph as a state (or effect) of $X\otimes \op X$, as well as notation for an equivalent formulation as a projection onto the ``edges''. We denote
    \ctikzfig{graph\_thick\_wire}
and
    \ctikzfig{graph\_projection}
which are both well-defined due to self-conjugacy of $G$. The notation for $P_G$ should evoke its self-adjointness but non-self-conjugacy.

Another way of thinking about $P_G$ is as the right-multiplication by the graph state $\ket G$. Indeed,
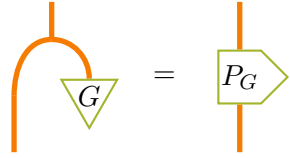
\ctikzfig{right\_mult\_G}

As a first example, consider the \emph{reflexive} quantum graph given by $G= \mathrm{Id}_X$ with $P_G$ given by sharing.

Note that the idempotency condition in the definition of a quantum graph is equivalent to
    \ctikzfig{graph\_thick\_comult}
Also note that irreflexivity of $G$ can be equivalently defined using sharing by
    \ctikzfig{thick\_irreflexive}

\subsection*{The quantum graph homomorphism game}
\label{sec:hom_game}
\begin{defn}
    Given irreflexive quantum graphs $G:X\to X$ and $H:A\to A$, the \emph{quantum graph homomorphism game} is denoted $\lambda_{G\to H}:X\otimes \op X \to A\otimes \op A$ and given by
    \ctikzfig{qgraph\_hom\_game}
    or equivalently, using thick wires,
    \ctikzfig{qgraph\_hom\_game\_thick}
\end{defn}

We work here with irreflexive quantum graphs for simplicity, but an analogous definition should be possible for the reflexive case.

In order for the above to be a reasonable definition, we should check that it is actually a game. 

\begin{prop}\label{prop:graph_hom_game}
    The quantum graph homomorphism game is a game.
\end{prop}

The proof is in \ref{pf:graph_hom_game}.

In the case of classical questions and answers, $\lambda_{G\to H}$ reduces to the graph homomorphism game between classical graphs, as introduced in \cite{MR16}. The rule function in the classical case becomes
\begin{align*}
    \lambda_{G\to H}(a,b|x,y) &= \bra a H \ket b \bra x G \ket y + \langle a | b \rangle \langle x | y \rangle + 1 - \bra x G \ket y - \langle x | y \rangle
\end{align*}
If $x=y$, then because $G$ is irreflexive, we get
$$\lambda_{G\to H}(a,b|x,y) = \langle a | b \rangle = \mathbbm{1}_{a=b}.$$
If $x\sim_G y$, we have
$$\lambda_{G\to H}(a,b|x,y) = \bra a H \ket b = \mathbbm{1}_{a\sim_H b}.$$
Finally, if $x,y$ are unrelated by $G$,
$$\lambda_{G\to H}(a,b|x,y) =  1.$$

\begin{exam}[Quantum colouring a quantum graph]
    Let $G$ be any irreflexive quantum graph on $X$. Let $K_A$ be the irreflexive quantum graph on $A$ given by the diagram
    \ctikzfig{complete\_graph}
    Then the game $\lambda_{G\to K_A}$ can be thought of as a quantum colouring game for $G$, and if the quantum vertex set $A$ is classical of dimension $n$, then this simplifies to the classical $n$-colouring game for $G$.
\end{exam}

\begin{prop} \label{prop:perf_strat_hom}
    The quantum graph homomorphism game is synchronous.
\end{prop}

\begin{proof} We evaluate both sides of the synchronicity equation. Let $\lambda := \lambda_{G\to H}$. 
By irreflexivity of $G$, 
    \ctikzfig{pf\_irref\_sharing\_zero}
We have
    \ctikzfig{pf\_graph\_hom\_sync\_1}
and we also have
    \ctikzfig{pf\_graph\_hom\_sync\_2}
so $\lambda$ is synchronous.
\end{proof}

One would hope that the quantum graph homomorphism game is connected to some notion of homomorphisms between quantum graphs, as is the case for classical graphs.

In the case of classical graphs, the graph homomorphism game was defined first, and then used to define quantum homomorphisms in \cite{MR16}. Quantum homomorphisms between classical graphs are defined for the first time as perfect strategies for the graph homomorphism game.

We don't need to take an analogous approach here, as there are already notions in the literature of a quantum homomorphism between quantum graphs. We restate here the quantum homomorphism definition of \cite{MRV}.

\begin{defn}[Quantum homomorphism of quantum graphs, \cite{MRV}]\label{defn:qgraph_hom}
    Given quantum graphs $G:X\to X, H:A\to A$, a quantum function $E:X\to_\cl K A$ is said to be a \emph{quantum homomorphism} from $G$ to $H$ if
        \ctikzfig{quantum\_graph\_hom}
\end{defn}

We can now show that quantum homomorphisms give Alice and Bob a perfect way to play the quantum graph homomorphism game, in the form of a perfect quantum correlation.

In the remainder of this section, take $G:X\to X$, $H:A\to A$ to be irreflexive quantum graphs. 

\begin{thm}\label{thm:hom_to_perf_qt_corr}
    Given a quantum graph homomorphism $E$ from $G$ to $H$ over $\cl K$, the quantum tensor strategy $E\otimes E_*$ realizes a perfect correlation for $\lambda_{G\to H}$ with respect to the cup state of $\cl K \otimes \cl K^*$.
\end{thm}

In order to prove the theorem, we split it into several smaller results, in the sections to come. The mapping of homomorphisms to strategies, is given by
\ctikzfig{homs\_to\_corrs\_fwd}
with a quantum homomorphism on the left and a quantum tensor strategy realizing a perfect correlation on the right. We denote the strategy employed on the right by $E\otimes E_*$. Note that $E_*$ is equal to the linear map defined from $E$ by transposing only the Hilbert space wire.

The fact that Bob uses the opposite quantum sets to Alice is vital here. To be understood as Bob's operator, $E_*$ is interpreted as a linear map $\cl H ^* \otimes \op X \to \op A \otimes \cl H^*$, using the canonical linear isomorphism from a quantum set to its opposite. The reader may confirm that $E_*: \op X \to _{\cl H^*} \op A$ constitutes a quantum function, inheriting each of the three properties from $E$. This would not be the case without the opposite structure.

\begin{remark}
    The above mapping is one-to-one in the sense that there is an injective linear map from quantum graph homomorphisms to quantum tensor strategies with the given properties. However, note that different quantum tensor strategies may give rise to a single correlation.
\end{remark}

\begin{remark}[Do all perfect q-correlations arise this way?] We would ideally also like to have the converse result, namely that perfect quantum correlations must arise from quantum homomorphisms. We expect this to be the case. Alternatively, perhaps there is a perfect quantum correlation that does not arise from a quantum homomorphism along with a cup state. In that case, our game invites a well-motivated, weaker definition of quantum homomorphism for quantum graphs.
\end{remark}

We begin with the following lemma, allowing us to build a synchronous correlation from two copies of the same quantum function.

\begin{lemma}\label{lem:E_copied_sync}
Let $E:X\to_{\cl H} A$ be a quantum function. Then the quantum tensor strategy $E \otimes E_*$ realizes a synchronous correlation along with the cup state of $\cl H \otimes \cl H^*$.
\end{lemma}

The proof is in \ref{pf:E_copied_sync}.

Next, we demonstrate the structure of perfect correlations for the quantum graph homomorphism game.

\begin{prop}\label{prop:perf_corr_graph_hom}
    Perfect correlations $P$ for $\lambda_{G\to H}$ are exactly the synchronous correlations satisfying
        \ctikzfig{pf\_perf\_corr\_graph\_hom3}    
\end{prop}

\begin{proof}
    Let $P:X\otimes \op X \to A \otimes \op A$ be a perfect correlation for $\lambda_{G\to H}$. 
    Then we have that
        \ctikzfig{pf\_perf\_corr\_graph\_hom1}    
    Cancelling out $P$ from each side and rearranging, we have
        \ctikzfig{pf\_perf\_corr\_graph\_hom2} 
    By Proposition \ref{prop:perf_strat_hom}, $\lambda_{G\to H}$ is synchronous. Then, by Theorem \ref{thm:perf_sync}, $P$ is also  synchronous. Cancelling the terms with sharing,
        \ctikzfig{pf\_perf\_corr\_graph\_hom3}

    Conversely, if $P$ is synchronous and satisfies the above equation, then we can reverse all of the steps to show that $P$ is perfect for $\lambda_{G\to H}$.
\end{proof}

We now have all the tools we need to prove Theorem \ref{thm:hom_to_perf_qt_corr}.

\begin{proof}[Proof of Theorem \ref{thm:hom_to_perf_qt_corr}]
     If $E:X\to _\cl K A$ is a quantum homomorphism from $G$ to $H$, define $\varphi$ to be the quantum tensor strategy where Alice uses $E$ and Bob uses $E_*$. Along with the cup state of $\cl K \otimes \cl K^*$, $\varphi$ realizes a synchronous correlation $P$, as shown in Lemma \ref{lem:E_copied_sync}.

     Because $E$ is a quantum graph homomorphism, we have
        \ctikzfig{quantum\_graph\_hom}
      Tracing out the $\cl K$ wire and sliding the righthand $E$ around the cap, we get
        \ctikzfig{qt\_graph\_strat}
     which is exactly the equation condition of Proposition \ref{prop:perf_corr_graph_hom}.
     
     $P$ satisfies the conditions of Proposition \ref{prop:perf_corr_graph_hom}, so $P$ is a perfect correlation for $\lambda_{G\to H}$.     
\end{proof}

\subsubsection*{The deterministic case gives a partial converse}

While we don't have a general converse to the above, we do when the players use deterministic strategies for the quantum graph homomorphism game. Note that in the following, a classical homomorphism between quantum graphs means a quantum homomorphism over $\bb C$.

\begin{thm}\label{thm:perf_det_to_hom}
    Let $P$ be a perfect correlation for $\lambda_{G\to H}$ arising from a deterministic strategy. Then the strategy is of the form $E\otimes E$ where $E:X\to A$ is a classical homomorphism from $G$ to $H$.
\end{thm}

Note that in the deterministic case, $E = E_*$.

\begin{proof}
    By Proposition \ref{prop:perf_corr_graph_hom}, we have
        \ctikzfig{pf\_perf\_corr\_graph\_hom3}
    
    Denote the deterministic strategy realizing $P$ by $E\otimes F$. Then by Theorem \ref{thm:sync_qc_corr}, $F=E$. 
  
    At the level of the quantum function $E$ and single wires, the fact that $P$ is perfect can be rewritten
    \ctikzfig{det\_graph\_hom} 
    which says exactly that $E$ is a classical homomorphism from $G$ to $H$.
\end{proof}

We may conclude that deterministic strategies for the quantum graph homomorphism game are in one-to-one correspondence with classical homomorphisms.

\subsubsection*{Graph homomorphisms and associated games} 

    The preceding results allow us to connect the quantum graph homomorphism game to other notions of graph homomorphism. 
    
    We present relationships, in Figure \ref{fig:graph_homs}, between the quantum and classical graph homomorphism games, and quantum and classical homomorphisms of quantum and classical graphs. Our quantum graph homomorphism game adds a previously missing vertex to this diagram. The solid arrows are due to the work presented here.

\subsection*{The quantum graph isomorphism game}

It is also possible to similarly construct a quantum graph isomorphism game, with more terms in the rule function. This can be constructed directly by inclusion-exclusion reasoning, in analogy to how the graph homomorphism game was defined, but can also be constructed from graph homomorphism games as a commuting product.

\begin{defn}[Quantum graph isomorphism game]
    Given irreflexive quantum graphs $G:X\to X$ and $H:A\to A$, the \emph{quantum graph isomorphism game} is denoted $\lambda_{G\simeq H}:X\otimes \op X \to A\otimes \op A$ and given by
    $$\lambda_{G\simeq H} := \lambda_{G\to H}\star \lambda_{H\to G}^\dagger.$$
\end{defn}

We quickly check that the commuting product is well defined here. We just need to see that $\lambda_{G\to H}$ and $\lambda_{H\to G}^\dagger$ commute. This follows from the fact that Table \ref{tbl:hom_comm_prod} is symmetric, as these are the building blocks of $\lambda_{G\to H}$ and $\lambda_{H\to G}^\dagger$. The commuting product $\lambda_{G\to H} \star \lambda_{H\to G}^\dagger$ can be computed in either order to give
\ctikzfig{qgraph\_iso\_game\_thick}

\begin{prop}\label{prop:iso_game_bisync}
    The quantum graph isomorphism game is bisynchronous.
\end{prop}

\begin{proof}
    Because $\lambda_{G\to H}$ is synchronous, then by Proposition \ref{prop:prod_sync_sync}, $\lambda_{G\simeq H}$ is also synchronous. Note further that  $\lambda_{G\simeq H}^\dagger = \lambda_{H\simeq G}$ is also synchronous so $\lambda_{G\simeq H}$ is bisynchronous. 
\end{proof}

\begin{remark}[Classical graph isomorphism game]
    \label{rem:cl_iso}
    In the case of classical graphs, $\lambda_{G\simeq H}$ is a $0$-$1$ matrix with entries $\bra a \bra b \lambda_{G\simeq H} \ket x \ket y$ that are $1$ whenever both
    \begin{align*}
        \bra a \bra b \lambda_{G\to H} \ket x \ket y &= 1, \\
        \bra x \bra y \lambda_{H\to G} \ket a \ket b &= 1.
    \end{align*}
    Perfect deterministic bistrategies in the classical case are bijections $E:X\to A$ mapping $x$ to $a$ and $y$ to $b$ only when the above conditions are satisfied. But this is exactly the condition that $E$ is a valid graph isomorphism from $G$ to $H$. This and related results are considered in \cite{LMR20}.

    This game is a simplified version of the graph isomorphism game of \cite{AMRSSV}, without some of the rules that force perfect strategies to come from bijections. Instead, we look to \emph{bistrategies} to get bijections, and only force perfect strategies to intertwine the adjacency matrices.
\end{remark}

We would like to extend Remark \ref{rem:cl_iso} to the quantum game setting, using the definition of quantum isomorphism of quantum graphs, given in \cite{MRV}. The following result is analogous to Theorem \ref{thm:hom_to_perf_qt_corr} from the homomorphism case.

\begin{thm}\label{thm:iso_to_perf_qt_bicor}
    Given a quantum isomorphism $E$ from $G$ to $H$, $E\otimes E_*$ is a quantum tensor bistrategy realizing a perfect bicorrelation (with respect to a cup state) for the quantum graph isomorphism game $\lambda_{G\simeq H}$.
\end{thm}

\begin{proof}
    Given a quantum graph isomorphism $E:G\to H$, we know by Theorem \ref{thm:hom_to_perf_qt_corr} $E\otimes E_*$ is a quantum bistrategy for $\lambda_{G\to H}$ and for $\lambda_{H\to G}^\dagger$ that realises a perfect bicorrelation for each, so by Proposition \ref{prop:perf_product}, $E\otimes E_*$ realises a perfect bicorrelation for the commuting product.
\end{proof}

A converse result to the above theorem would follow from a hypothetical converse to Theorem \ref{thm:hom_to_perf_qt_corr}. At least in the deterministic case, we do have a converse, as follows.

Note that in the following, a classical isomorphism between quantum graphs means a quantum isomorphism, as per \cite{MRV}, over $\bb C$.

\begin{thm}\label{thm:perf_det_to_iso}
    Let $P$ be a perfect bicorrelation for $\lambda_{G\to H}$ given by a deterministic bistrategy. Then the bistrategy is of the form $E\otimes E$ where $E:X\to A$ is a classical isomorphism from $G$ to $H$.
\end{thm}

\begin{proof}
    Given a deterministic bistrategy $E\otimes F$ realizing a perfect bicorrelation $P$ for the quantum graph isomorphism game, we can apply Theorem \ref{thm:perf_det_to_hom} in both directions, as follows.
    We know, by Proposition \ref{prop:perf_product}, that the bicorrelation $P$ is perfect for both $\lambda_{G\to H}$ and $\lambda_{H\to G}^\dagger$. By Theorem \ref{thm:perf_det_to_hom}, $F = E$ and $E$ is a homomorphism $G\to H$.
    We can also conclude that $P^\dagger$ is a perfect correlation for $\lambda_{H\to G}$, so again by Theorem \ref{thm:perf_det_to_hom}, $E^\dagger$ is a homomorphism $H\to G$. 
    
    But then, by combining these two homomorphism equations (and taking the dagger of the second one), we get that $E$ is a quantum bijection (over $\bb C$) satisfying
    \ctikzfig{qgraph\_det\_iso}
    Applying a counit on the top left and a unit on the bottom right, we get
    \ctikzfig{qgraph\_det\_iso1}
    so $E$ is a classical isomorphism from $G$ to $H$.
    
\end{proof}

Similarly to the homomorphism case, we may conclude that deterministic bistrategies for the quantum graph isomorphism game are in one-to-one correspondence with classical isomorphisms.

\section{Other quantum approaches to synchronicity}
\label{sec:other_approaches}

We now consider two other recent generalizations of synchronous games and correlations. We relate these to our definition of synchronicity.

\subsection*{Concurrent correlations}
In \cite{BHTT21, BHTT23}, there is a related attempt to generalize synchronicity to quantum question and answer games. The authors call these correlations \emph{concurrent} instead of synchronous, and their definition is specific to the case where $X,A$ are full matrix algebras. Concurrency says a correlation takes the cup state (which is maximally entangled) to the cup state.  

The question and answer sets of \cite{BHTT21}, thought of as quantum sets, are maximally quantum, as in Example \ref{exam:max_quant}. With this in mind, their definition of the property concurrent for a quantum correlation $P$ is represented by the graphical equation
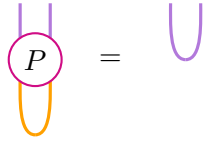
\ctikzfig{concurrent_corr}

Using the above as a template, we can extend concurrency to when the question and answer sets are general quantum sets.

\begin{defn}[Concurrent]
    A linear map $F:X\otimes \op X \to A\otimes \op A$ is called \emph{concurrent} if it satisfies
    \ctikzfig{concurrent_map}
\end{defn}

We can think of concurrency as preserving the cup state.

Note that replacing $P$ with $P^\dagger$ gives ``co-concurrency" of $P$. In the case that $X,A$ are classical, this is causality of the quantum map $P$, as given in \cite{PQP}. 

In the case that $X=A$, and $F$ is a quantum relation on $X$ as defined in \cite{MRV}, concurrency is the reflexivity condition of \cite[Definition 7.6]{MRV}.

It turns out that synchronous games and correlations need not be concurrent.

\begin{exam}[Unfair function game]
The synchronous game 
\ctikzfig{square\_square}
is not concurrent, because
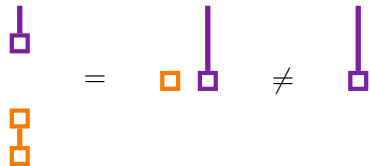
\ctikzfig{square\_square\_not\_concur}
This is an unfair version of the function game, where the players cannot win if they receive different answers.
\end{exam}

The concurrency equation fails by a scaling factor. There are also synchronous games that fail more seriously to be concurrent.

\begin{exam}[Unfair graph game]
Given irreflexive graphs $G, H$, the synchronous game
\ctikzfig{G\_H\_unfair}
is not concurrent, because
\ctikzfig{G\_H\_not\_concur}  
\end{exam}

Note that these are both unfair games. It would be interesting to see a fair synchronous game that is not concurrent.

If $P$ is a bicorrelation, then we can relate concurrency to bisynchronicity, as follows.

\begin{prop}\label{prop:bisync_concur}
    Bisynchronous bicorrelations are concurrent.
\end{prop}

\begin{proof}
    Let $P:X\otimes \op X  \to A \otimes \op A$ be a bisynchronous bicorrelation. Then
    \ctikzfig{pf\_bi\_bi\_concur}
    with steps justified as follows:
    \begin{enumerate}
        \item $P$ preserves sharing by Proposition \ref{prop:bisync_sharing}.
        \item $P$ is unital.
    \end{enumerate}
\end{proof}

Note that the proof did not use the property that a bicorrelation is CP, nor that a bicorrelation is counital. This means there is a somewhat more general result here, that bisynchronous unital maps are concurrent.

\subsection*{Synchronicity in the quantum spaces of maps setting}

In \cite{BKS21}, there is also a definition for synchronicity of quantum question and answer correlations, which the authors refer to as quantum correlations of quantum spaces. Their input space denoted $C(\bb{P})$ corresponds to our question space $X$ while their output space $C(\bb{O})$ corresponds to our answer space $A$. For the purposes of this section, we denote $X = \bigoplus_{\ell=1}^{N_X} X_\ell$ where each $X_\ell$ is the full matrix algebra $M_{m_\ell}(\bb{C})$, and $A = \bigoplus_{k=1}^{N_A} A_k$ where each  $A_k$ is the full matrix algebra $M_{n_k}(\bb{C})$. 

Their definition of synchronous is given by the property
\begin{equation}
\label{eq:polish_sync}
\sum_{s, t, i, j, k, \ell} \frac{1}{n_k m_\ell}\ { }_{\ell \ell}^{k k} P_{(i j),(i j)}^{(s t),(s t)}=N_X
\end{equation}

where the coefficients of $P$ come from writing the correlation $P$ in a matrix unit basis for $X$ and $A$. In \cite{BKS21}, these are denoted with an $X$ instead of a $P$. See their paper for the details of index placement.

When matching summed indices using strings, this looks a lot like a consequence of our definition, given by the loop-sliding equation of Lemma \ref{lem:loop},
\ctikzfig{loop_slide_sync}
however, neither the weighting in the sum, nor the number $N_X$ can arise when interpreting each wire as a quantum set, or even as an (unnormalized) $\dagger$-SFA. (This last claim is based on the $\dagger$-SFA classification result \cite[Proposition 5.33]{catsQT}, which is slightly more general than what's described in Remark \ref{rem:qset_structure}.)

Instead, we show a graphical way to generate  \eqref{eq:polish_sync} using a similar loop-sliding diagram, but allowing for two Frobenius structures on the same set. The approach uses only classical structures, which means that \cite{BKS21}'s definition of synchronicity of QNS correlations is a classical property. We compare the graphical translation of \eqref{eq:polish_sync} to our definition of synchronous, and show that in the classical question and answer case, it is equivalent to our definition.

\subsubsection*{Using classical structures}
Using two commutative Frobenius structures (not both quantum sets) on each of $X$ and $A$, we can draw  \eqref{eq:polish_sync} like this:

\ctikzfig{bks\_sync\_classical}

To read \eqref{eq:polish_sync} from the above diagram, we must specify the $\dagger$-CFAs. Each $\dagger$-CFA on $X$ uses a differently weighted Schur product and the appropriate unit. The $\dagger$-CFA represented by the filled node \begin{tikzpicture}
	\begin{pgfonlayer}{nodelayer}
		\node [style=dark-x-frob] (0) at (0, 0) {};
	\end{pgfonlayer}
\end{tikzpicture}
 has its multiplication given by the Schur product of matrices, $\circ$, with entrywise multiplication in the $\ell^\mathrm{th}$ block weighted by the square root of its size $\sqrt{m_\ell}$ and the Schur unit (having a $1$ in each nonzero entry)  weighted by $\sqrt{m_\ell}^{-1}$. In other words, if $a, b \in X_\ell$, then the filled multiplication takes $a \otimes b$ to $\sqrt{m_\ell}\ a\circ b$, and the unit is given by $\bigoplus_\ell \sqrt{m_\ell}^{-1}J_\ell$ where $J_\ell \in X_\ell$ is the all-ones matrix.

The $\dagger$-CFA on $X$ represented by the empty node \begin{tikzpicture}
	\begin{pgfonlayer}{nodelayer}
		\node [style=x-frob] (0) at (0, 0) {};
	\end{pgfonlayer}
\end{tikzpicture}
 is the same as above but the weights are replaced with their multiplicative inverses. This second $\dagger$-CFA is actually special and so is indeed a quantum set, matching our empty node notation from the rest of the paper. 

For the $\dagger$-CFAs on $A$, the structures are the same as above, replacing $X_\ell$ with $A_k$ and $\sqrt{m_\ell}$ with $\sqrt{n_k}$.

\subsubsection*{When the questions and answers are classical}

In the case of classical questions and answers, $m_\ell = 1 = n_k$ for all $k, \ell$. This means that the matrix blocks are all size $1$, and the coefficients on the operations are all $1$, and so the two Frobenius structures (filled and empty) coincide to be the classical quantum set that copies the diagonal matrix units. In this case, \eqref{eq:polish_sync} actually \emph{is} the loop-sliding equation
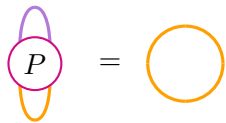
\ctikzfig{loop_slide_sync}
of Lemma \ref{lem:loop}.

In general, the loop-sliding equation is a consequence of our definition of synchronicity. In the classical question and answer setting, they are equivalent, as loop-sliding says
$$\sum_{x\in X, a\in A} p(a,a|x,x) = |X|$$
or equivalently, $p(a,b|x,x) = 0$ for $a\neq b$. This is equivalent to
$$\delta_{ab}p(a,b|x,x)=p(a,b|x,x)$$
for all $a,b,x$, which is the commutative interpretation of our synchronicity definition.

This equivalence in the classical setting is to be expected, as both our definition and that of \cite{BKS21} have been shown to generalize typical synchronicity. However, the generalizations lead in different directions for quantum sets. Our definition generalizes the measure-and-encode step of checking synchronicity, replacing that map by sharing, its noncommutative counterpart. This leads to a synchronicity that observes quantum structure. The \cite{BKS21} definition appears to be equivalent to treating the quantum set as a large classical set, in the sense of superdense coding, and then applying the typical definition of synchronicity.

\section*{Acknowledgements}
    I acknowledge the support of the Natural Sciences and Engineering Research Council of Canada (NSERC). I would like to thank my PhD advisor Michael Brannan for discussing this work with me at length, and asking many fruitful questions. I would also like to thank Dominic Verdon for several conversations that improved this work. I'm grateful to nannies Anna and Victoria, as well as the teachers at Owl's Nest, for giving me the quiet space to think mathematically, and to my partner Max for braving with me the whirlwind that is parenting a toddler as early-career mathematicians.

\DeclareFieldFormat{doi}{DOI: \href{https://doi.org/#1}{\nolinkurl{#1}}}
\addcontentsline{toc}{section}{References}
\printbibliography

@article{LMR20,
title = {Nonlocal games and quantum permutation groups},
journal = {Journal of Functional Analysis},
volume = {279},
number = {5},
pages = {108592},
year = {2020},
issn = {0022-1236},
doi = {10.1016/j.jfa.2020.108592},
url = {https://www.sciencedirect.com/science/article/pii/S002212362030135X},
author = {Martino Lupini and Laura Man\v{c}inska and David E. Roberson}
}

@article{AMRSSV,
author = {Atserias, Albert and Man\v{c}inska, Laura and Roberson, David and \v{S}\'amal, Robert and Severini, Simone and Varvitsiotis, Antonios},
year = {2019},
month = {11},
pages = {},
title = {Quantum and non-signalling graph isomorphisms},
volume = {136},
journal = {Journal of Combinatorial Theory, Series B},
doi = {10.1016/j.jctb.2018.11.002}
}

@article{BCEHPSW,
   title={Bigalois Extensions and the Graph Isomorphism Game},
   volume={375},
   ISSN={1432-0916},
   url={http://dx.doi.org/10.1007/s00220-019-03563-9},
   DOI={10.1007/s00220-019-03563-9},
   number={3},
   journal={Communications in Mathematical Physics},
   publisher={Springer Science and Business Media LLC},
   author={Brannan, Michael and Chirvasitu, Alexandru and Eifler, Kari and Harris, Samuel and Paulsen, Vern and Su, Xiaoyu and Wasilewski, Mateusz},
   year={2019},
   month=sep, pages={1777–1809} }

@article{MR16,
   title={Quantum homomorphisms},
   volume={118},
   ISSN={0095-8956},
   url={http://dx.doi.org/10.1016/j.jctb.2015.12.009},
   DOI={10.1016/j.jctb.2015.12.009},
   journal={Journal of Combinatorial Theory, Series B},
   publisher={Elsevier BV},
   author={Man\v{c}inska, Laura and Roberson, David E.},
   year={2016},
   month=may, pages={228–267} }

@inproceedings{PR21,
  title={Bisynchronous games and factorizable maps},
  author={Vern I. Paulsen and Mizanur Rahaman},
  booktitle={Annales Henri Poincar{\'e}},
  volume={22},
  pages={593--614},
  year={2021},
  doi={10.1007/s00023-020-01003-2},
  organization={Springer}
}

@article{CHK14,
    author = {Coecke, Bob and Heunen, Chris and Kissinger, Aleks},
    year = {2013},
    month = {05},
    pages = {},
    title = {Categories of Quantum and Classical Channels},
    volume = {15},
    journal = {Quantum Information Processing},
    doi = {10.1007/s11128-014-0837-4}
}

@article{BHTT21,
    title = {Synchronicity for quantum non-local games},
    journal = {Journal of Functional Analysis},
    volume = {284},
    number = {2},
    pages = {109738},
    year = {2023},
    issn = {0022-1236},
    doi = {10.1016/j.jfa.2022.109738},
    url = {https://www.sciencedirect.com/science/article/pii/S0022123622003585},
    author = {Michael Brannan and Samuel J. Harris and Ivan G. Todorov and Lyudmila Turowska}
}

@article{BHTT23,
    title = {Quantum no-signalling bicorrelations},
    journal = {Advances in Mathematics},
    volume = {449},
    pages = {109732},
    year = {2024},
    issn = {0001-8708},
    doi = {10.1016/j.aim.2024.109732},
    url = {https://www.sciencedirect.com/science/article/pii/S0001870824002470},
    author = {Michael Brannan and Samuel J. Harris and Ivan G. Todorov and Lyudmila Turowska}
}

@article{MRV,
   title={A compositional approach to quantum functions},
   volume={59},
   ISSN={1089-7658},
   url={http://dx.doi.org/10.1063/1.5020566},
   DOI={10.1063/1.5020566},
   number={8},
   journal={Journal of Mathematical Physics},
   publisher={AIP Publishing},
   author={Musto, Benjamin and Reutter, David and Verdon, Dominic},
   year={2018},
   month=aug }

@article{Vic10,
   title={Categorical Formulation of Finite-Dimensional Quantum Algebras},
   volume={304},
   ISSN={1432-0916},
   url={http://dx.doi.org/10.1007/s00220-010-1138-0},
   DOI={10.1007/s00220-010-1138-0},
   number={3},
   journal={Communications in Mathematical Physics},
   publisher={Springer Science and Business Media LLC},
   author={Vicary, Jamie},
   year={2010},
   month=nov, pages={765–796} }

@article{PSSTW16,
   title={Estimating quantum chromatic numbers},
   volume={270},
   ISSN={0022-1236},
   url={http://dx.doi.org/10.1016/j.jfa.2016.01.010},
   DOI={10.1016/j.jfa.2016.01.010},
   number={6},
   journal={Journal of Functional Analysis},
   publisher={Elsevier BV},
   author={Paulsen, Vern I. and Severini, Simone and Stahlke, Daniel and Todorov, Ivan G. and Winter, Andreas},
   year={2016},
   month=mar, pages={2188–2222} }

@book{PQP,
  title={Picturing Quantum Processes},
  author={Coecke, B. and Kissinger, A.},
  isbn={9781107104228},
  lccn={2016035537},
  url={https://books.google.ca/books?id=I9gcDgAAQBAJ},
  year={2017},
  publisher={Cambridge University Press},
  doi={10.1017/9781316219317}
}

@book{catsQT,
  title={Categories for Quantum Theory: An Introduction},
  author={Heunen, C. and Vicary, J.},
  isbn={9780198739623},
  series={Oxford Graduate Texts in Mathematics Series},
  url={https://books.google.ca/books?id=PdG8DwAAQBAJ},
  year={2019},
  publisher={Oxford University Press}
}

@article{BKS21,
    author = {Bochniak, Arkadiusz and Kasprzak, Pawe\l{} and So\l{}tan, Piotr M},
    title = "{Quantum Correlations on Quantum Spaces}",
    journal = {International Mathematics Research Notices},
    volume = {2023},
    number = {14},
    pages = {12400-12440},
    year = {2023},
    month = {07},
    issn = {1073-7928},
    doi = {10.1093/imrn/rnac139},
    url = {https://doi.org/10.1093/imrn/rnac139},
    eprint = {https://academic.oup.com/imrn/article-pdf/2023/14/12400/50903565/rnac139.pdf},
}

@ARTICLE{DW16,
  author={Duan, Runyao and Winter, Andreas},
  journal={IEEE Transactions on Information Theory}, 
  title={No-Signalling-Assisted Zero-Error Capacity of Quantum Channels and an Information Theoretic Interpretation of the Lovász Number}, 
  year={2016},
  volume={62},
  number={2},
  pages={891-914},
  doi={10.1109/TIT.2015.2507979}
}

@phdthesis{myThesis,
    author={Goldberg, Adina},
    title={Synchronous and quantum games:
Graphical and algebraic methods},
    school={University of Waterloo, Pure Mathematics},
    year={2025},
    note = {Available online at \url{https://hdl.handle.net/10012/21736}}
}

@article{fritz12,
  title={Tsirelson's problem and Kirchberg's conjecture},
  author={Fritz, Tobias},
  journal={Reviews in Mathematical Physics},
  volume={24},
  number={05},
  pages={1250012},
  year={2012},
  publisher={World Scientific},
  doi={10.1142/S0129055X12500122}
}

@article{PV16,
    author = {Palazuelos, Carlos and Vidick, Thomas},
    title = {Survey on nonlocal games and operator space theory},
    journal = {Journal of Mathematical Physics},
    volume = {57},
    number = {1},
    pages = {015220},
    year = {2016},
    month = {01},
    issn = {0022-2488},
    doi = {10.1063/1.4938052},
    url = {https://doi.org/10.1063/1.4938052},
    eprint = {https://pubs.aip.org/aip/jmp/article-pdf/doi/10.1063/1.4938052/16729413/015220\_1\_online.pdf},
}

@online{paulsen16,
  author        = {Vern Paulsen},
  title         = {Entanglement and Non-Locality},
  editor       = {Samuel J. Harris and Satish K. Pandey},
  year          = {2016},
  howpublished  = {University of Waterloo lecture notes for PMATH 990/QIC 890}, url = {https://www.math.uwaterloo.ca/~vpaulsen/EntanglementAndNonlocality_LectureNotes_7.pdf}
}

@article{RV15,
author = {Regev, Oded and Vidick, Thomas},
title = {Quantum XOR Games},
year = {2015},
issue_date = {September 2015},
publisher = {Association for Computing Machinery},
address = {New York, NY, USA},
volume = {7},
number = {4},
issn = {1942-3454},
url = {https://doi.org/10.1145/2799560},
doi = {10.1145/2799560},
journal = {ACM Trans. Comput. Theory},
month = aug,
articleno = {15},
numpages = {43}
}

@article{BGH22,
    author = {Brannan, Michael and Ganesan, Priyanga and Harris, Samuel J.},
    title = {The quantum-to-classical graph homomorphism game},
    journal = {Journal of Mathematical Physics},
    volume = {63},
    number = {11},
    pages = {112204},
    year = {2022},
    month = {11},
    issn = {0022-2488},
    doi = {10.1063/5.0072288},
    url = {https://doi.org/10.1063/5.0072288},
    eprint = {https://pubs.aip.org/aip/jmp/article-pdf/doi/10.1063/5.0072288/16593548/112204\_1\_online.pdf},
}

@article{LTW08,
  title={Coherent state exchange in multi-prover quantum interactive proof systems},
  author={Debbie W. Leung and Ben Toner and John Watrous},
  journal={Chic. J. Theor. Comput. Sci.},
  year={2008},
  volume={2013},
  url={https://api.semanticscholar.org/CorpusID:1517639},
  doi={10.4086/cjtcs.2013.011}
}

@article{CJPP11,
  title={Rank-one quantum games},
  author={Tom Cooney and Marius Junge and Carlos Palazuelos and David P{\'e}rez-Garc{\'i}a},
  journal={computational complexity},
  year={2011},
  volume={24},
  pages={133-196},
  url={https://api.semanticscholar.org/CorpusID:6122906},
  doi={10.1007/s00037-014-0096-x}
}

@article{HT25,
title = {Homomorphisms of quantum hypergraphs},
journal = {Journal of Mathematical Analysis and Applications},
volume = {543},
number = {2, Part 1},
pages = {128907},
year = {2025},
issn = {0022-247X},
doi = {10.1016/j.jmaa.2024.128907},
url = {https://www.sciencedirect.com/science/article/pii/S0022247X24008291},
author = {Gage Hoefer and Ivan G. Todorov}
}

@misc{HT22,
      title={Quantum hypergraph homomorphisms and non-local games}, 
      author={Gage Hoefer and Ivan G. Todorov},
      year={2022},
      eprint={2211.04851},
      archivePrefix={arXiv},
      primaryClass={math.OA},
      url={https://arxiv.org/abs/2211.04851}, 
}

@misc{CLTT23,
      title={Values of cooperative quantum games}, 
      author={Jason Crann and Rupert H. Levene and Ivan G. Todorov and Lyudmila Turowska},
      year={2023},
      eprint={2310.17735},
      archivePrefix={arXiv},
      primaryClass={quant-ph},
      url={https://arxiv.org/abs/2310.17735}, 
}

@article{Bus11,
  title={All entangled quantum states are nonlocal.},
  author={Francesco Buscemi},
  journal={Physical review letters},
  year={2011},
  volume={108 20},
  pages={200401},
  doi={10.1103/PhysRevLett.108.200401},
  url={https://api.semanticscholar.org/CorpusID:14393220}
}

@inproceedings{AC04,
  title={A categorical semantics of quantum protocols},
  author={Abramsky, Samson and Coecke, Bob},
  booktitle={Proceedings of the 19th Annual IEEE Symposium on Logic in Computer Science, 2004.},
  pages={415--425},
  year={2004},
  organization={IEEE},
  doi={10.1109/LICS.2004.1319636}
}

@article{Kor20,
    author = {Kornell, Andre},
    title = {Quantum sets},
    journal = {Journal of Mathematical Physics},
    volume = {61},
    number = {10},
    pages = {102202},
    year = {2020},
    month = {10},
    issn = {0022-2488},
    doi = {10.1063/1.5054128},
    url = {https://doi.org/10.1063/1.5054128},
    eprint = {https://pubs.aip.org/aip/jmp/article-pdf/doi/10.1063/1.5054128/16167920/102202\_1\_online.pdf},
}

@article{JMRW16,
  title={Extended non-local games and monogamy-of-entanglement games},
  author={Johnston, Nathaniel and Mittal, Rajat and Russo, Vincent and Watrous, John},
  journal={Proceedings of the Royal Society A: Mathematical, Physical and Engineering Sciences},
  volume={472},
  number={2189},
  pages={20160003},
  year={2016},
  publisher={The Royal Society Publishing},
  doi={10.1098/rspa.2016.0003}
}

@article{TFKW13,
  title={A monogamy-of-entanglement game with applications to device-independent quantum cryptography},
  author={Tomamichel, Marco and Fehr, Serge and Kaniewski, J{\textcommabelow{e}}drzej and Wehner, Stephanie},
  journal={New Journal of Physics},
  volume={15},
  number={10},
  pages={103002},
  year={2013},
  publisher={IOP Publishing},
  doi={10.1088/1367-2630/15/10/103002}
}

@article{TT24,
  title={Quantum no-signalling correlations and non-local games},
  author={Todorov, Ivan G and Turowska, Lyudmila},
  journal={Communications in Mathematical Physics},
  volume={405},
  number={6},
  pages={141},
  year={2024},
  publisher={Springer},
  doi={10.1007/s00220-024-05001-x}
}

@misc{Wea10,
      title={Quantum relations}, 
      author={Nik Weaver},
      year={2010},
      eprint={1005.0354},
      archivePrefix={arXiv},
      primaryClass={math.OA},
      url={https://arxiv.org/abs/1005.0354}, 
}

\newpage
\appendix
\section{Omitted proofs}
\label{sec:proofs}
\addtocontents{toc}{\setcounter{tocdepth}{-10}}

\subsection{CP maps are self-conjugate}\label{pf:CP_self_conj}
\begin{proof}[Proof of Proposition \ref{prop:CP_self_conj}]
Let $X, A$ be quantum sets and let $f:X \to A$ be completely positive. By definition, there is some $g:X\otimes A\to Q$ for some other quantum set $Q$ so that $f$ satisfies
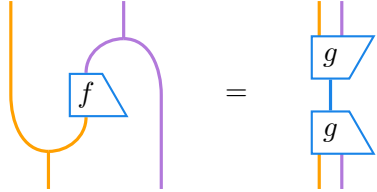
\ctikzfig{CP\_defn}

But then we have
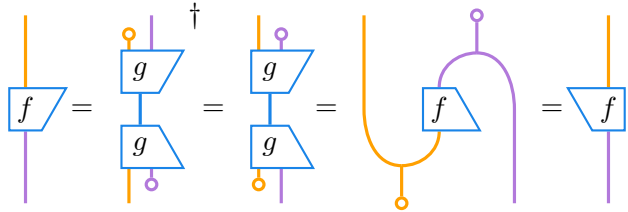
\ctikzfig{CP\_self\_conj}
\end{proof}

\subsection{Strategies realize correlations}\label{pf:corr_from_strat}
\begin{proof}[Proof of Remark \ref{rem:corr_from_strat}]

Let $\varphi$ be a combined strategy over a finite-dimensional $\cl{H}$ and let $\ket\psi \in \cl{H}$ be a normalized state. Let $P$ be the expectation of $\varphi$ with respect to $\ket\psi$, as in Definition \ref{defn:real}. We show that $P$ is a correlation.

$P$ inherits counitality directly from \eqref{eq:q_counital} for $\varphi$, using the fact that $\ket\psi$ is normalized. To see that $P$ is completely positive, we proceed diagrammatically, applying \eqref{eq:q_comult} and \eqref{eq:self_conj} to $\varphi$ and bending some wires for the second equality, and using the Frobenius equation for the third, to get
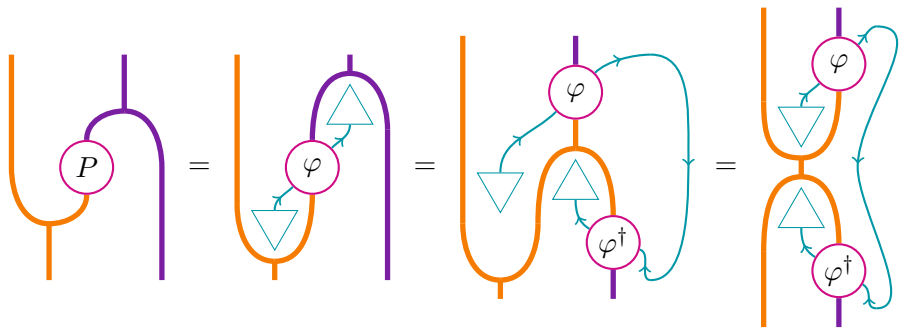
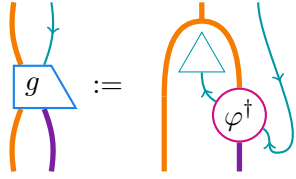
\ctikzfig{strat\_real\_corr}
Letting
\ctikzfig{sqrt\_strat}
we can see that $P$ satisfies Definition \ref{defn:CP}.

\end{proof}

\subsection{Nonsignalling marginals are channels}\label{pf:margs_channels}
\begin{proof}[Proof of Corollary \ref{cor:margs_channels}]

By definition of $P$ being nonsignalling, we have
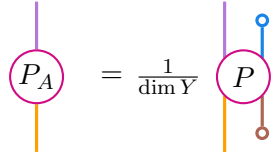
\ctikzfig{P\_A}

We can see that $P_A$ is counital as follows:
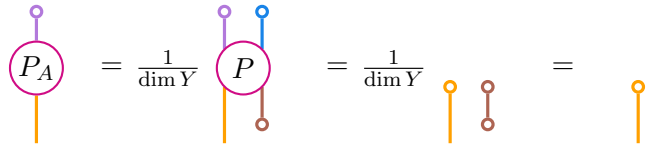
\ctikzfig{P\_A\_counital}

To see complete positivity, let $f$ be a square root of $P$ as given in the definition of complete positivity for $P$. Then $g$ given as follows is a square root of $P_A$.
\ctikzfig{P\_A\_sqrt}

The approach is the same for $P_B$.
\end{proof}
\newpage
\subsection{Cauchy-Schwarz inequality for quantum functions}\label{pf:cs_ineq_qfunc}
\begin{proof}[Proof of Theorem \ref{thm:cs_ineq_qfunc}]
    \begin{enumerate}
        \item 
            We have
            \ctikzfig{pf\_cs\_ineq}

            where the first equality is obtained graphically as follows. The leftmost component on the left-hand side becomes the bottom component on the right-hand side, by pulling $E$ to the right of $F$ and letting the wires follow. The remaining component transforms by sliding both quantum functions $E$ and $F$ to the other side of the loop to get their respective daggers.
            
            The subsequent inequality is the Cauchy-Schwarz inequality \ref{lem:cs_ineq} with 
            \ctikzfig{cs\_ineq\_qfunc\_maps}

            But, given that $E$ and $F$ are quantum functions, we can further simplify this by computing
            \ctikzfig{equal\_qc\_pf\_same}
            using the properties of $E$ as a quantum function, the definition of the cup, and the fact that $X, A$ are special and $\ket \psi$ is normal.
            The same holds for $F$. Substituting these loops into the above inequality, we have our desired result.

            Now, let us examine the case of equality. Equality holds if and only if $g:= \alpha f$ for some $\alpha$ in $\mathbb{C}$. This is exactly the condition that
            \begin{align*}
            \begin{tikzpicture}
	\begin{pgfonlayer}{nodelayer}
		\node [style=ket] (0) at (-0.75, -0.5) {};
		\node [style=correlation] (7) at (0, 0.25) {$E$};
		\node [style=none] (12) at (0.5, 1) {};
		\node [style=none] (23) at (0, -1) {};
		\node [style=none] (26) at (0, 1) {};
	\end{pgfonlayer}
	\begin{pgfonlayer}{edgelayer}
		\draw [style=hilbert, in=-150, out=90] (0) to (7);
		\draw [style=hilbert, in=270, out=30] (7) to (12.center);
		\draw [style=x-wire] (7) to (23.center);
		\draw [style=a-wire] (26.center) to (7);
	\end{pgfonlayer}
\end{tikzpicture}     &= \alpha \left ( \begin{tikzpicture}
	\begin{pgfonlayer}{nodelayer}
		\node [style=none] (0) at (-0.75, -0.5) {};
		\node [style=none] (1) at (0, -1) {};
		\node [style=bra] (2) at (0.75, 0.5) {};
		\node [style=correlation] (3) at (0, -0.25) {$F$};
		\node [style=none] (4) at (0, 1) {};
		\node [style=none] (5) at (0.5, 1) {};
	\end{pgfonlayer}
	\begin{pgfonlayer}{edgelayer}
		\draw [style=hilbert, in=120, out=-105, looseness=0.75] (5.center) to (0.center);
		\draw [style=hilbert, in=-135, out=-60, looseness=1.25] (0.center) to (3);
		\draw [style=hilbert, in=-90, out=30] (3) to (2);
		\draw [style=x-wire] (1.center) to (3);
		\draw [style=a-wire] (3) to (4.center);
	\end{pgfonlayer}
\end{tikzpicture}     \right)_*
                            = \alpha \left ( \begin{tikzpicture}
	\begin{pgfonlayer}{nodelayer}
		\node [style=ket] (6) at (0.75, -0.5) {};
		\node [style=correlation] (7) at (0, 0.25) {$F^\dagger$};
		\node [style=none] (8) at (-1, 0.5) {};
		\node [style=none] (9) at (0, -1) {};
		\node [style=none] (10) at (0, 1) {};
		\node [style=none] (11) at (0.5, -1) {};
	\end{pgfonlayer}
	\begin{pgfonlayer}{edgelayer}
		\draw [style=hilbert, in=315, out=135] (6) to (7);
		\draw [style=hilbert, in=45, out=135] (7) to (8.center);
		\draw [style=x-wire] (10.center) to (7);
		\draw [style=hilbert, in=90, out=-135, looseness=0.75] (8.center) to (11.center);
		\draw [style=a-wire] (7) to (9.center);
	\end{pgfonlayer}
\end{tikzpicture} \right)^* \\
                            &= \alpha\ \begin{tikzpicture}
	\begin{pgfonlayer}{nodelayer}
		\node [style=ket] (6) at (0.75, -0.5) {};
		\node [style=correlation] (7) at (0, 0.25) {$F^\dagger$};
		\node [style=none] (8) at (-1, 0.5) {};
		\node [style=none] (9) at (0, -1) {};
		\node [style=none] (10) at (0, 1) {};
		\node [style=none] (11) at (0.25, -1) {};
		\node [style=none] (12) at (-1.5, -0.5) {};
		\node [style=none] (13) at (-1.25, -0.5) {};
		\node [style=none] (14) at (1, 1) {};
		\node [style=none] (15) at (-0.75, -1.5) {};
		\node [style=none] (16) at (-0.5, -1.5) {};
		\node [style=none] (17) at (0.5, 1.5) {};
		\node [style=none] (18) at (1.5, -1.75) {};
		\node [style=none] (19) at (-1.25, 1.75) {};
		\node [style=none] (20) at (-1.5, 1.75) {};
	\end{pgfonlayer}
	\begin{pgfonlayer}{edgelayer}
		\draw [style=hilbert, in=315, out=135] (6) to (7);
		\draw [style=hilbert, in=45, out=135] (7) to (8.center);
		\draw [style=x-wire] (7)
			 to (10.center)
			 to [in=-180, out=90] (17.center)
			 to [in=90, out=0] (14.center)
			 to [in=90, out=-90, looseness=1.25] (18.center);
		\draw [style=hilbert] (8.center)
			 to [in=90, out=-135, looseness=0.75] (11.center)
			 to [in=0, out=-90] (16.center);
		\draw [style=a-wire] (7)
			 to (9.center)
			 to [in=0, out=-90] (15.center)
			 to [in=-90, out=-180] (12.center)
			 to (20.center);
		\draw [style=hilbert] (16.center)
			 to [in=-90, out=-180] (13.center)
			 to (19.center);
	\end{pgfonlayer}
\end{tikzpicture} 
                            = \alpha \begin{tikzpicture}
	\begin{pgfonlayer}{nodelayer}
		\node [style=ket] (0) at (-0.75, -0.5) {};
		\node [style=correlation] (7) at (0, 0.25) {$F$};
		\node [style=none] (12) at (0.5, 1) {};
		\node [style=none] (23) at (0, -1) {};
		\node [style=none] (26) at (0, 1) {};
	\end{pgfonlayer}
	\begin{pgfonlayer}{edgelayer}
		\draw [style=hilbert, in=-150, out=90] (0) to (7);
		\draw [style=hilbert, in=270, out=30] (7) to (12.center);
		\draw [style=x-wire] (7) to (23.center);
		\draw [style=a-wire] (26.center) to (7);
	\end{pgfonlayer}
\end{tikzpicture}             
            \end{align*}
            Because $E$ is a quantum function, due to \eqref{eq:q_counital} we must have $\alpha =1$, resulting in equality if and only if
            \ctikzfig{equal\_qc}
           
        \item 
            Let $\cl H^*$ denote the linear dual of $\cl H$. Consider the map $E': \cl H \otimes \cl H^* \otimes X \to A \otimes \cl H \otimes \cl H^*$, given by 
            \ctikzfig{q\_func\_id}
            Observe that $E'$ is again a quantum function $X\to_{\cl H \otimes \cl H^*} A$, by noting that the added $\cl H^*$ wire does not interfere with the diagrammatic properties of $E$. Construct $F'$ in the same way from $F$. Observe that $(E',F')$ is again a synchronous strategy, by adding an $\cl H^*$ wire to the synchronicity equation for $(E,F)$, as pictured.
            \ctikzfig{pf\_qc\_id\_sync}
            Let $\ket \psi$ be the normalized cup state of $\cl H \otimes \cl H^*$ coming from the duality of $\cl H^*$ to $\cl H$. Then $(E',F')$ together with $\ket \psi$ generates a synchronous correlation, $P$. 

            We can now apply part (1) to get 
            \ctikzfig{pf\_cs\_ineq\_2}
            with equality if and only if
            \ctikzfig{equal\_qc\_id}
            or equivalently, by the snake equations,
            \ctikzfig{equal\_qc\_strat}
    \end{enumerate}
\end{proof}

\subsection{Structure of bistrategies}
\label{pf:bistrat_structure}

\begin{proof}[Proof of Proposition \ref{prop:bistrat_structure}]
     First note that because $\varphi$ is a bistrategy, it is both unital and counital. Applying all units and counits to the quantum set wires of $\varphi$, we can see that its input dimensions and output dimensions must be equal, that is $\dim A \dim B = \dim X \dim Y$, or graphically \begin{tikzpicture}
	\begin{pgfonlayer}{nodelayer}
		\node [style=none] (0) at (-1.5, 0.25) {};
		\node [style=none] (1) at (-1.5, -0.25) {};
		\node [style=none] (2) at (-0.75, 0.25) {};
		\node [style=none] (3) at (-0.75, -0.25) {};
		\node [style=none] (4) at (1.5, 0.25) {};
		\node [style=none] (5) at (1.5, -0.25) {};
		\node [style=none] (6) at (0.75, 0.25) {};
		\node [style=none] (7) at (0.75, -0.25) {};
		\node [style=none] (8) at (0, 0) {$=$};
	\end{pgfonlayer}
	\begin{pgfonlayer}{edgelayer}
		\draw [style=a-wire, in=0, out=0, looseness=1.75] (0.center) to (1.center);
		\draw [style=a-wire, in=-180, out=180, looseness=1.75] (0.center) to (1.center);
		\draw [style=b-wire, in=0, out=0, looseness=1.75] (2.center) to (3.center);
		\draw [style=b-wire, in=-180, out=180, looseness=1.75] (2.center) to (3.center);
		\draw [style=y-wire, in=0, out=0, looseness=1.75] (4.center) to (5.center);
		\draw [style=y-wire, in=-180, out=180, looseness=1.75] (4.center) to (5.center);
		\draw [style=x-wire, in=0, out=0, looseness=1.75] (6.center) to (7.center);
		\draw [style=x-wire, in=-180, out=180, looseness=1.75] (6.center) to (7.center);
	\end{pgfonlayer}
\end{tikzpicture}.
    \begin{itemize}
        \item
            If $\varphi$ is quantum commuting then we can write $\overline \varphi$ as follows:
            \ctikzfig{co\_qc\_is\_qc\_1}
            But then using the fact that $(E,F)$ is quantum commuting, we have
            \ctikzfig{co\_qc\_is\_qc\_2}
            so it looks like $\overline\varphi$ is a quantum commuting strategy for $\lambda^\dagger$ implemented by $(\overline E, \overline F)$, so long as $\overline E, \overline F$ are themselves quantum functions. We will first check if the operators commute. Indeed,
            \ctikzfig{co\_qc\_is\_qc\_3}

            Now, let's see that $\overline E$ is counital, as in \eqref{eq:q_counital}. By symmetry, the same argument applies for $\overline F$. We have
            \ctikzfig{E\_bar\_counital}
            with the steps justified as follows:
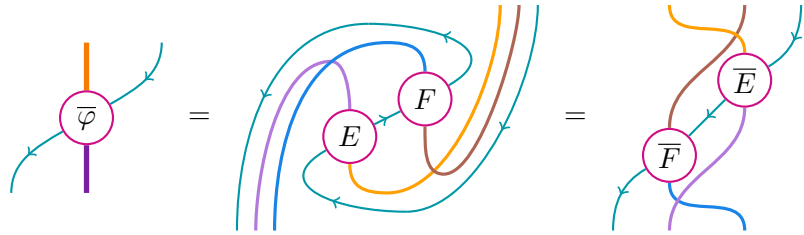
            \begin{enumerate}
                \item As shown at the beginning of this proof, because $\overline \varphi$ is a bistrategy, its input dimensions are equal to its output dimensions. We rewrite $\dim Y$ using specialness.
                \item Because $F$ is counital, $\overline F$ is unital.
                \item Counitality of $\overline \varphi$.
            \end{enumerate}
            Dividing by the output dimensions ($\dim A \dim B$), we get the counitality equation for $\overline E$, so long as $\dim X = \dim A$. Otherwise, we get counitality up to a factor of $\dim X / \dim A$. 

            Next, let us see that $\overline E$ (and $\overline F$ by symmetry) is comultiplicative, as in \eqref{eq:q_comult}.

            We have
            \ctikzfig{E\_bar\_comult}
            with the steps justified as follows:
            \begin{enumerate}
                \item Counitality of $\overline F$.
                \item Comultiplication of $Y$ is counital.
                \item Comultiplicativity of $\overline \varphi$, and contracting the $B$ wire.
                \item Counitality of $\overline F$, and specialness of $B$.
            \end{enumerate}
            Dividing by $\dim B$, we have that $\overline E$ is comultiplicative.

            Finally, realness \eqref{eq:self_conj} of $\overline \varphi$ follows from realness of $\varphi$. So $\overline \varphi$ is a quantum commuting strategy.
        \item
            If $\varphi$ is quantum tensor, implememented by some $(E,F)$, then it is quantum commuting with operators
            \ctikzfig{co\_qt\_is\_qt\_1}
            But then $\overline \varphi$ is quantum commuting and implemented by $(\overline{E'}, \overline{F'})$. These are given by
            \ctikzfig{co\_qt\_is\_qt\_2}
            and so $\overline\varphi$ is a quantum tensor strategy implemented by $(\overline E, \overline F)$.
        \item
            If $\varphi$ is local, then $\cl H = \mathbb{C}$ but also by the above, $\overline \varphi$ is quantum tensor and implemented by $\dagger$-cohomomorphisms $(E^\dagger, F^\dagger)$. So $\overline \varphi: A \otimes B \to X \otimes Y$ is local.
    \end{itemize}
\end{proof}

\subsection{The graph homomorphism game is well-defined}
\label{pf:graph_hom_game}

\begin{proof}[Proof of Proposition \ref{prop:graph_hom_game}]
    To see that $\lambda_{G\to H}$ is self-conjugate, note that it is defined by composing only self-conjugate maps.

    We will now see that $\lambda_{G\to H}$ satisfies the idempotency condition of Definition \ref{defn:game}. We have a sum of five terms, leading to twenty-five terms to compute. In order to make this easier, note that all of the terms are disconnected diagrams, so we may compute top and bottom halves separately. We may also take daggers of our computations.

    There are three different bottom-halves, so we begin by computing their coproducts, given in Table \ref{tbl:hom_coproducts}. The same holds by taking daggers for products of the top-halves of $\lambda_{G\to H}$.

    We use Table \ref{tbl:hom_comm_prod} to compute the twenty-five terms in the sum. Adding up the terms in the table, we get $\lambda_{G\to H}$, concluding the proof.
\end{proof}

\begin{table}[p]
    \caption{Coproducts of the bottom halves of $\lambda_{G\to H}$.}
    \label{tbl:hom_coproducts}
    \begin{center}
    \begin{tabular}{c|c c c}
        Coproducts      & \begin{tikzpicture}
	\begin{pgfonlayer}{nodelayer}
		\node [style=game] (10) at (0, 0) {$G$};
		\node [style=none] (11) at (0, -0.5) {};
		\node [style=none] (12) at (0.75, 0) {};
		\node [style=none] (13) at (0.75, -0.5) {};
	\end{pgfonlayer}
	\begin{pgfonlayer}{edgelayer}
		\draw [style=x-wire] (11.center) to (10);
		\draw [style=x-wire] (10.center)
			 to [bend left=90, looseness=1.75] (12.center)
			 to (13.center);
	\end{pgfonlayer}
\end{tikzpicture}    & \begin{tikzpicture}
	\begin{pgfonlayer}{nodelayer}
		\node [style=none] (10) at (0, 0) {};
		\node [style=none] (11) at (0, -0.5) {};
		\node [style=none] (12) at (0.75, 0) {};
		\node [style=none] (13) at (0.75, -0.5) {};
	\end{pgfonlayer}
	\begin{pgfonlayer}{edgelayer}
		\draw [style=x-wire] (11.center) to (10.center);
		\draw [style=x-wire] (10.center)
			 to [bend left=90, looseness=1.75] (12.center)
			 to (13.center);
	\end{pgfonlayer}
\end{tikzpicture}    & \begin{tikzpicture}
	\begin{pgfonlayer}{nodelayer}
		\node [style=x-frob] (10) at (0, 0) {};
		\node [style=none] (11) at (0, -0.5) {};
		\node [style=x-frob] (12) at (0.75, 0) {};
		\node [style=none] (13) at (0.75, -0.5) {};
	\end{pgfonlayer}
	\begin{pgfonlayer}{edgelayer}
		\draw [style=x-wire] (11.center) to (10);
		\draw [style=x-wire] (12) to (13.center);
	\end{pgfonlayer}
\end{tikzpicture}  \\
        \hline
        \begin{tikzpicture}
	\begin{pgfonlayer}{nodelayer}
		\node [style=game] (10) at (0, 0) {$G$};
		\node [style=none] (11) at (0, -0.5) {};
		\node [style=none] (12) at (0.75, 0) {};
		\node [style=none] (13) at (0.75, -0.5) {};
	\end{pgfonlayer}
	\begin{pgfonlayer}{edgelayer}
		\draw [style=x-wire] (11.center) to (10);
		\draw [style=x-wire] (10.center)
			 to [bend left=90, looseness=1.75] (12.center)
			 to (13.center);
	\end{pgfonlayer}
\end{tikzpicture}  & \begin{tikzpicture}
	\begin{pgfonlayer}{nodelayer}
		\node [style=game] (10) at (0, 0) {$G$};
		\node [style=none] (11) at (0, -0.5) {};
		\node [style=none] (12) at (0.75, 0) {};
		\node [style=none] (13) at (0.75, -0.5) {};
	\end{pgfonlayer}
	\begin{pgfonlayer}{edgelayer}
		\draw [style=x-wire] (11.center) to (10);
		\draw [style=x-wire] (10.center)
			 to [bend left=90, looseness=1.75] (12.center)
			 to (13.center);
	\end{pgfonlayer}
\end{tikzpicture}    & $0$               & \begin{tikzpicture}
	\begin{pgfonlayer}{nodelayer}
		\node [style=game] (10) at (0, 0) {$G$};
		\node [style=none] (11) at (0, -0.5) {};
		\node [style=none] (12) at (0.75, 0) {};
		\node [style=none] (13) at (0.75, -0.5) {};
	\end{pgfonlayer}
	\begin{pgfonlayer}{edgelayer}
		\draw [style=x-wire] (11.center) to (10);
		\draw [style=x-wire] (10.center)
			 to [bend left=90, looseness=1.75] (12.center)
			 to (13.center);
	\end{pgfonlayer}
\end{tikzpicture}    \\
        \begin{tikzpicture}
	\begin{pgfonlayer}{nodelayer}
		\node [style=none] (10) at (0, 0) {};
		\node [style=none] (11) at (0, -0.5) {};
		\node [style=none] (12) at (0.75, 0) {};
		\node [style=none] (13) at (0.75, -0.5) {};
	\end{pgfonlayer}
	\begin{pgfonlayer}{edgelayer}
		\draw [style=x-wire] (11.center) to (10.center);
		\draw [style=x-wire] (10.center)
			 to [bend left=90, looseness=1.75] (12.center)
			 to (13.center);
	\end{pgfonlayer}
\end{tikzpicture}  & $0$               & \begin{tikzpicture}
	\begin{pgfonlayer}{nodelayer}
		\node [style=none] (10) at (0, 0) {};
		\node [style=none] (11) at (0, -0.5) {};
		\node [style=none] (12) at (0.75, 0) {};
		\node [style=none] (13) at (0.75, -0.5) {};
	\end{pgfonlayer}
	\begin{pgfonlayer}{edgelayer}
		\draw [style=x-wire] (11.center) to (10.center);
		\draw [style=x-wire] (10.center)
			 to [bend left=90, looseness=1.75] (12.center)
			 to (13.center);
	\end{pgfonlayer}
\end{tikzpicture}    &  \begin{tikzpicture}
	\begin{pgfonlayer}{nodelayer}
		\node [style=none] (10) at (0, 0) {};
		\node [style=none] (11) at (0, -0.5) {};
		\node [style=none] (12) at (0.75, 0) {};
		\node [style=none] (13) at (0.75, -0.5) {};
	\end{pgfonlayer}
	\begin{pgfonlayer}{edgelayer}
		\draw [style=x-wire] (11.center) to (10.center);
		\draw [style=x-wire] (10.center)
			 to [bend left=90, looseness=1.75] (12.center)
			 to (13.center);
	\end{pgfonlayer}
\end{tikzpicture}   \\
        \begin{tikzpicture}
	\begin{pgfonlayer}{nodelayer}
		\node [style=x-frob] (10) at (0, 0) {};
		\node [style=none] (11) at (0, -0.5) {};
		\node [style=x-frob] (12) at (0.75, 0) {};
		\node [style=none] (13) at (0.75, -0.5) {};
	\end{pgfonlayer}
	\begin{pgfonlayer}{edgelayer}
		\draw [style=x-wire] (11.center) to (10);
		\draw [style=x-wire] (12) to (13.center);
	\end{pgfonlayer}
\end{tikzpicture}& \begin{tikzpicture}
	\begin{pgfonlayer}{nodelayer}
		\node [style=game] (10) at (0, 0) {$G$};
		\node [style=none] (11) at (0, -0.5) {};
		\node [style=none] (12) at (0.75, 0) {};
		\node [style=none] (13) at (0.75, -0.5) {};
	\end{pgfonlayer}
	\begin{pgfonlayer}{edgelayer}
		\draw [style=x-wire] (11.center) to (10);
		\draw [style=x-wire] (10.center)
			 to [bend left=90, looseness=1.75] (12.center)
			 to (13.center);
	\end{pgfonlayer}
\end{tikzpicture}    & \begin{tikzpicture}
	\begin{pgfonlayer}{nodelayer}
		\node [style=none] (10) at (0, 0) {};
		\node [style=none] (11) at (0, -0.5) {};
		\node [style=none] (12) at (0.75, 0) {};
		\node [style=none] (13) at (0.75, -0.5) {};
	\end{pgfonlayer}
	\begin{pgfonlayer}{edgelayer}
		\draw [style=x-wire] (11.center) to (10.center);
		\draw [style=x-wire] (10.center)
			 to [bend left=90, looseness=1.75] (12.center)
			 to (13.center);
	\end{pgfonlayer}
\end{tikzpicture}    & \begin{tikzpicture}
	\begin{pgfonlayer}{nodelayer}
		\node [style=x-frob] (10) at (0, 0) {};
		\node [style=none] (11) at (0, -0.5) {};
		\node [style=x-frob] (12) at (0.75, 0) {};
		\node [style=none] (13) at (0.75, -0.5) {};
	\end{pgfonlayer}
	\begin{pgfonlayer}{edgelayer}
		\draw [style=x-wire] (11.center) to (10);
		\draw [style=x-wire] (12) to (13.center);
	\end{pgfonlayer}
\end{tikzpicture}  \\
    \end{tabular}
    \end{center}
\end{table}

\renewcommand{\arraystretch}{4}
\begin{table}[p]
    \small
    \caption{Commuting products of the terms of $\lambda_{G\to H}$.}
    \label{tbl:hom_comm_prod}
    \begin{center}
    \begin{tabular}{c|c c c c c}
        \begin{tikzpicture}
	\begin{pgfonlayer}{nodelayer}
		\node [style=none] (0) at (-0.5, -0.25) {};
		\node [style=none] (1) at (0.5, -0.25) {};
		\node [style=none] (2) at (0, -0.75) {};
		\node [style=none] (3) at (0, -1.25) {};
		\node [style=none] (4) at (-0.5, 0) {row};
		\node [style=none] (5) at (0.5, 0) {col};
		\node [style=none] (6) at (-0.5, 0.25) {};
		\node [style=none] (7) at (0.5, 0.25) {};
		\node [style=none] (8) at (0, 0.75) {};
		\node [style=none] (9) at (0, 1.25) {};
	\end{pgfonlayer}
	\begin{pgfonlayer}{edgelayer}
		\draw [style=xy-wire, bend left=45] (2.center) to (0.center);
		\draw [style=xy-wire, bend right=45, looseness=0.75] (2.center) to (1.center);
		\draw [style=xy-wire] (2.center) to (3.center);
		\draw [style=ab-wire, bend right=45] (8.center) to (6.center);
		\draw [style=ab-wire, bend left=45, looseness=0.75] (8.center) to (7.center);
		\draw [style=ab-wire] (8.center) to (9.center);
	\end{pgfonlayer}
\end{tikzpicture} & \begin{tikzpicture}
	\begin{pgfonlayer}{nodelayer}
		\node [style=game] (10) at (-0.375, -0.5) {$G$};
		\node [style=none] (11) at (-0.375, -1) {};
		\node [style=none] (12) at (0.375, -0.5) {};
		\node [style=none] (13) at (0.375, -1) {};
		\node [style=game] (14) at (-0.375, 0.5) {$H$};
		\node [style=none] (15) at (-0.375, 1) {};
		\node [style=none] (16) at (0.375, 0.5) {};
		\node [style=none] (17) at (0.375, 1) {};
	\end{pgfonlayer}
	\begin{pgfonlayer}{edgelayer}
		\draw [style=x-wire] (11.center) to (10.center);
		\draw [style=x-wire] (10.center)
			 to [bend left=90, looseness=1.75] (12.center)
			 to (13.center);
		\draw [style=a-wire] (15.center) to (14.center);
		\draw [style=a-wire] (14.center)
			 to [bend right=90, looseness=1.75] (16.center)
			 to (17.center);
	\end{pgfonlayer}
\end{tikzpicture}      & \begin{tikzpicture}
	\begin{pgfonlayer}{nodelayer}
		\node [style=none] (1) at (-0.375, -0.5) {};
		\node [style=none] (2) at (-0.375, -1) {};
		\node [style=none] (3) at (0.375, -0.5) {};
		\node [style=none] (4) at (0.375, -1) {};
		\node [style=none] (5) at (-0.375, 0.5) {};
		\node [style=none] (6) at (-0.375, 1) {};
		\node [style=none] (7) at (0.375, 0.5) {};
		\node [style=none] (8) at (0.375, 1) {};
	\end{pgfonlayer}
	\begin{pgfonlayer}{edgelayer}
		\draw [style=x-wire] (2.center)
			 to (1.center)
			 to [bend left=90, looseness=1.75] (3.center)
			 to (4.center);
		\draw [style=a-wire] (6.center)
			 to (5.center)
			 to [bend right=90, looseness=1.75] (7.center)
			 to (8.center);
	\end{pgfonlayer}
\end{tikzpicture}      & \begin{tikzpicture}
	\begin{pgfonlayer}{nodelayer}
		\node [style=x-frob] (1) at (-0.375, -0.5) {};
		\node [style=none] (2) at (-0.375, -1) {};
		\node [style=x-frob] (3) at (0.375, -0.5) {};
		\node [style=none] (4) at (0.375, -1) {};
		\node [style=a-frob] (5) at (-0.375, 0.5) {};
		\node [style=none] (6) at (-0.375, 1) {};
		\node [style=a-frob] (7) at (0.375, 0.5) {};
		\node [style=none] (8) at (0.375, 1) {};
	\end{pgfonlayer}
	\begin{pgfonlayer}{edgelayer}
		\draw [style=x-wire] (2.center) to (1);
		\draw [style=x-wire] (3) to (4.center);
		\draw [style=a-wire] (6.center) to (5);
		\draw [style=a-wire] (7) to (8.center);
	\end{pgfonlayer}
\end{tikzpicture}      & $-$\begin{tikzpicture}
	\begin{pgfonlayer}{nodelayer}
		\node [style=game] (1) at (-0.375, -0.5) {$G$};
		\node [style=none] (2) at (-0.375, -1) {};
		\node [style=none] (3) at (0.375, -0.5) {};
		\node [style=none] (4) at (0.375, -1) {};
		\node [style=a-frob] (5) at (-0.375, 0.5) {};
		\node [style=none] (6) at (-0.375, 1) {};
		\node [style=a-frob] (7) at (0.375, 0.5) {};
		\node [style=none] (8) at (0.375, 1) {};
	\end{pgfonlayer}
	\begin{pgfonlayer}{edgelayer}
		\draw [style=x-wire] (2.center) to (1.center);
		\draw [style=x-wire, bend left=90, looseness=1.75] (1.center) to (3.center);
		\draw [style=x-wire] (3.center) to (4.center);
		\draw [style=a-wire] (6.center) to (5);
		\draw [style=a-wire] (7) to (8.center);
	\end{pgfonlayer}
\end{tikzpicture} & $-$\begin{tikzpicture}
	\begin{pgfonlayer}{nodelayer}
		\node [style=none] (0) at (-0.375, -0.5) {};
		\node [style=none] (1) at (-0.375, -1) {};
		\node [style=none] (2) at (0.375, -0.5) {};
		\node [style=none] (3) at (0.375, -1) {};
		\node [style=a-frob] (4) at (-0.375, 0.5) {};
		\node [style=none] (5) at (-0.375, 1) {};
		\node [style=a-frob] (6) at (0.375, 0.5) {};
		\node [style=none] (7) at (0.375, 1) {};
	\end{pgfonlayer}
	\begin{pgfonlayer}{edgelayer}
		\draw [style=x-wire] (1.center) to (0.center);
		\draw [style=x-wire, bend left=90, looseness=1.75] (0.center) to (2.center);
		\draw [style=x-wire] (2.center) to (3.center);
		\draw [style=a-wire] (5.center) to (4);
		\draw [style=a-wire] (6) to (7.center);
	\end{pgfonlayer}
\end{tikzpicture}   \\
        \hline
        \begin{tikzpicture}
	\begin{pgfonlayer}{nodelayer}
		\node [style=game] (10) at (-0.375, -0.5) {$G$};
		\node [style=none] (11) at (-0.375, -1) {};
		\node [style=none] (12) at (0.375, -0.5) {};
		\node [style=none] (13) at (0.375, -1) {};
		\node [style=game] (14) at (-0.375, 0.5) {$H$};
		\node [style=none] (15) at (-0.375, 1) {};
		\node [style=none] (16) at (0.375, 0.5) {};
		\node [style=none] (17) at (0.375, 1) {};
	\end{pgfonlayer}
	\begin{pgfonlayer}{edgelayer}
		\draw [style=x-wire] (11.center) to (10.center);
		\draw [style=x-wire] (10.center)
			 to [bend left=90, looseness=1.75] (12.center)
			 to (13.center);
		\draw [style=a-wire] (15.center) to (14.center);
		\draw [style=a-wire] (14.center)
			 to [bend right=90, looseness=1.75] (16.center)
			 to (17.center);
	\end{pgfonlayer}
\end{tikzpicture}        & \begin{tikzpicture}
	\begin{pgfonlayer}{nodelayer}
		\node [style=game] (10) at (-0.375, -0.5) {$G$};
		\node [style=none] (11) at (-0.375, -1) {};
		\node [style=none] (12) at (0.375, -0.5) {};
		\node [style=none] (13) at (0.375, -1) {};
		\node [style=game] (14) at (-0.375, 0.5) {$H$};
		\node [style=none] (15) at (-0.375, 1) {};
		\node [style=none] (16) at (0.375, 0.5) {};
		\node [style=none] (17) at (0.375, 1) {};
	\end{pgfonlayer}
	\begin{pgfonlayer}{edgelayer}
		\draw [style=x-wire] (11.center) to (10.center);
		\draw [style=x-wire] (10.center)
			 to [bend left=90, looseness=1.75] (12.center)
			 to (13.center);
		\draw [style=a-wire] (15.center) to (14.center);
		\draw [style=a-wire] (14.center)
			 to [bend right=90, looseness=1.75] (16.center)
			 to (17.center);
	\end{pgfonlayer}
\end{tikzpicture}      & $0$                       &  \begin{tikzpicture}
	\begin{pgfonlayer}{nodelayer}
		\node [style=game] (10) at (-0.375, -0.5) {$G$};
		\node [style=none] (11) at (-0.375, -1) {};
		\node [style=none] (12) at (0.375, -0.5) {};
		\node [style=none] (13) at (0.375, -1) {};
		\node [style=game] (14) at (-0.375, 0.5) {$H$};
		\node [style=none] (15) at (-0.375, 1) {};
		\node [style=none] (16) at (0.375, 0.5) {};
		\node [style=none] (17) at (0.375, 1) {};
	\end{pgfonlayer}
	\begin{pgfonlayer}{edgelayer}
		\draw [style=x-wire] (11.center) to (10.center);
		\draw [style=x-wire] (10.center)
			 to [bend left=90, looseness=1.75] (12.center)
			 to (13.center);
		\draw [style=a-wire] (15.center) to (14.center);
		\draw [style=a-wire] (14.center)
			 to [bend right=90, looseness=1.75] (16.center)
			 to (17.center);
	\end{pgfonlayer}
\end{tikzpicture}         & $-$\begin{tikzpicture}
	\begin{pgfonlayer}{nodelayer}
		\node [style=game] (10) at (-0.375, -0.5) {$G$};
		\node [style=none] (11) at (-0.375, -1) {};
		\node [style=none] (12) at (0.375, -0.5) {};
		\node [style=none] (13) at (0.375, -1) {};
		\node [style=game] (14) at (-0.375, 0.5) {$H$};
		\node [style=none] (15) at (-0.375, 1) {};
		\node [style=none] (16) at (0.375, 0.5) {};
		\node [style=none] (17) at (0.375, 1) {};
	\end{pgfonlayer}
	\begin{pgfonlayer}{edgelayer}
		\draw [style=x-wire] (11.center) to (10.center);
		\draw [style=x-wire] (10.center)
			 to [bend left=90, looseness=1.75] (12.center)
			 to (13.center);
		\draw [style=a-wire] (15.center) to (14.center);
		\draw [style=a-wire] (14.center)
			 to [bend right=90, looseness=1.75] (16.center)
			 to (17.center);
	\end{pgfonlayer}
\end{tikzpicture}   & $0$                       \\
        \begin{tikzpicture}
	\begin{pgfonlayer}{nodelayer}
		\node [style=none] (1) at (-0.375, -0.5) {};
		\node [style=none] (2) at (-0.375, -1) {};
		\node [style=none] (3) at (0.375, -0.5) {};
		\node [style=none] (4) at (0.375, -1) {};
		\node [style=none] (5) at (-0.375, 0.5) {};
		\node [style=none] (6) at (-0.375, 1) {};
		\node [style=none] (7) at (0.375, 0.5) {};
		\node [style=none] (8) at (0.375, 1) {};
	\end{pgfonlayer}
	\begin{pgfonlayer}{edgelayer}
		\draw [style=x-wire] (2.center)
			 to (1.center)
			 to [bend left=90, looseness=1.75] (3.center)
			 to (4.center);
		\draw [style=a-wire] (6.center)
			 to (5.center)
			 to [bend right=90, looseness=1.75] (7.center)
			 to (8.center);
	\end{pgfonlayer}
\end{tikzpicture}        & $0$                       & \begin{tikzpicture}
	\begin{pgfonlayer}{nodelayer}
		\node [style=none] (1) at (-0.375, -0.5) {};
		\node [style=none] (2) at (-0.375, -1) {};
		\node [style=none] (3) at (0.375, -0.5) {};
		\node [style=none] (4) at (0.375, -1) {};
		\node [style=none] (5) at (-0.375, 0.5) {};
		\node [style=none] (6) at (-0.375, 1) {};
		\node [style=none] (7) at (0.375, 0.5) {};
		\node [style=none] (8) at (0.375, 1) {};
	\end{pgfonlayer}
	\begin{pgfonlayer}{edgelayer}
		\draw [style=x-wire] (2.center)
			 to (1.center)
			 to [bend left=90, looseness=1.75] (3.center)
			 to (4.center);
		\draw [style=a-wire] (6.center)
			 to (5.center)
			 to [bend right=90, looseness=1.75] (7.center)
			 to (8.center);
	\end{pgfonlayer}
\end{tikzpicture}      &  \begin{tikzpicture}
	\begin{pgfonlayer}{nodelayer}
		\node [style=none] (1) at (-0.375, -0.5) {};
		\node [style=none] (2) at (-0.375, -1) {};
		\node [style=none] (3) at (0.375, -0.5) {};
		\node [style=none] (4) at (0.375, -1) {};
		\node [style=none] (5) at (-0.375, 0.5) {};
		\node [style=none] (6) at (-0.375, 1) {};
		\node [style=none] (7) at (0.375, 0.5) {};
		\node [style=none] (8) at (0.375, 1) {};
	\end{pgfonlayer}
	\begin{pgfonlayer}{edgelayer}
		\draw [style=x-wire] (2.center)
			 to (1.center)
			 to [bend left=90, looseness=1.75] (3.center)
			 to (4.center);
		\draw [style=a-wire] (6.center)
			 to (5.center)
			 to [bend right=90, looseness=1.75] (7.center)
			 to (8.center);
	\end{pgfonlayer}
\end{tikzpicture}         & $0$                       & $-$\begin{tikzpicture}
	\begin{pgfonlayer}{nodelayer}
		\node [style=none] (1) at (-0.375, -0.5) {};
		\node [style=none] (2) at (-0.375, -1) {};
		\node [style=none] (3) at (0.375, -0.5) {};
		\node [style=none] (4) at (0.375, -1) {};
		\node [style=none] (5) at (-0.375, 0.5) {};
		\node [style=none] (6) at (-0.375, 1) {};
		\node [style=none] (7) at (0.375, 0.5) {};
		\node [style=none] (8) at (0.375, 1) {};
	\end{pgfonlayer}
	\begin{pgfonlayer}{edgelayer}
		\draw [style=x-wire] (2.center)
			 to (1.center)
			 to [bend left=90, looseness=1.75] (3.center)
			 to (4.center);
		\draw [style=a-wire] (6.center)
			 to (5.center)
			 to [bend right=90, looseness=1.75] (7.center)
			 to (8.center);
	\end{pgfonlayer}
\end{tikzpicture}     \\
        \begin{tikzpicture}
	\begin{pgfonlayer}{nodelayer}
		\node [style=x-frob] (1) at (-0.375, -0.5) {};
		\node [style=none] (2) at (-0.375, -1) {};
		\node [style=x-frob] (3) at (0.375, -0.5) {};
		\node [style=none] (4) at (0.375, -1) {};
		\node [style=a-frob] (5) at (-0.375, 0.5) {};
		\node [style=none] (6) at (-0.375, 1) {};
		\node [style=a-frob] (7) at (0.375, 0.5) {};
		\node [style=none] (8) at (0.375, 1) {};
	\end{pgfonlayer}
	\begin{pgfonlayer}{edgelayer}
		\draw [style=x-wire] (2.center) to (1);
		\draw [style=x-wire] (3) to (4.center);
		\draw [style=a-wire] (6.center) to (5);
		\draw [style=a-wire] (7) to (8.center);
	\end{pgfonlayer}
\end{tikzpicture}    & \begin{tikzpicture}
	\begin{pgfonlayer}{nodelayer}
		\node [style=game] (10) at (-0.375, -0.5) {$G$};
		\node [style=none] (11) at (-0.375, -1) {};
		\node [style=none] (12) at (0.375, -0.5) {};
		\node [style=none] (13) at (0.375, -1) {};
		\node [style=game] (14) at (-0.375, 0.5) {$H$};
		\node [style=none] (15) at (-0.375, 1) {};
		\node [style=none] (16) at (0.375, 0.5) {};
		\node [style=none] (17) at (0.375, 1) {};
	\end{pgfonlayer}
	\begin{pgfonlayer}{edgelayer}
		\draw [style=x-wire] (11.center) to (10.center);
		\draw [style=x-wire] (10.center)
			 to [bend left=90, looseness=1.75] (12.center)
			 to (13.center);
		\draw [style=a-wire] (15.center) to (14.center);
		\draw [style=a-wire] (14.center)
			 to [bend right=90, looseness=1.75] (16.center)
			 to (17.center);
	\end{pgfonlayer}
\end{tikzpicture}      & \begin{tikzpicture}
	\begin{pgfonlayer}{nodelayer}
		\node [style=none] (1) at (-0.375, -0.5) {};
		\node [style=none] (2) at (-0.375, -1) {};
		\node [style=none] (3) at (0.375, -0.5) {};
		\node [style=none] (4) at (0.375, -1) {};
		\node [style=none] (5) at (-0.375, 0.5) {};
		\node [style=none] (6) at (-0.375, 1) {};
		\node [style=none] (7) at (0.375, 0.5) {};
		\node [style=none] (8) at (0.375, 1) {};
	\end{pgfonlayer}
	\begin{pgfonlayer}{edgelayer}
		\draw [style=x-wire] (2.center)
			 to (1.center)
			 to [bend left=90, looseness=1.75] (3.center)
			 to (4.center);
		\draw [style=a-wire] (6.center)
			 to (5.center)
			 to [bend right=90, looseness=1.75] (7.center)
			 to (8.center);
	\end{pgfonlayer}
\end{tikzpicture}      &  \begin{tikzpicture}
	\begin{pgfonlayer}{nodelayer}
		\node [style=x-frob] (1) at (-0.375, -0.5) {};
		\node [style=none] (2) at (-0.375, -1) {};
		\node [style=x-frob] (3) at (0.375, -0.5) {};
		\node [style=none] (4) at (0.375, -1) {};
		\node [style=a-frob] (5) at (-0.375, 0.5) {};
		\node [style=none] (6) at (-0.375, 1) {};
		\node [style=a-frob] (7) at (0.375, 0.5) {};
		\node [style=none] (8) at (0.375, 1) {};
	\end{pgfonlayer}
	\begin{pgfonlayer}{edgelayer}
		\draw [style=x-wire] (2.center) to (1);
		\draw [style=x-wire] (3) to (4.center);
		\draw [style=a-wire] (6.center) to (5);
		\draw [style=a-wire] (7) to (8.center);
	\end{pgfonlayer}
\end{tikzpicture}     & $-$\begin{tikzpicture}
	\begin{pgfonlayer}{nodelayer}
		\node [style=game] (1) at (-0.375, -0.5) {$G$};
		\node [style=none] (2) at (-0.375, -1) {};
		\node [style=none] (3) at (0.375, -0.5) {};
		\node [style=none] (4) at (0.375, -1) {};
		\node [style=a-frob] (5) at (-0.375, 0.5) {};
		\node [style=none] (6) at (-0.375, 1) {};
		\node [style=a-frob] (7) at (0.375, 0.5) {};
		\node [style=none] (8) at (0.375, 1) {};
	\end{pgfonlayer}
	\begin{pgfonlayer}{edgelayer}
		\draw [style=x-wire] (2.center) to (1.center);
		\draw [style=x-wire, bend left=90, looseness=1.75] (1.center) to (3.center);
		\draw [style=x-wire] (3.center) to (4.center);
		\draw [style=a-wire] (6.center) to (5);
		\draw [style=a-wire] (7) to (8.center);
	\end{pgfonlayer}
\end{tikzpicture} & $-$\begin{tikzpicture}
	\begin{pgfonlayer}{nodelayer}
		\node [style=none] (0) at (-0.375, -0.5) {};
		\node [style=none] (1) at (-0.375, -1) {};
		\node [style=none] (2) at (0.375, -0.5) {};
		\node [style=none] (3) at (0.375, -1) {};
		\node [style=a-frob] (4) at (-0.375, 0.5) {};
		\node [style=none] (5) at (-0.375, 1) {};
		\node [style=a-frob] (6) at (0.375, 0.5) {};
		\node [style=none] (7) at (0.375, 1) {};
	\end{pgfonlayer}
	\begin{pgfonlayer}{edgelayer}
		\draw [style=x-wire] (1.center) to (0.center);
		\draw [style=x-wire, bend left=90, looseness=1.75] (0.center) to (2.center);
		\draw [style=x-wire] (2.center) to (3.center);
		\draw [style=a-wire] (5.center) to (4);
		\draw [style=a-wire] (6) to (7.center);
	\end{pgfonlayer}
\end{tikzpicture}   \\
        $-$\begin{tikzpicture}
	\begin{pgfonlayer}{nodelayer}
		\node [style=game] (1) at (-0.375, -0.5) {$G$};
		\node [style=none] (2) at (-0.375, -1) {};
		\node [style=none] (3) at (0.375, -0.5) {};
		\node [style=none] (4) at (0.375, -1) {};
		\node [style=a-frob] (5) at (-0.375, 0.5) {};
		\node [style=none] (6) at (-0.375, 1) {};
		\node [style=a-frob] (7) at (0.375, 0.5) {};
		\node [style=none] (8) at (0.375, 1) {};
	\end{pgfonlayer}
	\begin{pgfonlayer}{edgelayer}
		\draw [style=x-wire] (2.center) to (1.center);
		\draw [style=x-wire, bend left=90, looseness=1.75] (1.center) to (3.center);
		\draw [style=x-wire] (3.center) to (4.center);
		\draw [style=a-wire] (6.center) to (5);
		\draw [style=a-wire] (7) to (8.center);
	\end{pgfonlayer}
\end{tikzpicture}   & $-$\begin{tikzpicture}
	\begin{pgfonlayer}{nodelayer}
		\node [style=game] (10) at (-0.375, -0.5) {$G$};
		\node [style=none] (11) at (-0.375, -1) {};
		\node [style=none] (12) at (0.375, -0.5) {};
		\node [style=none] (13) at (0.375, -1) {};
		\node [style=game] (14) at (-0.375, 0.5) {$H$};
		\node [style=none] (15) at (-0.375, 1) {};
		\node [style=none] (16) at (0.375, 0.5) {};
		\node [style=none] (17) at (0.375, 1) {};
	\end{pgfonlayer}
	\begin{pgfonlayer}{edgelayer}
		\draw [style=x-wire] (11.center) to (10.center);
		\draw [style=x-wire] (10.center)
			 to [bend left=90, looseness=1.75] (12.center)
			 to (13.center);
		\draw [style=a-wire] (15.center) to (14.center);
		\draw [style=a-wire] (14.center)
			 to [bend right=90, looseness=1.75] (16.center)
			 to (17.center);
	\end{pgfonlayer}
\end{tikzpicture}   & $0$                       &  $-$\begin{tikzpicture}
	\begin{pgfonlayer}{nodelayer}
		\node [style=game] (1) at (-0.375, -0.5) {$G$};
		\node [style=none] (2) at (-0.375, -1) {};
		\node [style=none] (3) at (0.375, -0.5) {};
		\node [style=none] (4) at (0.375, -1) {};
		\node [style=a-frob] (5) at (-0.375, 0.5) {};
		\node [style=none] (6) at (-0.375, 1) {};
		\node [style=a-frob] (7) at (0.375, 0.5) {};
		\node [style=none] (8) at (0.375, 1) {};
	\end{pgfonlayer}
	\begin{pgfonlayer}{edgelayer}
		\draw [style=x-wire] (2.center) to (1.center);
		\draw [style=x-wire, bend left=90, looseness=1.75] (1.center) to (3.center);
		\draw [style=x-wire] (3.center) to (4.center);
		\draw [style=a-wire] (6.center) to (5);
		\draw [style=a-wire] (7) to (8.center);
	\end{pgfonlayer}
\end{tikzpicture}    & \begin{tikzpicture}
	\begin{pgfonlayer}{nodelayer}
		\node [style=game] (1) at (-0.375, -0.5) {$G$};
		\node [style=none] (2) at (-0.375, -1) {};
		\node [style=none] (3) at (0.375, -0.5) {};
		\node [style=none] (4) at (0.375, -1) {};
		\node [style=a-frob] (5) at (-0.375, 0.5) {};
		\node [style=none] (6) at (-0.375, 1) {};
		\node [style=a-frob] (7) at (0.375, 0.5) {};
		\node [style=none] (8) at (0.375, 1) {};
	\end{pgfonlayer}
	\begin{pgfonlayer}{edgelayer}
		\draw [style=x-wire] (2.center) to (1.center);
		\draw [style=x-wire, bend left=90, looseness=1.75] (1.center) to (3.center);
		\draw [style=x-wire] (3.center) to (4.center);
		\draw [style=a-wire] (6.center) to (5);
		\draw [style=a-wire] (7) to (8.center);
	\end{pgfonlayer}
\end{tikzpicture}    & $0$                       \\
        $-$\begin{tikzpicture}
	\begin{pgfonlayer}{nodelayer}
		\node [style=none] (0) at (-0.375, -0.5) {};
		\node [style=none] (1) at (-0.375, -1) {};
		\node [style=none] (2) at (0.375, -0.5) {};
		\node [style=none] (3) at (0.375, -1) {};
		\node [style=a-frob] (4) at (-0.375, 0.5) {};
		\node [style=none] (5) at (-0.375, 1) {};
		\node [style=a-frob] (6) at (0.375, 0.5) {};
		\node [style=none] (7) at (0.375, 1) {};
	\end{pgfonlayer}
	\begin{pgfonlayer}{edgelayer}
		\draw [style=x-wire] (1.center) to (0.center);
		\draw [style=x-wire, bend left=90, looseness=1.75] (0.center) to (2.center);
		\draw [style=x-wire] (2.center) to (3.center);
		\draw [style=a-wire] (5.center) to (4);
		\draw [style=a-wire] (6) to (7.center);
	\end{pgfonlayer}
\end{tikzpicture}   & $0$                       & $-$\begin{tikzpicture}
	\begin{pgfonlayer}{nodelayer}
		\node [style=none] (1) at (-0.375, -0.5) {};
		\node [style=none] (2) at (-0.375, -1) {};
		\node [style=none] (3) at (0.375, -0.5) {};
		\node [style=none] (4) at (0.375, -1) {};
		\node [style=none] (5) at (-0.375, 0.5) {};
		\node [style=none] (6) at (-0.375, 1) {};
		\node [style=none] (7) at (0.375, 0.5) {};
		\node [style=none] (8) at (0.375, 1) {};
	\end{pgfonlayer}
	\begin{pgfonlayer}{edgelayer}
		\draw [style=x-wire] (2.center)
			 to (1.center)
			 to [bend left=90, looseness=1.75] (3.center)
			 to (4.center);
		\draw [style=a-wire] (6.center)
			 to (5.center)
			 to [bend right=90, looseness=1.75] (7.center)
			 to (8.center);
	\end{pgfonlayer}
\end{tikzpicture}   & $-$\begin{tikzpicture}
	\begin{pgfonlayer}{nodelayer}
		\node [style=none] (0) at (-0.375, -0.5) {};
		\node [style=none] (1) at (-0.375, -1) {};
		\node [style=none] (2) at (0.375, -0.5) {};
		\node [style=none] (3) at (0.375, -1) {};
		\node [style=a-frob] (4) at (-0.375, 0.5) {};
		\node [style=none] (5) at (-0.375, 1) {};
		\node [style=a-frob] (6) at (0.375, 0.5) {};
		\node [style=none] (7) at (0.375, 1) {};
	\end{pgfonlayer}
	\begin{pgfonlayer}{edgelayer}
		\draw [style=x-wire] (1.center) to (0.center);
		\draw [style=x-wire, bend left=90, looseness=1.75] (0.center) to (2.center);
		\draw [style=x-wire] (2.center) to (3.center);
		\draw [style=a-wire] (5.center) to (4);
		\draw [style=a-wire] (6) to (7.center);
	\end{pgfonlayer}
\end{tikzpicture}     & $0$                       & \begin{tikzpicture}
	\begin{pgfonlayer}{nodelayer}
		\node [style=none] (0) at (-0.375, -0.5) {};
		\node [style=none] (1) at (-0.375, -1) {};
		\node [style=none] (2) at (0.375, -0.5) {};
		\node [style=none] (3) at (0.375, -1) {};
		\node [style=a-frob] (4) at (-0.375, 0.5) {};
		\node [style=none] (5) at (-0.375, 1) {};
		\node [style=a-frob] (6) at (0.375, 0.5) {};
		\node [style=none] (7) at (0.375, 1) {};
	\end{pgfonlayer}
	\begin{pgfonlayer}{edgelayer}
		\draw [style=x-wire] (1.center) to (0.center);
		\draw [style=x-wire, bend left=90, looseness=1.75] (0.center) to (2.center);
		\draw [style=x-wire] (2.center) to (3.center);
		\draw [style=a-wire] (5.center) to (4);
		\draw [style=a-wire] (6) to (7.center);
	\end{pgfonlayer}
\end{tikzpicture}     \\
    \end{tabular}
    \end{center}
\end{table}

\subsection{Shared operators give synchronous strategies}
\label{pf:E_copied_sync}

Let $E:X\to_{\cl H} A$ be a quantum function. Then the quantum strategy $E \otimes E_*$ realizes a synchronous correlation along with the normalized cup state of $\cl H \otimes \cl H^*$.

\begin{proof}
    We first show that the linear map given by composing $E$ with itself on the $\cl H$ wire is synchronous.
    We proceed diagramatically.
    \ctikzfig{E\_copied\_sync}

    Tracing out the Hilbert space wire on the above equation and then sliding the right copy of $E$ around the cap to get $E_*$ gives the (unnormalized) synchronicity condition for the correlation realized by $E \otimes E_*$ with the cup state.
\end{proof}

\subsection{The graph isomorphism game is well-defined}
\label{pf:graph_iso_defn}

We just need to check that $\lambda_{G\to H}$ and $\lambda_{H\to G}^\dagger$ commute. This follows from the fact that Table \ref{tbl:hom_comm_prod} is symmetric, as these are the building blocks of $\lambda_{G\to H}$ and $\lambda_{H\to G}^\dagger$. The commuting product $\lambda_{G\to H} \star \lambda_{H\to G}^\dagger$ can be computed in either order to give
\ctikzfig{qgraph\_iso\_game\_thick}

\subsection{Graph homomorphisms and associated games} \label{fig:graph_homs}
    We present the relationships between the quantum and classical graph homomorphism games, and quantum and classical homomorphisms of quantum and classical graphs. Our quantum graph homomorphism game adds a previously missing vertex to this diagram.
    \begin{figure}[h!]
    \centering
    \[\begin{tikzcd}
	&& {\textrm{hom game}} \\
	\\
	{\textrm{homs of graphs}} &&&& {\textrm{q-homs of graphs}} \\
	\\
	\\
	{\textrm{homs of q-graphs}} &&&& {\textrm{q-homs of q-graphs}} \\
	\\
	&& {\textbf{q-hom game}}
	\arrow["{\textrm{perfect det strategies}}"', curve={height=18pt}, dotted, from=1-3, to=3-1]
	\arrow["{\textrm{perfect q-correlations}}", curve={height=-18pt}, dotted, from=1-3, to=3-5]
	\arrow["{\textrm{restrict to classical resource}}"', curve={height=-12pt}, dotted, from=3-5, to=3-1]
	\arrow["{\textrm{restrict to classical graphs}}"{description}, curve={height=-12pt}, dotted, from=6-1, to=3-1]
	\arrow["{\textrm{restrict to classical graphs}}"{description}, curve={height=12pt}, dotted, from=6-5, to=3-5]
	\arrow["{\textrm{restrict to classical resource}}", curve={height=12pt}, dotted, from=6-5, to=6-1]
	\arrow["{\textrm{restrict to classical graphs}}"{description}, from=8-3, to=1-3]
	\arrow["{\textrm{perfect det strategies}}", curve={height=-18pt}, from=8-3, to=6-1]
	\arrow["{\textrm{(all?) perfect q-correlations}}"', curve={height=18pt}, from=8-3, to=6-5]
    \end{tikzcd}\]
    \end{figure}

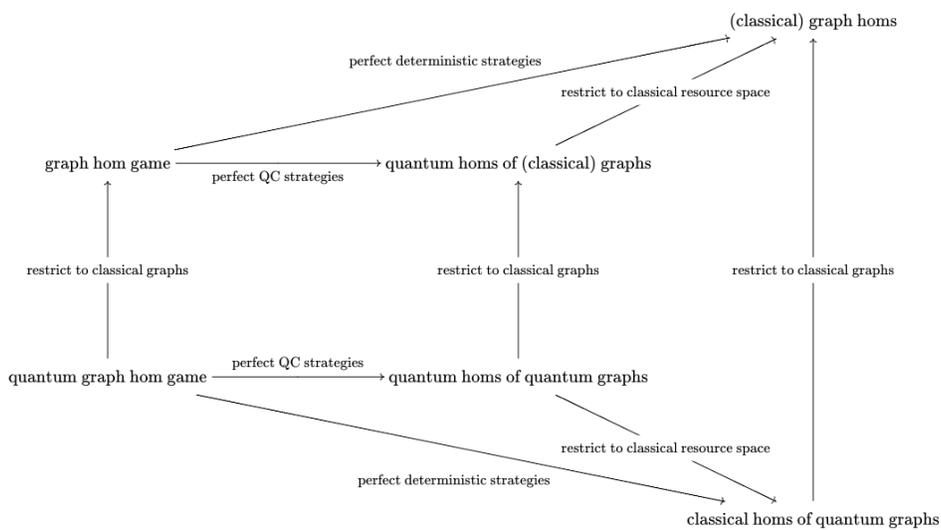
\end{document}